\documentclass{article}
\usepackage[utf8]{inputenc}
\usepackage{amsmath}
\usepackage{hyperref}
\usepackage{todonotes}
\usepackage{amsthm}
\usepackage{comment}
\usepackage{amsfonts}
\usepackage{dsfont}
\usepackage{amssymb}
\usepackage{geometry}
\usepackage{esint}
\usepackage{bbm}
\usepackage{xcolor}
\usepackage{cancel}
\usepackage{appendix}

\definecolor{linkcolor}{rgb}{0,0,0.6}
\definecolor{ccitecolor}{rgb}{0.75, 0, 0} 
\hypersetup{
	colorlinks=true, 
	pdfstartview=FitV, 
	urlcolor=linkcolor, 
	linkcolor= linkcolor, 
	citecolor=ccitecolor 
}

\newtheorem{theorem}{Theorem}[section]
\newtheorem{lemma}[theorem]{Lemma}

\newtheorem{corollary}[theorem]{Corollary}
\newtheorem{proposition}[theorem]{Proposition}
\newtheorem{properties}[theorem]{Properties}
\newtheorem{property}[theorem]{Property}
\newtheorem{definition}[theorem]{Definition}

\theoremstyle{definition}
\newtheorem{notation}[theorem]{Notation}

\theoremstyle{definition}
\newtheorem{assumption}[theorem]{Assumption}

\theoremstyle{remark}
\newtheorem{remark}[theorem]{Remark}

\numberwithin{equation}{section}

\DeclareMathOperator{\R}{\mathbb{R}}

\DeclareMathOperator{\N}{\mathbb{N}}

\DeclareMathOperator{\Hcal}{\mathcal{H}}
\DeclareMathOperator{\tr}{Tr}

\DeclareMathOperator{\supp}{supp}

\renewcommand{\d}{ \, \mathrm{d}}

\title{Semi-classical limit of an attractive Fermi gas in one or two dimensions}
\author{Thomas Gamet\thanks{Ecole Normale Supérieure de Lyon, UMPA (UMR 5669)}~\thanks{thomas.gamet@ens-lyon.fr}}

\begin{document}
	\maketitle
	
	\begin{abstract}
		We study the ground-state of a Fermi gas with short range attrative interactions in one or two dimensions. $N$ fermions are placed in a confining potential, and interact with each other through a negative potential, whose range is larger than the typical distance between particles. We show the convergence of the ground-state energy of the Hamiltonian to a Thomas-Fermi energy in the large $N$ limit. Furthermore, we prove convergence of the ground states, in the sense of their Husimi functions. These results extend to the case of a repulsive interaction of positive Fourier transform.
	\end{abstract}
	
	\tableofcontents 
	
	\section{Introduction}
	\subsection{Context}
	Large and interacting quantum systems are very difficult to describe directly, whence the interest of deriving effective models. Many such models replace the $N$-body wave function by the density of particles (density functional theory). For fermions, the thermodynamic limit of a system of spin $1/2$ particles in three dimensions has been widely studied from a mathematical point of view in the last decade. For a total density of particles $\varrho$, the volumic energy is given by the Huang-Yang formula \cite{huang_quantum-mechanical_1957} in the low density limit $\varrho \to 0$
	\begin{equation}\label{equation_huang_yang}
		e(\varrho) = \frac{3}{5}(3\pi^2)^{2/3}\varrho^{5/3} + 2\pi a \varrho^2(1 + c\varrho^{1/3}) + o(\varrho^{7/3}),
	\end{equation}
	where $a$ is the scattering length of the interaction potential, and $c$ an explicit constant. The first order corresponds to the kinetic energy of the system without interactions. The first two terms of \eqref{equation_huang_yang} -- that is up to $\varrho^2$ -- have been first derived in \cite{lieb_ground-state_2005}. \cite{falconi_dilute_2021} provided an alternative proof of the $\varrho^2$ contribution using Bogoliubov-type rotations, a method originally developed for Bose systems; for improvements of the error bounds see \cite{giacomelli_optimal_2023, giacomelli_optimal_2024, lauritsen_almost_2025}. Notably, the authors of \cite{falconi_dilute_2021} were the first to implement this Bogoliubov approach directly in the thermodynamic limit. More recently, the authors of the seminal papers \cite{giacomelli_huang-yang_2024, giacomelli_huang-yang_2025} went a step further and proved the Huang–Yang formula -- i.e. the expansion up to order $\varrho^{7/3}$. The one dimensional case has also been studied in \cite{agerskov_ground_2025}. For spin-polarized fermions, the first order is the same (up to a constant) and the first interaction term is of order $\varrho^{8/3}$ as proved in \cite{lauritsen_ground_2024, lauritsen_ground_2024_2}, hence the interaction terms in \eqref{equation_huang_yang} must only come from interactions between particles of opposite spin.~\\
	
	Other works study trapped gases in the large number of particles limit rather than a thermodynamic limit. Let 
	\begin{equation}
		H_N = \sum_{j=1}^N (-\hbar^2\Delta_j) + \sum_{j=1}^N V(x_j) + \sum_{1\leq j < k\leq N} v_N(x_j - x_k)
	\end{equation}
	with $\hbar = N^{-1/d}$. The scaling of $\hbar$ comes from the growth of the sum of the $N$ first eigenvalues of the Laplacian on an hypercube, which is proportional to $N^{5/3}$ and gives a kinetic energy $\propto N$. The Hamiltonian $H_N$ is composed of three parts: the kinetic, potential, and interaction energies. The Gross-Pitaevskii limit, i.e. a scaling 
	\begin{equation} 
		v_N = N^2v(N\cdot)\geq 0
	\end{equation}
	in three dimension, with $V=0$ and in a box $[0, L]^3$ has been studied in \cite{chen_second_2024, chen_second_2025}. In a previous work \cite{gamet_ground_2025}, we have studied the ground-state energy of $H_N$ in the scaling 
	\begin{equation} 
		v_N = N^{2\beta -2/3}v(N^{\beta}\cdot)
	\end{equation}
	with an inhomogeneous potential $V$ and $1/3< \beta < 34/81$. For spin-polarized fermions, many papers (from \cite{lieb_hartree-fock_1977, lieb_thomas-fermi_1977} in the case of Coulomb interactions) have focused on the $N\to +\infty$ limit in the mean-field scaling, with a general potential $V$, that is
	\begin{equation}
		v_N = N^{-1}v
	\end{equation}
	to ensure that the interaction energy remains of order $N$. In this scaling, each particle interacts in the limit with a huge number of particles, and the range of the interaction remains macroscopically of order 1. In \cite{fournais_semi-classical_2018}, the authors show that the ground-state energy of such a system converges to the Thomas-Fermi energy
	\begin{equation}
		\inf\left\{c_{\mathrm{TF}}\int \rho^{1+2/d} +\int V\rho + \frac{1}{2}\int (v * \rho)\rho,~\rho\geq 0,~\int \rho = 1\right\}
	\end{equation}
	In following papers \cite{cardenas_norm_2024, cardenas_quantitative_2025}, the norm convergence and control (by $\hbar$) of the rate of convergence of the Husimi and Wigner functions of the minimizers have then been proved, under some additional assumptions. These quantitative estimates are motivated by the limit $N\to +\infty$ for the solutions of an evolution equation involving $H_N$, for which quantitative estimates in $\hbar$ are needed \cite{benedikter_meanfield_2014, benedikter_effective_2015, benedikter_hartree_2016, fournais_optimal_2020}. Several papers have extended the results of \cite{fournais_semi-classical_2018} to different settings. For instance, in \cite{fournais_semi-classical_2020}, the authors established the convergence of the energy and the states when a strong external magnetic potential was added, that is when $-\hbar^2 \Delta_j$ is replaced by the magnetic Laplacian
	\begin{equation}
		\big(-i\hbar\nabla_j + b_N A(x_j)\big)^2
	\end{equation}
	with $b_N N^{-1/3}\to \beta \in [0, +\infty]$. See also \cite{lieb_heavy_1992, lieb_asymptotics_1994, lieb_asymptotics_1994-1, lieb_ground_1995, perice_multiple_2024} for other results on the high magnetic field limit. Almost fermionic anyons have also been studied in a similar fashion in \cite{girardot_semiclassical_2021, girardot_lieb-thirring_2023}. These can be assimilated to fermions under a density dependent magnetic field. The corrections to the ground-state energy in the mean-field limit, on the torus, and for repulsive interactions (in the sense that $\widehat{v} \geq 0$) has been studied in a series of works \cite{hainzl_correlation_2020, benedikter_optimal_2020, benedikter_correlation_2021, benedikter_correlation_2023, christiansen_random_2023, christiansen_gell-mannbrueckner_2023, christiansen_correlation_2024, casadei_ground_2026}. Other scalings for  $v_N$ have also been considered. In \cite{fournais_ground_2024}, strong interactions are considered with
	\begin{equation}
		v_N = N^{-\alpha}v,~~~~~0<\alpha<1.
	\end{equation}
	In \cite{lewin_semi-classical_2019}, the result of \cite{fournais_semi-classical_2018} have been generalized to a positive temperature setting and for several scalings. In particular, for repulsive interactions of the form
	\begin{equation}
		v_N = N^{d\beta -1}v(N^{\beta}\cdot),
	\end{equation}
	when $\beta<1/d$, contrary to the mean-field regime, in the limit, the interactions are local, and for $\beta >1/d$, there are no interactions in the limit because of the Pauli principle. We want to study a similar limit with $d\beta < 1$ in order to recover interactions in the limit. However, we take attractive interactions\footnote{With a generalization in Sections \ref{subsection_non_attractive} and \ref{section5}.}, that is
	\begin{equation}
		v \leq 0.
	\end{equation}
	From a mathematical point of view, the interest of the problem comes mainly of the inability to use -- at least directly -- some standard methods for repulsive potentials (Onsager lemma, throwing positive terms away). However, this attractive potential is not just a mathematical curiosity, as attractive interactions have been constructed experimentally for a 1D Fermi gas using Feshbach resonance \cite{kafle_low_2025}. In this experiment, the fermions have a spin $1/2$, but in our model the spin only changes the constants in the Thomas-Fermi energy functional, and not the qualitative results.
	
	\subsection{Convergence of the energy}
	First, let us introduce the Hamiltonian we will study.
	\begin{definition}[The Hamiltonian and its domain]\label{definition_hamiltonien}
		Let 
		\begin{equation}
			\mathcal{H}_N = \left\{\Psi_N \in H^1(\R^{dN})\cap L^2_{\mathrm{as}}((\R^d)^N),~\forall j\leq N,~\int_{\R^{dN}} V(x_j)|\Psi_N(x_1,...x_N)|^2\d x_1...\d x_N< +\infty\right\}
		\end{equation}
		be the Hilbert space on which we define the Hamiltonian
		\begin{equation}
			H_N = \sum_{j=1}^N (-\hbar^2\Delta_j) + \sum_{j=1}^NV(x_j) - N^{-1}\sum_{j<k} w_N(x_j - x_k)
		\end{equation}
		with a potential satisfying
		\begin{equation}\label{equation_premiere_condition_V}
			V\in L^{\infty}_{\mathrm{loc}}(\R^d),~~~~~V\geq 0,~~~~~ V(x)\to +\infty ~\mathrm{when}~ |x|\to +\infty
		\end{equation}
		and with a short range scaling
		\begin{equation}\label{equation_premiere_condition_w}
			w_N = N^{d\beta}w(N^{\beta}\cdot),~~~~~0<\beta<1/d,~~~~~w \in L^{\infty}(\R^d)\cap L^1(\R^d),~~~~~w\geq 0.
		\end{equation}
		As explained briefly in the previous section, we choose the semi-classical parameter $\hbar$ to be
		\begin{equation}
			\hbar = N^{-1/d}.
		\end{equation}
		The ground-state energy $E(N)$ is defined by
		\begin{equation}\label{equation_definition_e(n)}
			E(N) = \inf_{\Psi_N\in \mathcal{H}_N}\frac{\langle \Psi_N| H_N \Psi_N\rangle}{\langle \Psi_N|\Psi_N\rangle}.
		\end{equation}
	\end{definition}
	
	In the large number of particles limit, one expects the ground-state energy to be approximately given by the minimum of a semi-classical functional. 
	
	\begin{definition}[Vlasov and Thomas-Fermi energies]\label{definition_energies_vlasov_thomas_fermi}
		We define the Vlasov energy of a positive integrable function $m$ on the phase space $\R^d \times \R^d$ by
		\begin{equation}\label{equation_definition_vlasov_energy_functional}
			\mathcal{E}^{\mathrm{V}}[m] = \frac{1}{(2\pi)^d}\iint_{\R^d \times \R^d} \big(|p|^2 + V(x)\big) m(x, p)\d x\d p - I_w \int_{\R^d} \rho_m(x)^2\d x 
		\end{equation}
		with
		\begin{equation}
			I_w = \frac{1}{2}\int w
		\end{equation}
		and \begin{equation}\label{equation_definition_rhom}
			\rho_m(x) = \frac{1}{(2\pi)^d}\int_{\R^d} m(x, p)\d p.
		\end{equation}
		We define the Thomas-Fermi functional by
		\begin{equation}
			\mathcal{E}^{\mathrm{TF}}[\rho] = c_{\mathrm{TF}}\int_{\R^d} \rho^{1+2/d} + \int_{\R^d} V\rho - I_w\int_{\R^d} \rho^2
		\end{equation}
		with 
		\begin{equation}\label{equation_definition_constante_thomas_fermi}
			c_{\mathrm{TF}} = \begin{cases}
				\pi^2 & \mathrm{if}~~ d = 1 \\
				8\pi &\mathrm{if}~~ d=2,
			\end{cases}
		\end{equation}
		and the Thomas-Fermi energy by
		\begin{equation}\label{equation_definition_energie_thomas_fermi}
			E^{\mathrm{TF}} = \inf \left\{\mathcal{E}^{\mathrm{TF}}[\rho],~\rho\geq 0,~\int_{\R^d} \rho =1\right\}.
		\end{equation}
	\end{definition}
	\begin{remark}[Constraints on the dimension]\label{remark_constraint_dimension}
		When $d\geq 3$, we have 
		\begin{equation}
			E^{\mathrm{TF}} = - \infty.
		\end{equation}
		On the contrary, we have
		\begin{equation}
			E^{\mathrm{TF}} > -\infty
		\end{equation}
		for $d = 1$, and for $d = 2$ as long as 
		\begin{equation}\label{equation_lien_c_I}
			c_{\mathrm{TF}} \geq I_w.
		\end{equation}
		A proof is given in Appendix \ref{subsection_appendice_bonus}.
		\hfill $\diamond$
	\end{remark}
	
	\begin{remark}[Link between Vlasov and Thomas-Fermi energies]\label{remark_link_vlasov_thomas_fermi}
		By the bath-tube principle \cite[Theorem 1.14]{lieb_analysis_1997}, clearly, we have
		\begin{equation}\label{equation_egalite_energie_thomas_fermi_vlasov}
			E^{\mathrm{TF}} = \inf \left\{\mathcal{E}^{\mathrm{V}}[m],~ 0\leq m \leq 1,~\iint_{\R^d\times \R^d} m = (2\pi)^d\right\},
		\end{equation}
		with minimizers of the form
		\begin{equation}
			m_{\rho}(x, p) = \mathds{1}\Big(|p| \leq c_d \rho(x)^{1/d}\Big)
		\end{equation}
		with
		\begin{equation}
			c_d =
			\begin{cases}
				\pi & \mathrm{if}~~ d=1\\
				\sqrt{4\pi} & \mathrm{if}~~ d=2.
			\end{cases}
		\end{equation}
		\hfill $\diamond$
	\end{remark}
	Let us now introduce the assumptions we will require in Theorems \ref{theorem1} and \ref{theorem2}.
	\begin{assumption}[The trapping potential $V$]\label{assumption_V} For a given $s>1$
		\begin{equation}\label{hypothese_V1}
			\begin{cases}
				V, \nabla V \in L^{\infty}_{\mathrm{loc}}\\
				V(x) \geq C |x|^s - c\\
				\nabla V(x) \leq C \big(|x|^{s-1} + 1\big).
			\end{cases}
		\end{equation} 
	\end{assumption}
	\begin{assumption}\label{assumption_V2}
		When $d=1$, we request that the level sets of $V$ have zero Lebesgue measure.\footnote{We use this to prove the existence of minimizers of the Thomas-Fermi energy, see Appendix \ref{section_appendixe_1}.}
	\end{assumption}
	\begin{assumption}[The interaction potential $w$]\label{assumption_w}
		\begin{equation}
			\begin{cases}
				w\geq 0\\
				w \in L^1(\R^d)\cap L^{\infty}(\R^d)\\
				\nabla w \in L^{1}(\R^d)\cap L^{\infty}(\R^d),
			\end{cases}
		\end{equation}
		and in the case $d=2$:
		\begin{equation}\label{hypothese_w}
			I_w = \frac{1}{2}\int_{\R^d} w < c_{\mathrm{TF}}.
		\end{equation}
	\end{assumption}
	Now, we can state our first theorem.
	\begin{theorem}[Ground-state energy]\label{theorem1}
		Let $d=1$ or $2$, and $0 < \beta < \frac{1}{d(d + 1)}$. If $V$ and $w$ satisfy Assumptions \ref{assumption_V} and \ref{assumption_w}, we have
		\begin{equation}
			E(N) = N E^{\mathrm{TF}} + o(N)
		\end{equation}
		with the notation introduced in Definitions \ref{definition_hamiltonien} and \ref{definition_energies_vlasov_thomas_fermi}.
	\end{theorem}
	\begin{remark}[Constraint on $\beta$]
		We cannot hope to have a similar result for $\beta > 1/d$, since in this case, the interactions between particles are very rare due to the range being shorter than the typical distance between particles, and the Pauli principle; hence we have to impose $\beta < 1/d$. In Theorem \ref{theorem1}, we require a stronger constraint, namely
		\begin{equation}\label{equation_constraint_beta_rk1}
			\beta < \frac{1}{d(d +1)}.
		\end{equation}
		Note that the upper bound remains true for all $\beta < 1/d$. For the semi-classical approximation of the energy used in the lower bound, we need to impose a stricter assumption on $\beta$. Because of the negative part of the Hamiltonian, the a priori estimates of Lemma \ref{lemma_apriori} do not allow us to cover the case of every $\beta < 1/d$. Note however that our constraint 
		\begin{equation}
			\beta < \frac{2}{d(2d+1)}
		\end{equation}
		is better than the naive $\beta < \frac{1}{d(d+1)}$ that we would obtain, should we bound the left-hand side of \eqref{equation_borne_apriori2} using only the $L^{\infty}$ norm of $w_N$. Nonetheless, later on, namely in Lemma \ref{lemma_averaging}, we have to impose further restrictions on $\beta$, whence the assumption \eqref{equation_constraint_beta_rk1}.
		~\hfill $\diamond$
	\end{remark}
	\begin{remark}[Assumptions on $w$ and $V$] 
		In the proof, we will specify which lemmas require Assumption \ref{assumption_V} and/or \ref{assumption_w}, so that the reader may track where it is needed. We will assume that $V$ and $w$ always satisfy \eqref{equation_premiere_condition_V} and \eqref{equation_premiere_condition_w}.
		\hfill $\diamond$
	\end{remark}

	\subsection{Convergence of Husimi functions}
	In order to state our second theorem, we introduce a few definitions. A reader already familiar with this notation and these usual properties might want to go directly to Theorem \ref{theorem2}.
	\begin{definition}[Space and momentum density]\label{definition_space_momentum_density}
		For 
		\begin{equation}
			\gamma = \sum \alpha_n|u_n\rangle\langle u_n| 
		\end{equation}
		a positive trace class operator, we define the associated space and momentum densities by
		\begin{equation}\label{equation_definition_tgamma}
			\begin{cases}
				\rho_{\gamma} = \sum \alpha_n |u_n|^2\\
				t_{\gamma} = \sum \alpha_n \big|\mathcal{F}_{\hbar}[u_n]\big|^2,
			\end{cases}
		\end{equation}
		with $\mathcal{F}_{\hbar}[g]$ the semi-classical Fourier transform of $g$
		\begin{equation}
			\mathcal{F}_{\hbar}[g](p) = (2\pi \hbar_x)^{-d/2} \int_{\R^d} g(p) e^{- ip.x/\hbar}\d x.
		\end{equation}
	\end{definition}
	Note in particular that for a given $\gamma$ as in the previous definition, we have
	\begin{equation}
		\int_{\R^d} \rho_{\gamma}(x)\d x= \int_{\R^d} t_{\gamma}(p) \d p = \tr(\gamma).
	\end{equation}
	\begin{definition}[Reduced density matrices]
		For $\Psi_N \in \Hcal_N$ and $k\leq N$, we define the $k$-particles reduced density matrix of $\Psi_N$ by
		\begin{equation}\label{equation_definition_matrice_densite_reduite}
			\gamma_{\Psi_N}^{(k)} = \frac{N !}{(N - k) !}\tr_{k+1\to N}\Big(|\Psi_N\rangle\langle\Psi_N|\Big).
		\end{equation}
		Moreover, we define the $k$-particles density of $\Psi_N$ by
		\begin{equation}
			\rho_{\Psi_N}^{(k)}(x_1,...,x_k) = \binom{N}{k}\int_{\R^{d(N-k)}} \big|\Psi_N(x_1, ... x_N)\big|^2 \d x_{k+1}...\d x_N
		\end{equation}
	\end{definition}
	Note that 
	\begin{equation}
		\rho_{\Psi_N}^{(1)} = \rho_{\gamma_{\Psi_N}^{(1)}}.
	\end{equation}
	\begin{definition}[Coherent state]\label{definition_coherent_state}
		Let $\hbar_x$ and $\hbar_p$ such that $\hbar_x\hbar_p = \hbar^2$ and $\hbar_x, \hbar_p\ll 1$. We fix $f\in C_c^{\infty}(\R^d)$ a real positive and radial function with a $L^2$ norm of 1. For $x, p \in \R^d$, we set 
		\begin{equation}\label{equation_definition_coherent_state_space}
			\begin{cases}
				f^{\hbar}(x) = \hbar_x^{-d/4}f(x/\sqrt{\hbar_x})\\
				g^{\hbar}(p) = \hbar_p^{-d/4}\widehat{f}(p/\sqrt{\hbar_p}),
			\end{cases}
		\end{equation}
		with $\widehat{f}$ the Fourier transform
		\begin{equation}
			\widehat{f}(p) = (2\pi)^{-d/2}\int_{\R^d} f(p)e^{-ip.x}\d x.
		\end{equation}
		For $x, p\in \R^d$, we denote by $f_{x, p}^{\hbar}$ the squeezed coherent state, defined by
		\begin{equation}
			\forall y\in \R^d, ~ f_{x, p}^{\hbar}(y) = \hbar_x^{-d/4}f\left(\frac{y-x}{\sqrt{\hbar_x}}\right)e^{ip.y/\hbar}.
		\end{equation}
		It is localized on a scale $\sqrt{\hbar_x}$ in space, and $\sqrt{\hbar_p}$ in momentum, around the point $(x, p)$ of the phase space. 
	\end{definition}
	\begin{property}[Resolution of the identity]\label{property1}
		We have
		\begin{equation}
			\iint_{\R^d\times\R^d} |f_{x, p}^{\hbar}\rangle\langle f_{x, p}^{\hbar}|\d x\d p = (2\pi\hbar)^d = \frac{(2\pi)^d}{N}.
		\end{equation}
	\end{property}
	\begin{definition}[Husimi functions]\label{definition_husimi}
		For $k\leq N$ and $\Psi_N \in \Hcal_N$, we define the $k$-particles Husimi function of $\Psi_N$ by
		\begin{equation}
			m_{\Psi_N}^{(k)}(x_1,...x_k;p_1,...,p_k) = \Big\langle f_{x_1, p_1}^{\hbar}\otimes...\otimes f_{x_k, p_k}^{\hbar}\Big|\gamma_{\Psi_N}^{(k)} f_{x_1, p_1}^{\hbar}\otimes...\otimes f_{x_k, p_k}^{\hbar}\Big\rangle.
		\end{equation}
	\end{definition}
	\begin{properties}[Husimi functions]\label{property3}~\\
		For $\Psi_N \in \mathcal{H}_N$ and $k\leq N$, we have
		\begin{equation}
			\frac{1}{(2\pi)^{dk}}\iint_{\R^{dk}\times \R^{dk}} m_{\Psi_N}^{(k)}(x_1,...x_k;p_1,...p_k)\d x_1...\d x_k \d p_1... \d p_k = \tr\big(\gamma_{\Psi_N}^{(k)}\big) = \frac{N !}{(N-k) !}.
		\end{equation}
		Moreover,
		\begin{equation}\label{equation_lien_m_rho}
			\frac{N^k}{(2\pi)^{dk}}\int_{\R^{dk}} m_{\Psi_N}^{(k)}(\cdot;p_1,... p_k)\d p_1...\d p_k = k! \rho_{\Psi_N}^{(k)}*\big(|f^{\hbar}|^2\big)^{\otimes k}
		\end{equation}
		and
		\begin{equation}\label{equation_lien_m_t}
			\frac{N}{(2\pi)^d}\int_{\R^d} m_{\Psi_N}^{(1)}(x, \cdot)\d x = t_{\gamma^{(1)}_{\Psi_N}}*|g^{\hbar}|^2.
		\end{equation}
	\end{properties}
	Properties \ref{property1} and \ref{property3} are standard results, one can find proofs in Section 2.1 of \cite{fournais_semi-classical_2018}. Now, we can state our second theorem.
	\begin{theorem}[Convergence of ground states]\label{theorem2}
		Let $d = 1$ or $2$, $0 < \beta < \frac{1}{d(d + 1)}$. We use the notation introduced in Definitions \ref{definition_hamiltonien}, \ref{definition_energies_vlasov_thomas_fermi} and \ref{definition_husimi}. Let $V$ and $w$ satisfy Assumptions \ref{assumption_V}, \ref{assumption_V2}\footnote{Only when $d=1$.} and \ref{assumption_w}. Let $(\Psi_N)$ be a minimizing sequence of $H_N$ in the sense that
		\begin{equation}
			\langle \Psi_N |H_N \Psi_N\rangle = 1,~~~~~ \langle \Psi_N|H_N\Psi_N\rangle = E(N) + o(N)
		\end{equation} 
		Then there exists a probability measure $P$ on $\mathcal{P}(\R^{2d})$ such that we have the following weak convergence up to extraction
		\begin{equation}
			(2\pi)^{-dk} m_{\Psi_N}^{(k)} \rightharpoonup \int_{\mathcal{P}(\R^d)} \sigma^{\otimes k} \d P(\sigma),
		\end{equation}
		weakly as measures. The probability measure $P$ is concentrated on measures $\sigma$ that satisfy
		\begin{equation}\label{equation_principe_pauli_df}
			\sigma \leq (2\pi)^{-d}.
		\end{equation}
		Moreover, $P$-almost surely, $(2\pi)^d\sigma$ minimizes the Vlasov energy $\mathcal{E}^{\mathrm{V}}$ under the constraints stated in \eqref{equation_egalite_energie_thomas_fermi_vlasov}.
	\end{theorem}
	\begin{remark}[Case $d=2$]
		When $d=2$, the Thomas-Fermi energy functional (and hence the Vlasov functional too) admits exactly one minimizer. Indeed, as $\mathcal{E}^{\mathrm{TF}}$ is strictly convex, there can be at most one minimizer. We recall that for $d=2$,
		\begin{equation}
			\mathcal{E}^{\mathrm{TF}}[\rho] = (c_{\mathrm{TF}}-I_w)\int_{\R^2}\rho^2 + \int_{\R^2}V\rho.
		\end{equation}
		Since $c_{\mathrm{TF}}-I_w>0$, the first part of the energy is weakly lower semi-continuous in $L^2$, and by a tightness argument, we have conservation of the mass. Therefore, by standard arguments, there exists a positive minimizer of $\mathcal{E}^{\mathrm{TF}}$ with integral $1$. By Remark \ref{remark_link_vlasov_thomas_fermi}, the Vlasov energy $\mathcal{E}^{\mathrm{V}}$ admits exactly one minimizer, of the form
		\begin{equation}
			m(x, p) = \mathds{1}_{|p|\leq \sqrt{4\pi\rho(x)}}
		\end{equation}
		where $\rho$ is the minimizer of $\mathcal{E}^{\mathrm{TF}}$, that satisfies
		\begin{equation}
			\rho = \frac{(\lambda - V)_+}{2(c_{\mathrm{TF}} - I_w)},
		\end{equation}
		with $\lambda$ a constant. In particular, $\rho$ is continuous.
		\hfill $\diamond$
	\end{remark}
	The existence of minimizers is less obvious in 1D, for the functional is not weakly lower semi-continuous. We state the result here, and give a proof in the appendix.
	\begin{theorem}[Minimizers of the Thomas-Fermi energy when $d=1$]\label{theorem3}
		Let $d=1$. The minimization problem \eqref{equation_definition_energie_thomas_fermi} admits at least one solution. We denote by $\mathcal{M}^{\mathrm{TF}}$ the sets of such minimizers. Then, for all $\rho \in \mathcal{M}^{\mathrm{TF}}$, $\rho$ is discontinuous and satisfies
		\begin{equation}
			\rho \geq  \frac{I_w}{2c_{\mathrm{TF}}}
		\end{equation}
		on the interior of $\supp(\rho)$.
		Furthermore, the minimization problem \eqref{equation_egalite_energie_thomas_fermi_vlasov} admits at least one solution too, and the solutions are of the form
		\begin{equation}
			m_{\rho}(x, p) = \mathds{1}_{|p| \leq \pi \rho(x)},
		\end{equation}
		with $\rho\in \mathcal{M}^{\mathrm{TF}}$.
	\end{theorem}
	\begin{remark}[Uniqueness of the minimizer]
		In this case, we do not expect uniqueness of the minimizer of the Thomas-Fermi energy. Indeed, if $V$ is a double-well potential, the minimizer might only populate one of the two wells. 
		~\hfill $\diamond$
	\end{remark}
	The proof of Theorem \ref{theorem3} is the object of Appendix \ref{section_appendixe_1}, it is inspired by \cite{benguria_von_1979} and \cite[Section 6]{lieb_thomas-fermi_1981}.

	\subsection{The case of non-purely attractive interactions}\label{subsection_non_attractive}
	In this section, we extend the results of Theorems \ref{theorem1} and \ref{theorem2} to the case of a non-purely attractive system. More accurately, the particles may have both attractive and repulsive interactions, with a stability constraint on the repulsive attractions: its potential has positive Fourier transform. Indeed, a potential with with a positive transform of finite mass is stable \cite[Proposition 1.2]{ruelle_superstable_1970}, and this stability assumption on the interaction potential is usually required for classical particle systems (see also \cite[Section 3.2]{ruelle_statistical_1989}). Before we state the result, let us first define the Hamiltonian and associated semi-classical energies.
	
	\begin{definition}[The Hamiltonian and its domain]\label{definition_hamiltonien2}
		We define the Hamiltonian
		\begin{equation}
			H_N = \sum_{j=1}^N (-\hbar^2\Delta_j) + \sum_{j=1}^NV(x_j) + N^{-1}\sum_{j<k} w_N(x_j - x_k)
		\end{equation}
		on $\mathcal{H}_N$ (see Definition \ref{definition_hamiltonien}) with a potential $V$ satisfying \eqref{equation_premiere_condition_V} and with a short range scaling
		\begin{equation}\label{equation_premiere_condition_w2}
			w_N = N^{d\beta}w(N^{\beta}\cdot) = w_N^+ - w_N^- = N^{d\beta}w^+(N^{\beta}\cdot) - N^{d\beta}w^-(N^{\beta}\cdot) ,~~~~~0<\beta<1/d,~~~~~w^+, w^-\geq 0.
		\end{equation}
		The ground-state energy $E(N)$ is still defined by \eqref{equation_definition_e(n)}.
	\end{definition}
	
	\begin{definition}[Vlasov and Thomas-Fermi energies]\label{definition_energies_vlasov_thomas_fermi2}
		We define the Vlasov energy of a positive integrable function $m$ on the phase space $\R^d \times \R^d$ by
		\begin{equation}\label{equation_definition_vlasov_energy_functional2}
			\mathcal{E}^{\mathrm{V}}[m] = \frac{1}{(2\pi)^d}\iint_{\R^d \times \R^d} \big(|p|^2 + V(x)\big) m(x, p)\d x\d p + I_w \int_{\R^d} \rho_m(x)^2\d x 
		\end{equation}
		with
		\begin{equation}
			I_w = I_{w^+} - I_{w^-} = \frac{1}{2}\int w^+ - \frac{1}{2}\int w^- = \frac{1}{2}\int w
		\end{equation}
		and $\rho_m$ defined by \eqref{equation_definition_rhom}. We define the Thomas-Fermi functional by
		\begin{equation}
			\mathcal{E}^{\mathrm{TF}}[\rho] = c_{\mathrm{TF}}\int_{\R^d} \rho^{1+2/d} + \int_{\R^d} V\rho + I_w\int_{\R^d} \rho^2
		\end{equation}
		with $c_{\mathrm{TF}}$ defined by \eqref{equation_definition_constante_thomas_fermi},
		and the Thomas-Fermi energy by
		\begin{equation}\label{equation_definition_energie_thomas_fermi2}
			E^{\mathrm{TF}} = \inf \left\{\mathcal{E}^{\mathrm{TF}}[\rho],~\rho\geq 0,~\int_{\R^d} \rho =1\right\}.
		\end{equation}
	\end{definition}
	Now, let us state the assumptions on the interaction potentials $w^+$ and $w^-$ that we need.
	\begin{assumption}[The interaction potential $w$]\label{assumption_w2}
		\begin{equation}
			\begin{cases}
				w^- \in L^1(\R^d)\cap L^{\infty}(\R^d)\\
				\nabla w^- \in L^{1}(\R^d)\cap L^{\infty}(\R^d)\\
				\widehat{w^+} \geq 0\\
				w^+ \in L^{1}(\R^d)\cap L^{\infty}(\R^d)\\
				\nabla w^+ \in L^{1 + d/2}(\R^d),
			\end{cases}
		\end{equation}
		and in the case $d=2$:
		\begin{equation}\label{hypothese_w2}
			I_{w^-} = \frac{1}{2}\int_{\R^d} w^- < c_{\mathrm{TF}}.
		\end{equation}
	\end{assumption}
	\begin{remark}[On the last assumption]
		Equation \eqref{hypothese_w2} is stronger than the condition 
		\begin{equation}\label{stabilite_macroscopique}
			I_w + c_{\mathrm{TF}} > 0 
		\end{equation}
		that is required to ensure the stability of the macroscopic system. Indeed, provided that \eqref{stabilite_macroscopique} holds true, the Thomas-Fermi energy is bounded from below. Hence, the stricter assumption \eqref{hypothese_w2} corresponds to a microscopic stability condition. \hfill $\diamond$
	\end{remark}
	\begin{theorem}[Convergence of the energy and the minimizers in the non-purely attractive case]\label{theorem_repulsif}
		Let $d=1$ or $2$, and $0 < \beta < \frac{1}{d(d + 1)}$. If $V$ and $w$ satisfy Assumptions \ref{assumption_V} and \ref{assumption_w2}, we have
		\begin{equation}
			E(N) = N E^{\mathrm{TF}} + o(N)
		\end{equation}
		with the notation introduced in Definitions \ref{definition_hamiltonien2} and \ref{definition_energies_vlasov_thomas_fermi2}. Moreover, let $V$ satisfy Assumption \ref{assumption_V2}\footnote{Only when $d=1$.}. Let $(\Psi_N)$ be a minimizing sequence of $H_N$ in the sense that
		\begin{equation}
			\langle \Psi_N |H_N \Psi_N\rangle = 1,~~~~~ \langle \Psi_N|H_N\Psi_N\rangle = E(N) + o(N)
		\end{equation} 
		Then there exists a probability measure $P$ on $\mathcal{P}(\R^{2d})$ such that we have the following convergence up to extraction
		\begin{equation}
			(2\pi)^{-dk} m_{\Psi_N}^{(k)} \rightharpoonup \int_{\mathcal{P}(\R^d)} \sigma^{\otimes k} \d P(\sigma),
		\end{equation}
		weakly as measures. The probability measure $P$ is concentrated on measures $\sigma$ that satisfy
		\begin{equation}\label{equation_principe_pauli_df2}
			\sigma \leq (2\pi)^{-d}.
		\end{equation}
		Moreover, $P$-almost surely, $(2\pi)^d\sigma$ minimizes the Vlasov energy $\mathcal{E}^{\mathrm{V}}$ under the constraints stated in \eqref{equation_egalite_energie_thomas_fermi_vlasov}.
	\end{theorem}
	
	\subsection{Short plan of the paper}
	~
	
	\indent We divide this paper in three main parts. First, we prove the upper bound corresponding to Theorem \ref{theorem1} in Section \ref{section_upper_bound}. We use Lieb's variational principle \cite{lieb_variational_1981} to prove that we can neglect the correlations in the interaction energy, and bound the ground state energy $E(N)$ by the Hartree energy of a given one-body density matrix. Then, we show that for a given semi-classical measure $m$, one can construct a one body density matrix whose Hartree energy is almost the Vlasov energy of $m$. \\
	\indent Then, inspired by \cite[Section 4]{girardot_semiclassical_2021}, we prove the lower bound corresponding to Theorem \ref{theorem1} in Section \ref{section_lower_bound}. First, we use standard semi-classical arguments to approximate the energy in terms of the Husimi functions. In \cite{fournais_semi-classical_2018}, the authors have shown that in the mean-field regime, the convergence of Husimi functions gives a limit measure $P$ by the Hewitt-Savage theorem. In our case, since the interaction $w_N$ depends on $N$, we have to establish quantitative bounds. Therefore, we use the Diaconis-Freedman theorem \cite{diaconis_finite_1980} in Section \ref{subsection_diaconis_freedman} in order to rewrite the energy as an integral on probability measures. However, this method gives a measure $P^{\mathrm{DF}}_N$ that only charges empirical measures, that do not satisfy the Pauli principle. Hence, we have to average the empirical measures in Section \ref{subsection_averaging}, check that the averaged measures mostly satisfy the Pauli principle, and show that this averaging does not change the energy too much.\\
	\indent Eventually, in Section \ref{section_convergence_states}, we prove the convergence of the Husimi functions associated to ground-states of $H_N$. We show that the one-body Husimi functions of the states is a tight sequence, which implies weak convergence of the Diaconis-Freedman measure. Then, we prove a lower semi-continuity of the energy to use Fatou's lemma, so that the upper bound is recovered in the liminf. In the case $d = 1$, we have to be careful and use relaxed functionals (see Appendix \ref{section_appendixe_1} and \ref{section_appendixe_2}, inspired by \cite{benguria_von_1979}). As the limit measure is concentrated on probability measures satisfying the Pauli principle, these measures must be minimizers of the Vlasov energy.\\
	\indent Finally, in Section \ref{section5}, we explain how to adapt the method for an interacting potential with a repulsive part of positive Fourier transform.
	
	\section{Upper bound on the energy}\label{section_upper_bound}
	In this section, we prove the following upper bound on the energy:
	\begin{proposition}[Upper bound]\label{proposition_upper_bound}
		For $d= 1$ or $2$ and for $0 < \beta < 1/d$, let $V$ and $w$ satisfy Assumptions \ref{assumption_V} and \ref{assumption_w}, we have
		\begin{equation}
			E(N) \leq N E^{\mathrm{TF}} + o(N).
		\end{equation}
	\end{proposition}
	\begin{remark}[Constraint on $\beta$]
		Note that for the upper bound, we only need 
		\begin{equation}
			\beta < 1/d.
		\end{equation}
		The constraints on $\beta$ in Theorem \ref{theorem1} come from the lower bound.
		\hfill $\diamond$
	\end{remark}
	First, we show that we can bound the ground state energy from above by the Hartree energy of a given one particle density matrix, and then we show that the Hartree energy of a well-chosen state is almost the Thomas-Fermi energy.
	
	\subsection{Reduction to the Hartree energy}
	\begin{definition}[Hartree energy]\label{definition_hartree_energy}
		For $\gamma$ a trace-class operator on $\Hcal$ satisfying 
		\begin{equation}
			0\leq \gamma \leq 1,
		\end{equation}  
		let $\mathcal{E}^{\mathrm{H}}[\gamma]$ be the Hartree energy of $\gamma$ defined by
		\begin{equation}\label{equation_definition_fonctionnelle_energie_hartree}
			\mathcal{E}^{\mathrm{H}}[\gamma] = \hbar^2 \tr(-\Delta\gamma) + \tr(V\gamma) - \frac{N^{-1}}{2}\int_{\R^d} (w_N * \rho_{\gamma})\rho_{\gamma},
		\end{equation}
		using the notation introduced in Definition \ref{definition_space_momentum_density} for the density of $\gamma$.
	\end{definition}
	In this section, we prove Proposition \ref{proposition_reduction_hartree_energy} that states that the Hartree energy bounds the ground state energy $E(N)$. For Slater determinants, it is quite clear that the energy is approximated by the Hartree energy (see Remark \ref{remark_slater_lieb}). The one-body reduced density matrix of a Slater determinant, however, is of rank N, while the one-body matrix $\gamma^{\hbar}$ that we define in Lemma \ref{lemma_approximation_semi_classique2} -- and that corresponds to a minimizer of the Vlasov energy -- is not. Therefore, we use Lieb's variational principle (see Proposition \ref{proposition_lieb_variational}) to relax the constraint on the rank.
	\begin{definition}[$N$-body density matrix]
		We call $N$-body density matrix a positive self adjoint operator $\Gamma_N$ acting on $\Hcal_N$ of trace $1$. We define its $k$-body reduced matrix (for $k\leq N$) by the partial trace
		\begin{equation}
			\Gamma_N^{(k)} = \frac{N!}{(N-k)!}\tr_{k+1\to N}(\Gamma_N).
		\end{equation}
	\end{definition} 
	Before we prove Proposition \ref{proposition_reduction_hartree_energy}, we have to show that under certain assumptions, for each one-body operator $\gamma$, there exists an $N$-body density matrix $\Gamma$ such that 
	\begin{equation}
		\Gamma^{(1)} = \gamma,~~~~~\Gamma^{(2)} \approx \gamma^{\otimes 2}.
	\end{equation}
	We define the exchange operator that we will need, as it appears naturally in the next proposition.
	\begin{definition}[Exchange operator]
		We define the exchange operator $\mathrm{Ex}$ on $L^2(\R^d)^{\otimes 2}$ by
		\begin{equation}
			\forall u, v\in L^2(\R^d),~\mathrm{Ex} (u\otimes v) = v \otimes u.
		\end{equation}
		Its kernel is
		\begin{equation}
			\mathrm{Ex}(x, x' ; y, y') = (x, x'; y', y).
		\end{equation}
	\end{definition}
	Now, we can state Lieb's variational principle \cite{lieb_variational_1981} with a use that was inspired by \cite[Theorem VI.4]{perice_multiple_2024}.
	\begin{proposition}[Lieb's variational principle]\label{proposition_lieb_variational}
		Let $\gamma_N$ acting on $\Hcal$ and satisfying 
		\begin{equation}\label{equation_condition_gamma_N}
			\tr(\gamma_N) = N ~~~~~0 \leq \gamma_N \leq 1.
		\end{equation}
		Then, there exists $\Gamma_N$ an N-body density matrix and $L$ a positive operator such that
		\begin{equation}\label{equation_lieb_variational}
			\Gamma_N^{(1)} = \gamma_N~~~~~\Gamma_N^{(2)} = \gamma_N\otimes \gamma_N(1 - \mathrm{Ex}) - L.
		\end{equation}
	\end{proposition}
	\begin{remark}[Slater determinants]\label{remark_slater_lieb}
		For $\Psi_N\in\Hcal_N$ a Slater determinant and $\Gamma_N = |\Psi_N\rangle \langle \Psi_N|$, we have the Wick rule
		\begin{equation}
			\Gamma_N^{(2)} = \Gamma_N^{(1)}\otimes \Gamma_N^{(1)} (1 - \mathrm{Ex}),
		\end{equation}
		i.e.
		\begin{equation}
			L = 0.
		\end{equation}
	\end{remark}
	Now, we show that the trace of $L$ is small compared to that of $\Gamma_N^{(2)}$. Indeed, $\tr(L)$ is of order $N$ while $\tr(\Gamma_N^{(2)})$ is of order $N^2$.
	\begin{lemma}[$L$ is small]\label{lemma_L_is_small}
		Let $\gamma_N$ be a self-adjoint operator satisfying \eqref{equation_condition_gamma_N} and $L$ defined as in Proposition \ref{proposition_lieb_variational}. Then the trace norm of $L$ is bounded
		\begin{equation}
			\|L\|_{\sigma^1} \leq  N.
		\end{equation}
		\hfill $\diamond$
	\end{lemma}
	\begin{proof}
		We take the trace in \eqref{equation_lieb_variational}:
		\begin{equation}
			\tr\big(\Gamma_N^{(2)}\big) = \tr\big(\gamma_N\otimes\gamma_N\big) - \tr\big(\gamma_N\otimes \gamma_N \mathrm{Ex}\big) - \tr(L).
		\end{equation}
		By definition, 
		\begin{equation}
			\tr(\Gamma^{(2)}_N) = N(N-1),~~~~~~\tr(\gamma_N\otimes\gamma_N) = N^2.
		\end{equation}
		Moreover, let us write 
		\begin{equation} 
			\gamma_N = \sum \lambda_j |u_j\rangle \langle u_j|
		\end{equation}
		with $(u_j)$ an orthonormal basis of $\Hcal$. We have
		\begin{equation}\label{equation_borne_trace_exchange1}
			\tr\big(\gamma_N\otimes\gamma_N \mathrm{Ex}\big) = \tr\Big(\sum_{j, k} \lambda_j\lambda_k |u_j\otimes u_k\rangle\langle u_k\otimes u_j|\Big)= \sum_j \lambda_j^2 = \tr\big((\gamma_N)^2\big)\geq 0.
		\end{equation}
		Therefore,
		\begin{equation}
			\tr(L) \leq N^2 - N(N-1) = N.
		\end{equation}
		Since $L$ is positive, this concludes the proof of Lemma \ref{lemma_L_is_small}.
	\end{proof}
	\begin{remark}[The exchange term is small]
		Note that \eqref{equation_borne_trace_exchange1} also implies
		\begin{equation}\label{equation_borne_trace_exchange2}
			\tr\big(\gamma_N\otimes\gamma_N \mathrm{Ex}\big) \leq \tr(\gamma_N) = N.
		\end{equation}
		\hfill $\diamond$
	\end{remark}
	Thanks to Proposition \ref{proposition_lieb_variational} and Lemma \ref{lemma_L_is_small}, for any operator satisfying \eqref{equation_condition_gamma_N}, we know that there exists an $N$-body density matrix such that 
	\begin{equation}
		\Gamma_N^{(2)} \approx \gamma_N \otimes \gamma_N.
	\end{equation}
	Hence, the interaction energy of $\gamma_N$ -- without correlations -- is almost the interaction energy of $\Gamma_N$ -- with correlations. Then, we can bound the ground state energy by the Hartree energy of a given $\gamma_N$.
	\begin{proposition}[Reduction to the Hartree energy]\label{proposition_reduction_hartree_energy}
		Let $\beta <1/d$ and $\gamma_N$ satisfy \eqref{equation_condition_gamma_N}. There exists $\Gamma_N$ an N-body density matrix such that
		\begin{equation}\label{equation_egalite_energies_hartree_vraie}
			\mathcal{E}^{\mathrm{H}}[\gamma_N] = \big(1 + o(1)\big)\tr(H_N\Gamma_N) + o(N),
		\end{equation}
		with $\mathcal{E}^{\mathrm{H}}$ defined in \eqref{equation_definition_fonctionnelle_energie_hartree}. In particular, the ground state energy $E(N)$ is bounded by the Hartree energy of $\gamma_N$
		\begin{equation}\label{equation_borne_sup_energie_vraie_hartree}
			E(N) \leq \big(1 + o(1)\big)\mathcal{E}^{\mathrm{H}}[\gamma_N] + o(N).
		\end{equation}
	\end{proposition}
	\begin{proof}
		First, let us establish \eqref{equation_egalite_energies_hartree_vraie}. Since we have
		\begin{equation}
			\Gamma_N^{(1)} = \gamma_N,
		\end{equation}
		it is clear that
		\begin{equation}
			\tr\Big((-\hbar^2 \Delta + V)\gamma_N\Big) = \tr\Bigg(\sum_{j=1}^N \big(-\hbar^2 \Delta_j + V(x_j)\big)\Gamma_N\Bigg).
		\end{equation}
		Furthermore,
		\begin{equation}
			\left|\tr\Big(N^{-1}w_N(x-y)\gamma_N^{\otimes 2}\Big) - \frac{N-1}{N}\tr\Big(N^{-1}w_N(x - y)\Gamma_N^{(2)} \Big)\right| \leq \|w_N\|_{L^{\infty}} \Big( \tr\big(\gamma_N\otimes\Gamma_N \mathrm{Ex}\big) + \tr(L)\Big),
		\end{equation}
		then
		\begin{equation}
			\tr\Big(N^{-1}w_N(x-y)\gamma_N^{\otimes 2}\Big) = \big(1 + o(1)\big)\tr\Big(N^{-1}w_N(x - y)\Gamma_N^{(2)}\Big) + O(N^{d\beta}),
		\end{equation}
		and hence \eqref{equation_egalite_energies_hartree_vraie}. Now, note that
		\begin{equation}
			E(N) = \inf_{\Hcal_N} \frac{\langle \Psi_N | H_N \Psi_N\rangle}{\langle \Psi_N|\Psi_N\rangle} = \inf \Big\{\tr(H_N\Gamma_N),~\Gamma_N~ N\mathrm{-body~density~matrix}\Big\}
		\end{equation}
		which implies \eqref{equation_borne_sup_energie_vraie_hartree}.
	\end{proof}
	
	\subsection{Estimate of the Hartree energy}
	Here, we want to prove that one can choose an appropriate $\gamma_N$ satisfying \eqref{equation_condition_gamma_N} such that
	\begin{equation}
		\mathcal{E}^{\mathrm{H}}[\gamma_N] \approx N E^{\mathrm{TF}}.
	\end{equation}
	To do so, we show that starting with a semi-classical measure $m$ on the phase space $\R^{2d}$, one can construct a $\gamma$ with a Hartree energy close to the Vlasov energy of $m$.
	\begin{definition}[$N$-Vlasov energy]\label{definition_N_Vlasov}
		We set
		\begin{equation}
			\mathcal{E}_N^{\mathrm{V}}[m] = \frac{1}{(2\pi)^d}\iint_{\R^d \times \R^d} \big(|p|^2 + V(x)\big) m(x, p)\d x\d p - \frac{1}{2}\int_{\R^{2d}} w_N(x - y) \rho_m(x)\rho_m(y)\d x\d y
		\end{equation}
		the $N$-Vlasov energy functional. Note that formally, the Vlasov functional defined in \eqref{equation_definition_vlasov_energy_functional} corresponds to $N=+\infty$.
	\end{definition}
	\begin{lemma}[Semi-classical approximation of the energy]\label{lemma_approximation_semi_classique2}
		Let $d=1$ or $2$ and $\beta <1/d$. Let $V$ satisfy Assumption \ref{assumption_V}, $m\in L^1(\R^{2d})$ satisfy
		\begin{equation}\label{equation_hypotheses_m_lemme_approximation_semiclassique2}
			\iint m = (2\pi)^d,~~~~~0 \leq m \leq 1,~~~~~\supp \rho_m \subset B(0, C),~~~~~\rho_m\in L^{1 + 2/d}
		\end{equation}
		with $\rho_m$ defined in \eqref{equation_definition_rhom}, and
		\begin{equation}\label{equation_defintion_gamma_hbar_lemme_approx}
			\gamma^{\hbar} = (2\pi\hbar)^{-d}\iint_{\R^d\times\R^d} m(x, p) |f_{x, p}^{\hbar}\rangle\langle f_{x, p}^{\hbar}|\d x\d p.
		\end{equation}
		Then, for $\hbar_x$ and $\hbar_p$ defined as in Definition \ref{definition_coherent_state}, satisfying $\hbar_x\hbar_p = \hbar^2$ and
		\begin{equation}\label{equation_conditions_sur_hbarx_hbarp}
			\begin{cases}
				\hbar_x \ll N^{-2\beta}\\
				\hbar_p \ll 1,
			\end{cases}
		\end{equation}
		we have
		\begin{equation}\label{equation_lemme_approximation_semi_classique2}
			\mathcal{E}^{\mathrm{H}}[\gamma^{\hbar}] = N \mathcal{E}_N^{\mathrm{V}}[m] + o(N),
		\end{equation}
		the energies being introduced in Definitions \ref{definition_hartree_energy} and \ref{definition_N_Vlasov}.
	\end{lemma}
	\begin{remark}[Choice of $\hbar_x$ and $\hbar_p$]\label{remark_choix_de_hbarx_hbarp}
		The condition \eqref{equation_conditions_sur_hbarx_hbarp} is compatible with $\hbar_x\hbar_p = \hbar^2$ because $\beta < 1/d$. One can take for instance $\hbar_x = N^{-\beta-1/d}$ and $\hbar_p = N^{\beta-1/d}$.
		\hfill $\diamond$
	\end{remark}
	\begin{proof}[Proof of Lemma \ref{lemma_approximation_semi_classique2}]
		For the kinetic energy, we have
		\begin{equation}
			\tr(-\hbar^2 \Delta \gamma^{\hbar}) = (2\pi\hbar)^{-d} \iint_{\R^d \times \R^d} \langle f_{x, p}^{\hbar}|-\hbar^2 \Delta f_{x, p}^{\hbar}\rangle m(x, p)\d x\d p.
		\end{equation}
		Since\footnote{The cross term is 0 because $\int f \nabla f = 0$.}
		\begin{equation}
			\iint |\nabla f_{x, p}^{\hbar}|^2 = \hbar_x^{-1} \|\nabla f\|_{L^2}^2 + \hbar^{-2}|p|^2,
		\end{equation}
		we have
		\begin{equation}\label{equation_semi_classique_approximation_cinetique}
			\tr(-\hbar^2 \Delta \gamma^{\hbar}) = (2\pi\hbar)^{-d}\iint_{\R^d\times\R^d} |p|^2 m(x, p)\d x\d p + O(N\hbar_p).
		\end{equation}
		Then, for the potential
		energy, we have
		\begin{equation}
			\tr(V \gamma^{\hbar}) = (2\pi\hbar)^{-d}\iint_{\R^d\times\R^d} \langle f_{x, p}^{\hbar}|V f_{x, p}^{\hbar}\rangle m(x, p)\d x \d p.
		\end{equation}
		As
		\begin{equation}
			\hbar_x^{-d/2}\int_{\R^d} V(y) \left|f\left(\frac{x - y}{\sqrt{\hbar_x}}\right)\right|^2 \d y = V(x) + O\big(\|\nabla V\|_{L^{\infty}\big(B(0, 2C)\big)}\sqrt{\hbar_x}\big),
		\end{equation}
		we have
		\begin{equation}\label{equation_semi_classique_approximation_potentiel}
			\tr(V\gamma^{\hbar}) = (2\pi\hbar)^{-d} \iint_{\R^d\times\R^d} V(x) m(x, p)\d x\d p + O(N\sqrt{\hbar_x}).
		\end{equation}
		Finally, for the interaction energy, we have
		\begin{equation}
			N^{-1}\int_{\R^d} (w_N*\rho_{\gamma^{\hbar}})\rho_{\gamma^{\hbar}} = N^{-1}\int_{\R^d} \bigg(\Big(\big(w_N*|f^{\hbar}|^2 \big) * |f^{\hbar}|^2\Big) * \rho_m\bigg)\rho_m.
		\end{equation}
		By Young's inequality, 
		\begin{equation}\label{equation_juste_avant_resultat_intermediaire}
			\left|\int \bigg(\Big(\big(w_N*|f^{\hbar}|^2 \big) * |f^{\hbar}|^2 - w_N\Big) * \rho_m\bigg)\rho_m\right| \leq \Big\|\big(w_N*|f^{\hbar}|^2 \big) * |f^{\hbar}|^2 - w_N\Big\|_{L^1} \|\rho_m\|_{L^2}^2.
		\end{equation}
		Moreover, for all $v$ in the Sobolev space $W^{1, p}(\R^d)$
		\begin{equation}\label{equation_resultat_intermediaire}
			\big\|v - v*|f^{\hbar}|^2\big\|_{L^p} \leq C\sqrt{\hbar_x}\|\nabla v\|_{L^p}.
		\end{equation}
		Indeed, by Taylor's theorem,
		\begin{equation}
			\Big|v(x) - \big(v*|f^{\hbar}|^2\big)(x)\Big| = \bigg|\int_{\R^d} \big(v(x) - v(x - y)\big)|f^{\hbar}(y)|^2\d y\bigg| = \bigg|\int_{\R^d} y.\int_0^1\nabla v(x-ty)\d t |f^{\hbar}(y)|^2 \d y\bigg|.
		\end{equation}
		Hence, since $|f^{\hbar}|^2$ is of integral 1, by Jensen's theorem, we have
		\begin{equation}
			\big\|v - v*|f^{\hbar}|^2\big\|^p_{L^p} \leq C \hbar_x^{p/2} \int_{\R^{2d}} \int_0^1 |\nabla v(x-ty)|^p |f^{\hbar}(y)|^2 \d t \d x\d y.
		\end{equation}
		Finally, by exchanging the integrals and a simple change of variable $z=x - ty$, we find
		\begin{equation}
			\big\|v - v*|f^{\hbar}|^2\big\|^p_{L^p} \leq C \hbar_x^{p/2} \int_0^1\int_{\R^{2d}}|\nabla v(z)|^p  \d z |f^{\hbar}(y)|^2\d y \d t = C \hbar_x^{p/2} \|\nabla v\|_{L^p}^p.
		\end{equation}
		Thus, by \eqref{equation_juste_avant_resultat_intermediaire} and \eqref{equation_resultat_intermediaire} applied to $p=1$ and $v = w_N$, we have 
		\begin{equation}\label{equation_semi_classique_approximation_interaction}
			N^{-1}\int_{\R^d} (w_N*\rho_{\gamma^{\hbar}})\rho_{\gamma^{\hbar}} = N^{-1}(2\pi\hbar)^{2d} \int_{\R^d} (w_N * \rho_m)\rho_m + O(\sqrt{\hbar_x}N^{1 + \beta}).
		\end{equation}
		Putting together \eqref{equation_semi_classique_approximation_cinetique}, \eqref{equation_semi_classique_approximation_potentiel} and \eqref{equation_semi_classique_approximation_interaction}, and choosing $\hbar_x$ and $\hbar_p$ such that
		\begin{equation}
			\hbar_x \ll N^{-2\beta}~~~~~\mathrm{and}~~~~~\hbar_p \ll 1,
		\end{equation} 
		e.g. as in Remark \ref{remark_choix_de_hbarx_hbarp}, we find \eqref{equation_lemme_approximation_semi_classique2}.
	\end{proof}
	
	With Lemma \ref{lemma_approximation_semi_classique2}, we can now prove Proposition \ref{proposition_upper_bound}.
	\begin{proof}[Proof of Proposition \ref{proposition_upper_bound}]
		Since minimizers of the Vlasov energy do not necessarily exist a priori in the case $d=1$, we take $m_{\varepsilon}$ satisfying \eqref{equation_hypotheses_m_lemme_approximation_semiclassique2} and such that
		\begin{equation}
			\mathcal{E}^{\mathrm{V}}[m_{\varepsilon}] \leq E^{\mathrm{TF}} + \varepsilon.
		\end{equation}
		The existence is clear: we take $m$ that approximates the Thomas-Fermi energy, and we move the mass outside of $B(0, C)$ inside, which does not change the energy too much since this mass is very small for $C$ large enough, because of the confining potential $V$. Then, using Lemma \ref{lemma_approximation_semi_classique2} and denoting by $\gamma_{\varepsilon}$ the operator defined by \eqref{equation_defintion_gamma_hbar_lemme_approx} with $m = m_{\varepsilon}$, we have
		\begin{equation}
			\mathcal{E}^{\mathrm{H}}[\gamma] = N\mathcal{E}_N^{\mathrm{V}}[m_{\varepsilon}] + o(N).
		\end{equation}
		Then, we have clearly convergence of the interaction energy when $N \to + \infty$
		\begin{equation}
			\frac{1}{2}\int_{\R^d} \big(w_N * \rho_{m_{\varepsilon}}\big) \rho_{m_{\varepsilon}} \to I_w \int_{\R^d} \rho_{m_{\varepsilon}}^2,
		\end{equation}
		and thus of the $N$-Vlasov energy. Hence,
		\begin{equation}
			\mathcal{E}^{\mathrm{H}}[\gamma] = N\mathcal{E}^{\mathrm{V}}[m_{\varepsilon}] + o(N) \leq N\big(E^{\mathrm{TF}} + \varepsilon \big) + o(N).
		\end{equation}
		Using Proposition \ref{proposition_reduction_hartree_energy}, we find
		\begin{equation}\label{equation_fin_preuve_borne_superieure}
			\frac{1}{N}E(N) \leq \big(1 + o(1)\big)\big(1 + \varepsilon\big)E^{\mathrm{TF}}.
		\end{equation}
		Since \eqref{equation_fin_preuve_borne_superieure} holds for all $\varepsilon > 0$, we finally have
		\begin{equation}
			\frac{1}{N}E(N) \leq E^{\mathrm{TF}} + o(1).
		\end{equation}
	\end{proof}
	\section{Lower bound on the energy}\label{section_lower_bound}
	In this section, we prove the following lower bound:
	\begin{proposition}[Lower bound]\label{proposition_lower_bound}
		Let $d=1$ or $2$, $\beta < \frac{1}{d(d + 1)}$, $V$ and $w$ satisfy Assumptions \ref{assumption_V} and \ref{assumption_w}. Let $\Psi_N$ be a sequence of approximate minimizers of the energy, i.e.
		\begin{equation}
			\frac{\langle \Psi_N|H_N\Psi_N\rangle}{\langle \Psi_N | \Psi_N\rangle} = E(N) + o(N),
		\end{equation}
		then we have
		\begin{equation}
			\langle\Psi_N |H_N\Psi_N\rangle \geq N E^{\mathrm{TF}} + o(N).
		\end{equation}
		In particular, 
		\begin{equation}
			E(N) \geq N E^{\mathrm{TF}} + o(N).
		\end{equation}
	\end{proposition}
	First, we prove a semi-classical approximation, using previously established a priori bounds, in order to rewrite the enrgy in terms of the Husimi functions. Then, we execute a mean-field approximation thanks to the Diaconis-Freedman theorem. However, the Diaconis-Freedman measure only charges empirical probability measures, that in particular do not satisfy the Pauli principle. In order to recover the Vlasov energy, we finally have to average the empirical measures so that the averaged measures satisfy the Pauli principle.
	\subsection{A priori bounds}\label{subsection_a_priori_bounds}
	First, we can establish the following a priori bounds on the kinetic and interaction energies of a sequence of states whose total energy is of order $N$.
	\begin{lemma}[A priori estimates on the energy]\label{lemma_apriori}
		Assume $\beta < 1/d$ and let $\Psi_N\in \Hcal_N$ of norm 1 such that 
		\begin{equation}\label{equation_condition_sur_psi_lemme32}
			\langle \Psi_N | H_N \Psi_N\rangle = O(N).
		\end{equation}
		Then, 
		\begin{equation}\label{equation_borne_apriori1}
			\Big\langle \Psi_N \Big| \sum\big(-\hbar^2 \Delta_j + V(x_j)\big) \Psi_N\Big\rangle = O(N^{1 + \beta d/2})
		\end{equation}
		and for every $v \in L^{1 + d/2}$,
		\begin{equation}\label{equation_borne_apriori2}
			\left|\frac{1}{N}\int_{\R^{2d}} v(x - y) \rho_{\Psi_{N}}^{(2)}(x, y)\d x\d y\right| \leq C N^{1 + \frac{\beta d^2}{2(d+2)}}\|v\|_{L^{1 + d/2}}.
		\end{equation}
	\end{lemma}
	The proof of \eqref{equation_borne_apriori1} was inspired by that of \cite[Lemma 21.3]{girardot_approximation_2021}
	\begin{proof}
		First, we prove \eqref{equation_borne_apriori1}. We fix $N \in \N$, let 
		\begin{equation}
			H_N^0 = \sum_{j = 1}^N \big(\hbar^2 (-\Delta_j) + V(x_j)\big),
		\end{equation}
		we want to bound $\langle \Psi_N |H_N^0\Psi_N\rangle$ for $\Psi_N\in \Hcal_N$ satisfying \eqref{equation_condition_sur_psi_lemme32}. Let
		\begin{equation}
			\widetilde{H}_N = \frac{1}{2}\sum_{j = 1}^N \hbar^2 (-\Delta_j) - \frac{1}{N}\sum_{j<k} w_N(x_j - x_k).
		\end{equation}
		Bounding $\langle \Psi_N |\widetilde{H}_N \Psi_N\rangle$ from below will prove \eqref{equation_borne_apriori1}; indeed, for $\Psi_N \in \Hcal_N$, we have
		\begin{equation}\label{equation_inegalite_entre_hamiltoniens}
			\langle \Psi_N|H_N\Psi_N\rangle \geq \frac{1}{2}\langle \Psi_N| H_N^0 \Psi\rangle + \langle \Psi_N |\widetilde{H}_N \Psi_N\rangle.
		\end{equation}
		Let 
		\begin{equation}
			\widetilde{H}_{N-1}(x_1) = \frac{1}{2}\sum_{j = 2}^{N} \left( \frac{N}{N-1}(-\hbar^2 \Delta_j)  - w_N(x_1 - x_j)\right),
		\end{equation}
		we have by antisymmetry of $\Psi_N$
		\begin{equation}\label{equation_egalite_entre_htilde}
			\langle\Psi_N|\widetilde{H}_N\Psi_N\rangle = \int_{\R^d} \Big\langle \Psi_N(x_1)\Big|\widetilde{H}_{N-1}(x_1)\Psi_N(x_1)\Big\rangle \d x_1.
		\end{equation}
		We take now $\Psi_{N-1} \in \Hcal_{N-1}$ of norm 1. We have the following Lieb-Thirring inequality \cite[Theorem 4.3]{lieb_stability_2009}
		\begin{equation}\label{equation_lieb_thirring_inequality}
			\Big\langle \Psi_{N-1}\Big|\sum_{j=2}^N -\hbar^2 \Delta_j \Psi_{N-1}\Big\rangle \geq C_1 (N-1) \int_{\R^d} \left(\frac{\rho_{\Psi_{N-1}}}{N-1}\right)^{1 + 2/d}
		\end{equation}
		and by Hölder's and Young's inequalities, there exists $C_2$ such that
		\begin{align}
			\left|\int_{\R^d} w_N(x_1 - x_2)\rho_{\Psi_{N-1}}(x_2)\d x_2\right|&\leq (N-1) \|w_N\|_{L^{1 + d/2}}\left\|\frac{\rho_{\Psi_{N-1}}}{N-1}\right\|_{L^{1+2/d}} \notag\\
			&\leq C_2 (N-1) \|w_N\|_{L^{1 + d/2}}^{1+d/2} + C_1 (N-1)\left\|\frac{\rho_{\Psi_{N-1}}}{N-1}\right\|_{L^{1+2/d}}^{1+2/d}.\label{equation_borne_sup_w_linearise}
		\end{align}
		Thus, 
		\begin{equation}\label{equation_borne_htilde_n-1}
			\big\langle \Psi_{N-1}\big|\widetilde{H}_{N-1}(x_1)\Psi_{N-1}\big\rangle \geq -C_2 N\|w_N\|_{L^{1+d/2}}^{1+d/2} = -C_2 N^{1 + \beta d/2}\int_{\R^d} w^{1+d/2}.
		\end{equation}
		Inserting \eqref{equation_borne_htilde_n-1} in \eqref{equation_egalite_entre_htilde} with
		\begin{equation}\label{equation_definition_psi_n1}
			\Psi_{N-1} = \frac{\Psi_N(x_1)}{\|\Psi_N(x_1)\|_{L^2(\R ^{d(N-1)})}},
		\end{equation}
		we find
		\begin{equation}\label{equation_borne_htilde_n}
			\langle \Psi_N|\widetilde{H}_N\Psi_N\rangle \geq - C N^{1+\beta d/2}.
		\end{equation}
		Therefore, using \eqref{equation_inegalite_entre_hamiltoniens} and \eqref{equation_borne_htilde_n}, we have 
		\begin{equation}
			\langle \Psi_N|H_N^0\Psi_N\rangle \leq 2\langle \Psi_N|H_N\Psi_N\rangle - 2\langle \Psi_N|\widetilde{H}_N\Psi_N\rangle = O(N^{1+\beta d/2}).
		\end{equation}
		Let us now prove \eqref{equation_borne_apriori2}. Using once again the notation \eqref{equation_definition_psi_n1} with an implicit dependence on $x_1$, we have
		\begin{equation}
			I:= \frac{2}{N}\int_{\R^{2d}} v(x-y) \rho_{\Psi_N}^{(2)}(x, y)\d x \d y = \int_{\R^d} \left(\int_{\R^d} v(x_1 - x_2)\rho_{\Psi_{N-1}}(x_2)\d x_2\right) \|\Psi_N(x_1)\|^2 \d x_1.
		\end{equation}
		Then, we can use the first inequality in \eqref{equation_borne_sup_w_linearise}, as well as the Lieb-Thirring's inequality \eqref{equation_lieb_thirring_inequality} and Jensen's inequality, to find
		\begin{align}
			I &\leq 2N \|v\|_{L^{1 + d/2}}\int_{\R^d}\left\|\frac{\rho_{\Psi_{N-1}}}{N-1}\right\|_{L^{1+2/d}}\|\Psi_N(x_1)\|^2\d x_1\notag\\
			&\leq 2N \|v\|_{L^{1 + d/2}}\int_{\R^d}  \bigg(\frac{C}{N} \Big\langle \Psi_{N-1} \Big| \sum_{j=2}^N -\hbar^2 \Delta_j \Psi_{N-1}\Big\rangle \bigg)^{\frac{1}{1 + 2/d}}\|\Psi_N(x_1)\|^2\d x_1\notag\\
			&\leq CN \|v\|_{L^{1 + d/2}} \left(N^{-1} \int_{\R^d} \Big\langle \Psi_{N-1} \Big| \sum_{j=2}^N -\hbar^2 \Delta_j \Psi_{N-1}\Big\rangle  \|\Psi_N(x_1)\|^2\d x_1\right)^{\frac{1}{1 + 2/d}} \notag\\
			&\leq CN \|v\|_{L^{1 + d/2}}\left(N^{-1}\Big\langle \Psi_N\Big| \sum_{j = 1}^N - \hbar^2 \Delta_ j \Psi_N\Big\rangle \right)^{\frac{1}{1 + 2/d}}.
		\end{align}
		Then, since
		\begin{equation}
			\Big\langle \Psi_N\Big| \sum_{j = 1}^N - \hbar^2 \Delta_ j \Psi_N\Big\rangle = O(N^{1 + \beta d/2}),
		\end{equation}
		we recover
		\begin{equation}
			I \leq CN \|v\|_{L^{1+d/2}} \big(N^{\beta d /2}\big)^{\frac{1}{1 + 2/d}} = C \|v\|_{L^{1 + d/2}} N^{1 + \frac{\beta d^2}{2(d+2)}}.
		\end{equation}
	\end{proof}
	
	\subsection{Semi-classical computations}\label{subsection_semiclassical}
	In this section, we show that for an appropriate state $\Psi_N$, we can approximate the energy of $\Psi_N$ by a semi-classical energy of the one and two particles Husimi functions of $\Psi_N$.
	\begin{lemma}[Semi-classical approximation of the energy]\label{lemma_semi_classical_approximation_energy}
		Let $V$ satisfy Assumption \ref{assumption_V}. Let $\Psi_N \in \Hcal_N$ satisfy
		\begin{equation}
			\langle \Psi_N| H_N \Psi_N\rangle = O(N).
		\end{equation}
		Let $\beta < \frac{2}{d(2d + 1)}$. Then, for
		\begin{equation}\label{equation_choix_hbarx_hbarp_lemme33}
			\begin{cases}
				\hbar_p \ll 1\\
				\hbar_x\ll N^{-(2d+1)\beta},
			\end{cases}
		\end{equation}
		we have
		\begin{multline}\label{equation_semi_classique_borne_inf}
			\frac{\langle \Psi_N|H_N\Psi_N\rangle}{N} = \frac{1}{(2\pi)^d}\iint_{\R^d\times \R^d} \big(|p|^2 + V(x)\big)m_{\Psi_N}^{(1)}(x, p)\d x \d p \\ - \frac{1}{2(2\pi)^{2d}}\iint_{\R^{2d}\times\R^{2d}} w_N(x-y)m_{\Psi_N}^{(2)}(x, y; p, q)\d x\d y \d p \d q + o(1).
		\end{multline}
	\end{lemma}
	\begin{remark}[Choice of $\hbar_x$ and $\hbar_p$]
		The condition \eqref{equation_choix_hbarx_hbarp_lemme33} is compatible with $\hbar_x\hbar_p = \hbar^2 = N^{-2/d}$ since
		\begin{equation}
			N^{-2/d}\ll N^{-(2d+1)\beta}
		\end{equation}
		because
		\begin{equation}
			-(2d+1)\beta > - 2/d.
		\end{equation}
		\hfill $\diamond$
	\end{remark}
	\begin{proof}[Proof of Lemma \ref{lemma_semi_classical_approximation_energy}]
		We denote $\gamma_N$ the one-particle reduced density matrix of $\Psi_N$. First, we compute the kinetic energy using \eqref{equation_lien_m_t}, as well as notation introduced in \eqref{equation_definition_tgamma} and \eqref{equation_definition_coherent_state_space}.
		\begin{align}\label{equation_energie_cinetique_preuve_lemme33}
			\frac{N}{(2\pi)^d} \iint_{\R^d\times\R^d} |p|^2 m_{\Psi_N}^{(1)}(x, p)\d x\d p &= \frac{N}{(2\pi)^d} \iint |p|^2 \langle f_{x, p}^{\hbar}|\gamma_N f_{x, p}^{\hbar}\rangle \d x\d p  \notag\\ 
			&= \int_{\R^{2d}}  |p-q|^2 t_{\gamma_N}(p) |g^{\hbar}(q)|^2 \d p \d q\notag\\
			&= \int_{\R^{2d}} |p|^2 t_{\gamma_N}(p) |g^{\hbar}(q)|^2 \d p \d q + \int_{\R^{2d}} |q|^2 t_{\gamma_N}(p) |g^{\hbar}(q)|^2 \d p \d q\notag\\
			&= \int_{\R^d}|p|^2 t_{\gamma_N}(p)\d p + \tr(\gamma_N)\hbar_p^{-d/2} \int_{\R^d} |q|^2 |\widehat{f}(q/\sqrt{\hbar_p})|^2 \d q\notag\\
			&= \int_{\R^d} |p|^2 t_{\gamma_N}(p)\d p + N \hbar_p \int_{\R^d} |k|^2 |\widehat{f}(k)|^2 \d k\notag\\
			&= \tr\big( -\hbar^2 \Delta \gamma_N\big) + N \hbar_p \|\nabla f\|_{L^2}^2.
		\end{align}
		The $pq$ term coming from the expansion of $|p-q|^2$ cancels out because $|g^{\hbar}|^2$ is radial. Then, we compute the potential energy using \eqref{equation_lien_m_rho}
		\begin{align}
			\frac{N}{(2\pi)^d}\iint_{\R^d\times\R^d} V(x) m_{\Psi_N}^{(1)}(x, p)\d x\d p &= \frac{N}{(2\pi)^d}\iint_{\R^d\times\R^d} V(x) \langle f_{x, p}^{\hbar}|\gamma_N f_{x, p}^{\hbar}\rangle \d x\d p\notag\\
			&= \int_{\R^d} V(x) \Big(\rho_{\gamma_N} * |f^{\hbar}|^2\Big)(x) \d x \notag\\
			&= \int_{\R^d} \Big(V * |f^{\hbar}|^2 \Big)(x) \rho_{\gamma_N}(x) \d x.
		\end{align}
		Recall that 
		\begin{equation}
			\supp f^{\hbar} \subset B(0, C\sqrt{\hbar_x}),
		\end{equation}
		then
		\begin{equation}
			\int_{\R^d} \bigg| V(x) - \int_{\R^d} V(x - y) |f^{\hbar}(y)|^2\d y\bigg| \rho_{\gamma_N}(x)\d x \leq C\sqrt{\hbar_x} \int_{\R^d} \|\nabla V\|_{L^{\infty}\big(B(0, |x|+C\sqrt{\hbar_x})\big)}\rho_{\gamma_N(x)}\d x.
		\end{equation}
		Using Assumption \ref{assumption_V} and Young's inequality,
		\begin{equation}
			\int_{\R^d} \big|V - V*|f^{\hbar}|^2\big|\rho_{\gamma_N} \leq C\sqrt{\hbar_x} \int_{\R^d} \big(|x|^{s-1} + 1\big) \rho_{\gamma_N}(x)\d x \leq C\sqrt{\hbar_x} \int_{\R^d} V \rho_{\gamma_N} + C N \sqrt{\hbar_x}.
		\end{equation}
		The a priori estimate \eqref{equation_borne_apriori1} then gives
		\begin{equation}
			\int_{\R^d} \big|V - V*|f^{\hbar}|^2\big|\rho_{\gamma_N} \leq C \sqrt{\hbar_x} N^{1 + \beta d/2}
		\end{equation}
		and therefore
		\begin{equation}\label{equation_energie_potentielle_preuve_lemme33}
			\frac{N}{(2\pi)^d}\iint_{\R^d\times\R^d} V(x) m_{\Psi_N}^{(1)}(x, p)\d x\d p = \tr(V\gamma_N) + O\big(N^{1 + \beta d/2} \sqrt{\hbar_x}\big)
		\end{equation}
		Finally, we compute the interaction energy using \eqref{equation_lien_m_rho}. We have
		\begin{multline}\label{equation_la_quon_va_reecrire_ensuite}
			\frac{N^2}{(2\pi)^{2d}}\iint_{\R^{2d}\times\R^{2d}} w_N(x-y)m_{\Psi_N}^{(2)}(x, y; p, q)\d x \d y \d p \d q \\
			= 2 \int_{\R^{4d}} w_N(x_1 - y_1) \rho_{\Psi_N}^{(2)}(x_2, y_2) \big|f^{\hbar}(x_1 - x_2)\big|^2 \big|f^{\hbar}(y_1 - y_2)\big|^2 \d x_1 \d x_2 \d y_1 \d y_2.
		\end{multline}
		Let
		\begin{equation} 
			v_N = w_N - \big(w_N*|f^{\hbar}|^2\big)*|f^{\hbar}|^2,
		\end{equation} 
		for $x, y\in \R^d$, we have
		\begin{align}
			v_N(x-y) &= w_N(x-y) - \Big(\big(w_N*|f^{\hbar}|^2\big)*|f^{\hbar}|^2\Big)(x-y)\notag\\
			&= w_N(x - y) - \int_{\R^d}\big(w_N*|f^{\hbar}|^2\big)(x-y - x') |f^{\hbar}(x')|^2\d x'\notag\\
			&= w_N(x - y) - \int_{\R^{2d}} w_N(x - y - x' - y') |f^{\hbar}(y')|^2|f^{\hbar}(x')|^2\d y' \d x'\notag\\
			&= w_N(x - y) - \int_{\R^{2d}} w_N(x - y - x' + y') |f^{\hbar}(y')|^2|f^{\hbar}(x')|^2\d y' \d x'.
		\end{align}
		By a change of variable $x'' = x-x'$ and $y'' = y - y'$, we find
		\begin{equation}
			v_N(x - y) = w_N(x-y) - \int w_N\big(x''-y''\big)|f^{\hbar}(x-x'')|^2|f^{\hbar}(y-y'')|^2\d x''\d y''.
		\end{equation}
		Then, with this new notation, we can rewrite \eqref{equation_la_quon_va_reecrire_ensuite} as
		\begin{multline}
			\frac{N^2}{(2\pi)^{2d}}\iint_{\R^{2d}\times\R^{2d}} w_N(x-y)m_{\Psi_N}^{(2)}(x, y; p, q)\d x \d y \d p \d q = 2\int_{\R^{2d}} w_N(x-y) \rho_{\Psi_N}^{(2)}(x, y)\d x\d y \\- 2\int_{\R^{2d}} v_N(x-y) \rho_{\Psi_N}^{(2)}(x, y)\d x\d y.
		\end{multline}
		By \eqref{equation_borne_apriori2} applied to $v_N$, we can bound the difference to $\int w_N\rho_{\Psi_N}^{(2)}$:
		\begin{multline}
			\frac{N^2}{(2\pi)^{2d}}\iint_{\R^{2d}\times\R^{2d}} w_N(x-y)m_{\Psi_N}^{(2)}(x, y; p, q)\d x \d y \d p \d q\\ = 2\int_{\R^{2d}} w_N(x - y)\rho_{\Psi_N}^{(2)}(x, y)\d x\d y + O\big( N^{1 + \frac{\beta d^2}{2(d+2)}} \|v_N\|_{L^{1 + d/2}}\big).
		\end{multline}
		Because of \eqref{equation_resultat_intermediaire},
		\begin{equation}
			\|v_N\|_{L^{1+d/2}} = \big\|w_N - \big(w_N*|f^{\hbar}|^2\big)*|f^{\hbar}|^2\big\|_{L^{1 + d/2}} \leq \sqrt{\hbar_x}\|\nabla w_N\|_{L^{1+d/2}} = \sqrt{\hbar_x} N^{(d+1)\beta - \frac{d\beta}{1 + d/2}}\|\nabla w\|_{L^{1 + d/2}}.
		\end{equation}
		Moreover,  
		\begin{equation}
			\frac{\beta d^2}{2(d + 2)} - \frac{d \beta }{1 +d/2} = \beta\frac{d^2 - 4d}{2(d+2)} = \frac{-\beta}{2}
		\end{equation}
		for $d =1$ or $2$, and we finally find
		\begin{multline}\label{equation_energie_interaction_preuve_lemme33}
			\frac{N^2}{(2\pi)^{2d}}\iint_{\R^{2d}\times\R^{2d}} w_N(x-y)m_{\Psi_N}^{(2)}(x, y; p, q)\d x \d y \d p \d q\\ = 2\int_{\R^{2d}} w_N(x - y)\rho_{\Psi_N}^{(2)}(x, y)\d x\d y + O\big( N^{1 + (d + 1/2)\beta}\sqrt{\hbar_x}\big).
		\end{multline}
		Collecting \eqref{equation_energie_cinetique_preuve_lemme33}, \eqref{equation_energie_potentielle_preuve_lemme33} and \eqref{equation_energie_interaction_preuve_lemme33}, we find
		\begin{multline}\label{equation_fin_preuve_lemme33}
			\frac{\langle \Psi_N|H_N\Psi_N\rangle}{N} = \frac{1}{(2\pi)^d}\iint_{\R^d\times\R^d} \big(|p|^2 + V(x)\big)m_{\Psi_N}^{(1)}(x, p)\d x \d p \\ - \frac{1}{2(2\pi)^{2d}}\iint_{\R^{2d}\times\R^{2d}} w_N(x-y)m_{\Psi_N}^{(2)}(x, y; p, q)\d x\d y \d p \d q + O(N\hbar_p) + O\big( N^{1 + (d + 1/2)\beta}\sqrt{\hbar_x}\big).
		\end{multline}
		Now, using \eqref{equation_choix_hbarx_hbarp_lemme33}, we recover \eqref{equation_semi_classique_borne_inf} from \eqref{equation_fin_preuve_lemme33}.
	\end{proof}
	
	\subsection{The Diaconis-Freedman theorem and its direct consequences}\label{subsection_diaconis_freedman}
	From now on, $\beta$ always satisfies $\beta < \frac{2}{d(2d+1)}$, and we fix a sequence $(\Psi_N)$ such that 
	\begin{equation}\label{equation_proprietes_psi}
		\Psi_N\in \Hcal_N,~~~~~~~~~~\|\Psi_N\| = 1,~~~~~~~~~~\langle\Psi_N |H_N\Psi_N\rangle = E(N) + o(N).
	\end{equation}
	In this section, we use the Diaconis-Freedman theorem to rewrite the energy. First, we introduce the notation for empirical measures in order to state the theorem. Then, we apply Theorem \ref{theorem_diaconis_freedman} and Lemma \ref{lemma_first_marginals} to $m_{\Psi_N}^{(N)}$.
	\begin{definition}[Empirical measure]
		For $Z_N = (z_j) = (x_j, p_j) \in \mathbb{R}^{2dN}$ a set of points in the phase space, we define the empirical measure associated with $Z_N$ by
		\begin{equation}
			\mathrm{Emp}_{Z_N} = \frac{1}{N}\sum_{j = 1}^N \delta_{z_j}.
		\end{equation}
	\end{definition}
	Now we can state the Diaconis-Freedman theorem \cite{diaconis_finite_1980}:
	\begin{theorem}[Diaconis-Freedman]\label{theorem_diaconis_freedman}
		Let $m_N$ be a symmetric probability measure over $\mathbb{R}^{2dN}$. Let $P_N^{\mathrm{DF}}$ be the probability measure over $\mathcal{P}(\mathbb{R}^{2d})$ defined by
		\begin{equation}\label{equation_definition_diaconis_freedman}
			P_N^{\mathrm{DF}}(\sigma) = \int_{\mathbb{R}^{2dN}}\delta_{\sigma = \mathrm{Emp}_{Z_N}}\d m_N(Z_N). 
		\end{equation}
		Moreover, we set
		\begin{equation}\label{equation_definition_approximation_diaconis_freedman}
			\widetilde{m}_N = \int_{\mathcal{P}(\mathbb{R}^{2d})} \sigma^{\otimes N}\d P_N^{\mathrm{DF}}(\sigma)
		\end{equation}
		with marginals
		\begin{equation}
			\widetilde{m}_N^{(k)} = \int_{\mathcal{P}(\mathbb{R}^{2d})} \sigma^{\otimes k}\d P_N^{\mathrm{DF}}(\sigma).
		\end{equation}
		Then, 
		\begin{equation}\label{equation_difference_m_tilde_ou_non_TV}
			\|m_N^{(k)} - \widetilde{m}_N^{(k)}\|_{\mathrm{TV}} \leq \frac{2k(k-1)}{N}.
		\end{equation}
		Furthermore, if the sequence of measures $(m_N^{(1)})$ is tight, there exists a probability measure $P^{\mathrm{DF}}$ on $\mathcal{P}(\mathbb{R}^{2dN})$ such that up to extraction, we have the weak convergence of $P_N^{\mathrm{DF}}$ as measures when $N \to + \infty$
		\begin{equation}
			P_N^{\mathrm{DF}} \rightharpoonup P^{\mathrm{DF}},
		\end{equation}
		and for any $k \in \mathbb{N}$, 
		\begin{equation}
			\widetilde{m}_N^{(k)} \rightharpoonup \int_{\mathcal{P}(\mathbb{R}^{2d})} \sigma^{(k)} \d P^{\mathrm{DF}}(\sigma),
		\end{equation}
		weakly as measures when $N\to+\infty$.
	\end{theorem}
	The limit probability $P^{\mathrm{DF}}$ in the second part of the theorem is that appearing in \cite{hewitt_symmetric_1955}. For a proof of Theorem \ref{theorem_diaconis_freedman}, see \cite[Theorem 2.1 and 2.2]{rougerie_finetti_2020}. The following lemma computes explicitly the first marginals of the measure $\widetilde{m}_N$ constructed above.
	\begin{lemma}[First marginals]\label{lemma_first_marginals}
		Let $m_N$ be a symmetric probability measure over $\mathbb{R}^{2dN}$ and its associated Diaconis-Freedman approximation $\widetilde{m}_N$ defined by \eqref{equation_definition_approximation_diaconis_freedman}. Then for all $x_1, x_2,p_1,p_2 \in \mathbb{R}^d$, we have
		\begin{equation}\label{equation_first_marginals1}
			\widetilde{m}_N^{(1)} (x_1, p_1) = m_N^{(1)}(x_1, p_1)
		\end{equation}
		and
		\begin{equation}\label{equation_first_marginals2}
			\widetilde{m}_N^{(2)}(x_1, p_1 ; x_2, p_2) = \frac{N-1}{N}m_N^{(2)}(x_1, p_1 ; x_2, p_2)+ \frac{1}{N}m_N^{(1)}(x_1, p_1)\delta_{(x_1, p_1) = (x_2, p_2)}. 
		\end{equation}
	\end{lemma}
	A proof of Lemma \ref{lemma_first_marginals} can be found in \cite[Remark 2.3]{rougerie_finetti_2020} for instance. In the rest of the section, we will use the following notation:
	\begin{equation}\label{equation_definition_m_N}
		m_N = \frac{m_{\Psi_N}^{(N)}}{N! (2\pi \hbar)^d}
	\end{equation}
	which implies, for all $k\leq N$,
	\begin{equation}\label{equation_definition_m_N_k}
		m_{\Psi_N}^{(k)} = \frac{N!}{(N-k)!}(2\pi \hbar)^{dk}m_N^{(k)}.
	\end{equation}
	We still have to introduce a last notation before rewriting the energy in terms of the Diaconis-Freedman measure.
	\begin{definition}[Semi-classical energy of a one-body measure]
		Let $\mu$ be a probability measure on the phase space $\R^{2d}$. We define its semi-classical energy by
		\begin{equation}\label{equation_definition_mathcalE_N}
			\mathcal{E}_N[\mu] = \iint_{\R^d\times\R^d} \big(|p|^2 + V(x)\big) \d \mu(x, p) - \frac{1}{2} \iint_{\R^{2d}\times\R^{2d}} w_N(x-y) \d \mu^{\otimes 2}(x, p ; y, q).
		\end{equation}
	\end{definition}
	\begin{remark}[Link with the $N$-Vlasov energy]
		This energy is closely related to the Vlasov energy. Indeed, for $m$ regular enough, we have
		\begin{equation}
			\mathcal{E}_N\Big[(2\pi)^{-d}m\Big] = \mathcal{E}_N^{\mathrm{V}}[m].
		\end{equation}
		The definition of $\mathcal{E}_N$, however, allows us to consider singular measures whose spatial density is not defined, and is also more convenient for probability measures -- as opposed to the constraints on $m$ in \eqref{equation_egalite_energie_thomas_fermi_vlasov}.
		\hfill $\diamond$
	\end{remark}
	With this new notation, Theorem \ref{theorem_diaconis_freedman} and Lemma \ref{lemma_first_marginals}, we have the following
	\begin{lemma}[Energy in terms of the Diaconis-Freedman measure]\label{lemma_energy_involving_diaconis_freedman}
		Let $V$ satisfy Assumption \ref{assumption_V}. Let $\Psi_N\in \Hcal_N$ satisfy \eqref{equation_proprietes_psi}, $P_N^{\mathrm{DF}}$ defined by \eqref{equation_definition_diaconis_freedman} with the notation \eqref{equation_definition_m_N}, and $\mathcal{E}_N$ defined by \eqref{equation_definition_mathcalE_N}. Then, we have
		\begin{equation}\label{equation_energy_involving_diaconis_freedman}
			\frac{1}{N}\langle\Psi_N| H_N \Psi_N\rangle = \int_{\mathcal{P}(\R^{2d})} \mathcal{E}_N[\mu] \d P_N^{\mathrm{DF}}(\mu) + o(1).
		\end{equation}
	\end{lemma}
	\begin{proof} First, since $\Psi_N$ satisfies \eqref{equation_proprietes_psi}, we can use Lemma \ref{lemma_semi_classical_approximation_energy}. Then, we have by \eqref{equation_definition_m_N_k} and \eqref{equation_first_marginals1}
		\begin{align}
			\frac{1}{(2\pi)^d}\iint_{\R^d\times\R^d}\big(|p|^2 + V(x)\big)\d m_{\Psi_N}^{(1)}(x, p) &= \iint_{\R^d\times\R^d}\big(|p|^2 + V(x)\big)\d m_N^{(1)}(x, p)\notag\\
			&= \iint_{\R^d\times\R^d}\big(|p|^2 + V(x)\big)\d \widetilde{m}_N^{(1)}(x, p)\notag\\
			&= \int_{\mathcal{P}(\R^{2d})} \iint_{\R^d\times\R^d} \big(|p|^2 + V(x)\big)\d \mu \d P^{\mathrm{DF}}_N(\mu),
		\end{align}
		and by \eqref{equation_first_marginals2},
		\begin{align}
			\frac{1}{(2\pi)^{2d}}\iint_{\R^{2d}\times\R^{2d}} w_N&(x - y)\d m_{\Psi_N}^{(2)}(x, y; p, q) = \frac{N-1}{N}\iint_{\R^{2d}\times\R^{2d}} w_N(x - y)\d m_{N}^{(2)}(x, y; p, q) \notag\\
			&= \iint_{\R^{2d}\times\R^{2d}} w_N(x - y)\d \widetilde{m}_{N}^{(2)}(x, y; p, q) - \frac{1}{N} w_N(0)\iint_{\R^d\times \R^d} \d \widetilde{m}_N(x, p) \notag\\
			&= \int_{\mathcal{P}(\R^{2d})} \iint_{\R^{2d}\times\R^{2d}} w_N(x-y)\d \mu^{\otimes 2}(x, y; p, q) \d P^{\mathrm{DF}}_N(\mu) + O(N^{d\beta - 1}).
		\end{align}
		It only remains to recall that $\beta < 1/d$ to conclude the proof.
	\end{proof}
	\subsection{Averaging the measure}\label{subsection_averaging}
	It is clear from its definition in \eqref{equation_definition_diaconis_freedman} that $P^{\mathrm{DF}}_N$ only charges empirical measures, that, in particular, do not respect the Pauli principle \eqref{equation_principe_pauli_df}. Hence, we have to average the measures $\mu$ that appear in \eqref{equation_energy_involving_diaconis_freedman}, and then check ($i$) that replacing $\mu$ by an average on small boxes does not modify the energy too much and ($ii$) that the averaged measures do satisfy a Pauli principle with large probability. Before we do so, we introduce the tiling on which we will average the measures.
	\begin{definition}[Tiling]\label{definition_tiling}
		We divide the phase space into small hyperrectangles $\Omega_j$ 
		\begin{equation}
			\R^{2d} = \overline{\bigcup_{m\in \N} \Omega_j}~~~~~~~~~~~ \forall j\neq n,~\Omega_j\cap\Omega_n = \emptyset.
		\end{equation}
		Let $\gamma>0$, we denote by $l_x$ and $l_p$ respectively the side-length of $\Omega_j$ in the space and momentum space, so that
		\begin{equation}\label{equation_volume_hyperrectangle}
			|\Omega_j| = l_x^d l_p^d = N^{- \gamma}.
		\end{equation}
	\end{definition}
	\begin{definition}[Empirical measures that violate the approximate Pauli principle on $\Omega_j$]~\\
		For a given $\varepsilon >0$, let $\Gamma_{\varepsilon}^j$ be the set of empirical measures that do not respect an approximated Pauli principle on $\Omega_j$, i.e.
		\begin{equation}
			\Gamma_{\varepsilon}^j = \bigg\{ \mathrm{Emp}_{Z_N}, \int_{\Omega_j} \mathrm{Emp}_{Z_N} \geq \frac{(1 + \varepsilon)}{(2\pi)^d}|\Omega_j|\bigg\}.
		\end{equation}
	\end{definition}
	We can then directly apply \cite[Theorem 4.4]{girardot_semiclassical_2021} since $m_N$ satisfies the following Pauli principle for any $k\leq N$
	\begin{equation}
		m_N^{(k)} \leq \frac{N^k}{(2\pi)^k N...(N-k+1)}.
	\end{equation}
	We reproduce the statement here.
	\begin{proposition}[Probability of violating the approximated Pauli principle on $\Omega_j$]\label{proposition_proba_violating_pauli_principle}
		Let $m_N$ be defined by \eqref{equation_definition_m_N}, $P_N^{\mathrm{DF}}$ the associated Diaconis-Freedman measure given by \eqref{equation_definition_diaconis_freedman}, and $\Omega_j$ from the tiling introduced in Definition \ref{definition_tiling}. Then, for $\delta$ satisfying 
		\begin{equation}\label{equation_condition_delta}
			0 < \delta < \frac{1 - \gamma}{2},
		\end{equation}
		with $\gamma > 0$ as in Definition \ref{definition_tiling}, there exists two constant $C_{\delta} > 0$ and $c_{\delta} > 0$ such that for any $\varepsilon > 0$,
		\begin{equation}
			P_N^{\mathrm{DF}}\big(\Gamma_{\varepsilon}^j\big) \leq C_{\delta}e^{-c_{\delta}N^{\delta}\ln(1 + \varepsilon)} .
		\end{equation}
	\end{proposition}
	Now, we want to restrict the measure $\mu$ to a compact support. 
	\begin{notation}[Hypercube]\label{notation_hypercube}
		For $n\in \N$, let 
		\begin{equation}\label{equation_definition_SL}
			L = n N^{- \gamma/2d},~~~~~~~~~~S_L = [-L, L]^{2d}.
		\end{equation}
		Note that 
		\begin{equation}
			|S_L| = (2L)^{2d} = (2n)^{2d} N^{-\gamma} = (2n)^{2d} |\Omega_j|.
		\end{equation}
		In particular, we can choose $(\Omega_j)$ such that 
		\begin{equation}
			S_L = \overline{\bigcup_{j\leq 4^d n^{2d}}\Omega_j}.
		\end{equation}
	\end{notation}
	\begin{lemma}[Restriction to a compact support]\label{lemme_restriction_support_compact}
		Let $\mu$ be a probability measure satisfying 
		\begin{equation}\label{equation_hypothese_energie_bornee}
			\iint_{\R^d\times\R^d} \big(|p|^2 + V(x)\big)\d\mu(x, p) \leq \tau.
		\end{equation}
		Then,
		\begin{equation}
			\mathcal{E}_N[\mu]\geq \mathcal{E}_N[\mu\mathds{1}_{S_L}] - \frac{C \tau^2 N^{d\beta}}{\inf(L^4, L^{2s})}.
		\end{equation}
	\end{lemma}
	\begin{proof} 
		First, it is clear by positivity of $V$ that
		\begin{equation}
			\iint_{\R^d\times\R^d}\big(|p|^2 + V(x)\big)\d \mu(x, p) \geq \int_{S_L} \big(|p|^2 + V(x)\big) \d \mu(x, p).
		\end{equation}
		For the interaction energy, we decompose the measure $\mu^{\otimes 2}$
		\begin{equation}\label{equation_decomposition_mu_carre}
			\mu^{\otimes 2} = \mathds{1}_{S_L} \mu \otimes \mathds{1}_{S_L}\mu + \mathds{1}_{S_L^c} \mu \otimes \mathds{1}_{S_L}\mu + \mathds{1}_{S_L} \mu \otimes \mathds{1}_{S_L^c}\mu + \mathds{1}_{S_L^c} \mu \otimes \mathds{1}_{S_L^c}\mu.
		\end{equation}
		By Young's inequality, we have 
		\begin{equation}
			\left|\iint_{\R^{2d}\times\R^{2d}} w_N(x - y) \d \big(\mathds{1}_{S_L^c}\mu\big)^{\otimes 2}\right| \leq \|w_N\|_{L^{\infty}}\big\|\mathds{1}_{S_L^c}\big\|_{L^1(\d \mu)}^2.
		\end{equation}
		Moreover, by \eqref{hypothese_V1} and \eqref{equation_hypothese_energie_bornee},
		\begin{equation}
			\iint_{S_L^c}\d \mu \leq C\iint_{S_L^c} \frac{\big(|p|^2 + V(x)\big)}{\inf(L^2, L^s)}\d \mu(x, p) \leq C \frac{\tau}{\inf(L^2, L^s)},
		\end{equation}
		and thus
		\begin{equation}
			\left|\iint_{\R^{2d}\times\R^{2d}} w_N(x - y) \d \big(\mathds{1}_{S_L^c}\mu\big)^{\otimes 2}\right| \leq C N^{d\beta}\frac{\tau^2}{\inf(L^4, L^{2s})}.
		\end{equation}
		The cross term from \eqref{equation_decomposition_mu_carre} can be dealt with using the same argument. Indeed, by Young's inequality, we have
		\begin{align}
			\left|\iint_{\R^{2d}\times\R^{2d}} w_N(x - y)\d \big(\mathds{1}_{S_L}\mu \otimes \mathds{1}_{S_L^c}\mu\big)\right| &= \left|\iint_{\R^{2d}\times\R^{2d}} w_N(x - y) \d \big(\mathds{1}_{S_L\setminus S_{L - R_wN^{-\beta}}}\mu \otimes \mathds{1}_{S_L^c}\mu\big)\right|\notag\\
			&\leq \|w_N\|_{L^{\infty}} \|\mathds{1}_{S_{L - R_w N^{-\beta}}^c}\|_{L^1(\d \mu)}^2\notag\\
			&\leq C N^{d\beta} \frac{\tau^2}{\inf(L^4, L^{2s})}.
		\end{align}
		Therefore,
		\begin{equation}
			\iint_{\R^{2d}\times\R^{2d}} w_N(x - y) \d \mu^{\otimes 2}(x, y; p, q) \geq \iint_{\R^{2d}\times\R^{2d}} w_N(x - y) \d\big(\mathds{1}_{S_L}\mu\big)^{\otimes 2} - CN^{d\beta}\frac{\tau^2}{\inf(L^4, L^{2s})}.
		\end{equation}
		This concludes the proof of Lemma \ref{lemme_restriction_support_compact}.
	\end{proof}
	Now, we can define the averaged measure.
	\begin{definition}[Averaged measure]
		Let $\mu$ be a probability measure, $(\Omega_j)$ and $n$ like in Definition \ref{definition_tiling} and Notation \ref{notation_hypercube}, we define the averaged measure of $\mu$ by
		\begin{equation}\label{equation_definition_averaged_measure}
			\bar{\mu} = \sum_{j = 1}^{4^d n^{2d}}\mathds{1}_{\Omega_j} \iint_{\Omega_j}\frac{\d \mu}{|\Omega_j|}.
		\end{equation}
	\end{definition}
	\begin{definition}[Empirical measures that violate the local Pauli principle] Let $(\Omega_j)$, $S_L$ and $n$ be like in Definition \ref{definition_tiling} and Notation \ref{notation_hypercube}. Let $\Gamma_{\varepsilon}$ be the set of empirical measures violating the approximated Pauli principle on at least one $\Omega_j \subset S_L$
		\begin{equation}\label{equation_definition_gamma_varepsilon}
			\Gamma_{\varepsilon} = \Big\{\mathrm{Emp}_{Z_N}, \exists \Omega_j \subset S_L, \iint_{\Omega_j} \mathrm{Emp}_{Z_N}\geq \frac{(1 + \varepsilon)}{(2\pi)^d}|\Omega_j|\Big\} = \bigcup_{j = 1}^{4^d n^{2d}} \Gamma_{\varepsilon}^j.
		\end{equation}
	\end{definition}
	\begin{lemma}[Averaging]\label{lemma_averaging}
		Let $V$ and $w$ satisfy Assumptions \ref{assumption_V} and \ref{assumption_w}. For $\varepsilon > 0$ small enough, let $Z_N$ such that $\mathrm{Emp}_{Z_N} \in \Gamma_{\varepsilon}^c$. Moreover, let $\beta < \frac{1}{d(d+1)}$, and assume
		\begin{equation}\label{equation_condition_lx_lp}
			l_x \ll N^{-(d+1)\beta},~~~~~l_p \ll 1
		\end{equation}
		and
		\begin{equation}
			\iint_{\R^d \times \R^d} \big(|p|^2 + V(x)\big) \d \mathrm{Emp}_{Z_N}(x, p) \leq \tau.
		\end{equation}
		Then,
		\begin{equation}\label{equation_lemme_averaging}
			\mathcal{E}_N[\mathrm{Emp}_{Z_N}] \geq \big(1 - o(1)\big)\mathcal{E}_N\Big[\overline{\mathrm{Emp}_{Z_N}}\Big]- o(1) - \frac{C \tau^2N^{d\beta}}{\inf(L^4, L^{2s})},
		\end{equation}
		using the notation introduced in \eqref{equation_definition_averaged_measure}.
	\end{lemma}
	\begin{proof}
		We denote by $\mu$ the empirical measure $\mathrm{Emp}_{Z_N}$. First,
		\begin{align}
			\Delta_1 :&= \left|\iint_{\R^d\times \R^d} \big(|p|^2 +V(x)\big) \d \big(\bar{\mu} - \mathds{1}_{S_L} \mu\big)(x, p)\right|\notag\\
			&\leq \sum_{j = 1}^{4^d n^{2d}}\int_{\Omega_j}\big(|p|^2 + V(x)\big) \d\left|\mu - \int_{\Omega_ j}\frac{\d \mu}{|\Omega_j|}\right|\notag\\
			&\leq \sum_{m=1}^{4^d n^{2d}} \sup_{\Omega_j} \Big(2 l_p |p| + |\nabla V(x)| l_x\Big) \int_{\Omega_j}\d \mu \notag\\
			&\leq \frac{2}{(2\pi)^d}\sum_{m=1}^{4^d n^{2d}} \sup_{\Omega_j} \Big(2 l_p |p| + |\nabla V(x)| l_x\Big)|\Omega_j|.
		\end{align}
		Then, we use Rieman sums, and the inequality
		\begin{equation}
			\forall \alpha >0, ~\alpha^{s-1}\leq \frac{s-1}{s}\alpha^s + \frac{1}{s},
		\end{equation}
		to find, using Assumption \ref{assumption_V},
		\begin{align}
			\Delta_1 &\leq l_p\iint_{\R^d\times\R^d} \big(|p|^2 + 1\big) \d \bar{\mu} + Cl_x\iint_{\R^d\times\R^d} \Big(\frac{s-1}{s}V(x) + \frac{1}{s}\Big) \d \bar{\mu}\notag\\
			&\leq C (l_p + l_x) \iint_{\R^d\times\R^d} \big(|p|^2 + V(x)\big) \d \bar{\mu} + C (l_p + l_x)\iint_{\R^d\times\R^d} \d\bar{\mu}.\label{equation_borne_Delta1}
		\end{align}
		Let us show that since $\bar{\mu}$ is regular, we can bound $\int(|p|^2 +V(x))\d \bar{\mu}(x, p)$ by a constant times the whole energy $\mathcal{E}_N[\bar{\mu}]$. Let
		\begin{equation}
			m := \frac{(2\pi)^d}{1+\varepsilon}\bar{\mu} \leq 1
		\end{equation} 
		since $\mu \in \Gamma_{\varepsilon}^{\mathrm{c}}$. Then, by the bath-tube principle,
		\begin{equation}
			\frac{1}{1+\varepsilon}\iint_{\R^d\times\R^d} |p|^2 \d \bar{\mu}(x, p) = \frac{1}{(2\pi)^d}\iint_{\R^d\times\R^d} |p|^2 \d m(x, p) \geq c_{\mathrm{TF}}\int_{\R^d} \rho_m^{1 + 2/d}.
		\end{equation}
		When $d=2$, we have
		\begin{equation}
			\frac{1}{(1+\varepsilon)^2}\iint_{\R^{2d}\times\R^{2d}} w_N(x-y)\d \bar{\mu}^{\otimes 2}(x, y; p, q) = \int_{\R^d} (w_N*\rho_{m})\rho_{m} \leq I_w \int_{\R^d} \rho_{m}^2 
		\end{equation}
		and therefore
		\begin{equation}
			\mathcal{E}_N[\bar{\mu}] \geq \bigg(1 - (1+\varepsilon)\frac{I_w}{c_{\mathrm{TF}}}\bigg)\iint_{\R^d\times\R^d} \big(|p|^2 + V(x)\big)\d \bar{\mu}(x, p).
		\end{equation}
		For $\varepsilon$ small enough, we have indeed
		\begin{equation}
			1 - (1+\varepsilon)\frac{I_w}{c_{\mathrm{TF}}} > 0.
		\end{equation}
		Likewise, when $d=1$, by Young's inequality,
		\begin{multline}
			\int_{\R^d} (w_N*\rho_{m})\rho_{m} \leq I_w \int_{\R^d} \rho_{m}^2 \leq I_w\bigg(\frac{I_w}{2c_{\mathrm{TF}}}\int_{\R^d} \rho_{m} + \frac{c_{\mathrm{TF}}}{2I_w}\int_{\R^d} \rho_{m}^{3}\bigg) \\ \leq \frac{I_w^2}{2c_{\mathrm{TF}}}(1+\varepsilon)^{-1} + \frac{1}{2}(1 + \varepsilon)^{-1}\iint_{\R^d\times\R^d} |p|^2\d\bar{\mu}(x, p),
		\end{multline}
		and thus, by positivity of $V$, 
		\begin{equation}
			\mathcal{E}_N[\bar{\mu}] \geq \frac{1-\varepsilon}{2}\iint_{\R^d\times\R^d} \big(|p|^2 + V(x)\big)\d \bar{\mu}(x, p) - \frac{I_w^2}{2c_{\mathrm{TF}}}(1+\varepsilon).
		\end{equation}
		Therefore, for $d=1$ or $2$, and for $\varepsilon < 1$, we find
		\begin{equation}
			\Delta_1 \leq C(l_p + l_x) \big(\mathcal{E}_N[\bar{\mu}] + 1\big) + C(l_p + l_x).
		\end{equation}
		Moreover,
		\begin{align}
			\Delta_2 :&= \left|\iint_{\R^{2d}\times\R^{2d}} w_N(x - y)\d \Big(\bar{\mu}^{\otimes 2} - (\mathds{1}_{S_L}\mu)^{\otimes 2}\Big)\right| \notag\\
			&\leq \sum_{j = 1}^{4^dn^{2d}}\sum_{q = 1}^{4^d n^{2d}}\left|\iint_{\Omega_j\times\Omega_q} w_N(x - y)\d \Big(\bar{\mu}^{\otimes 2} - (\mathds{1}_{S_L}\mu)^{\otimes 2}\Big)\right|\notag\\
			&\leq \sum_{j = 1}^{4^dn^{2d}}\sum_{q = 1}^{4^d n^{2d}}\left|\iint_{\Omega_j\times\Omega_q} \bigg( \frac{1}{|\Omega_j|}\frac{1}{|\Omega_q|}\iint_{\Omega_j\times\Omega_q} w_N(x' - y')\d x' \d p' \d y' \d k' - w_N(x - y)\bigg) \d \mu(x, p)\d \mu(y, k)\right|\notag\\
			&\leq 2 l_x \|\nabla w_N\|_{L^{\infty}}\sum_{j = 1}^{4^dn^{2d}}\sum_{q = 1}^{4^d n^{2d}}\mu(\Omega_j)\mu(\Omega_k)\notag\\
			&\leq 2 \|\nabla w\|_{L^{\infty}} l_x N^{(d+1)\beta}. \label{equation_borne_Delta2}
		\end{align}
		Using \eqref{equation_borne_Delta1}, \eqref{equation_borne_Delta2} and Lemma \ref{lemme_restriction_support_compact}, as well as the assumption \eqref{equation_condition_lx_lp}, we find \eqref{equation_lemme_averaging}.
	\end{proof}
	
	\subsection{Conclusion}\label{subsection_conclusion}
	In this section, we always require $\beta < \frac{1}{d(d+1)}$. As seen in Lemma \ref{lemme_restriction_support_compact} and Lemma \ref{lemma_averaging}, we have to restrict ourselves to measures $\mu$ whose kinetic and potential energies are bounded. We first define the set consisting of such measures.
	\begin{definition}[Measures of bounded energy]
		For $\tau > 0$, we set
		\begin{equation}\label{equation_definition_Xitau}
			\Xi_{\tau} = \bigg\{ \mu \in \mathcal{P}(\R^{2d}),~\iint_{\R^{2d}\times\R^{2d}} \big(|p|^2 + V(x)\big)\d \mu(x, p) \leq \tau\bigg\}.
		\end{equation}
	\end{definition}
	\begin{lemma}[Restricting to measures of bounded energy satisfying the Pauli principle] \label{lemme_restriction_measures_bornees_pauli}
		Let $P_N^{\mathrm{DF}}$ be the Diaconis-Freedman measure defined by \eqref{equation_definition_diaconis_freedman} with the convention \eqref{equation_definition_m_N}, $\mathcal{E}_N$ the semi-classical energy defined by \eqref{equation_definition_mathcalE_N}, and $\Gamma_{\varepsilon}$ and $\Xi_{\tau}$ be the sets of measures defined by \eqref{equation_definition_gamma_varepsilon} and \eqref{equation_definition_Xitau}. Then, we have the following lower bound on the energy:
		\begin{equation}
			\int_{\mathcal{P}(\R^{2d})} \mathcal{E}_N[\mu]\d P_N^{\mathrm{DF}}(\mu) \geq \int_{\Xi_{\tau}\setminus\Gamma_{\varepsilon}} \mathcal{E}_N[\mu]\d P_N^{\mathrm{DF}}(\mu) - C N^{d \beta} n^{2d}P^{\mathrm{DF}}_N(\Gamma_{\varepsilon}^j) - CN^{3d\beta/2}\tau^{-1} - o(1).
		\end{equation}
	\end{lemma}
	\begin{proof} First, it is clear that
		\begin{equation}
			\int_{\mathcal{P}(\R^{2d})} \iint_{\R^d\times\R^d}\big(|p|^2 + V(x)\big)\d \mu(x, p)\d P_N^{\mathrm{DF}}(\mu) \geq \int_{\Xi_{\tau}\setminus \Gamma_{\varepsilon}} \iint_{\R^d\times\R^d}\big(|p|^2 + V(x)\big)\d \mu(x, p)\d P_N^{\mathrm{DF}}(\mu).
		\end{equation}
		On the other hand, 
		\begin{multline}
			\int_{\mathcal{P}(\R^{2d})}\iint_{\R^{2d}\times\R^{2d}} w_N(x-y)\d\mu^{\otimes 2}(x, y; p, q) \leq \int_{\Xi_{\tau} \setminus \Gamma_{\varepsilon}}\iint_{\R^{2d}\times\R^{2d}} w_N(x-y)\d\mu^{\otimes 2}(x, y; p, q) \\ - C\|w_N\|_{L^{\infty}}\Big(P^{\mathrm{DF}}_N(\Gamma_{\varepsilon}) + P^{\mathrm{DF}}_N(\Xi_{\tau}^c)\Big),
		\end{multline}
		and therefore
		\begin{equation}
			\int_{\mathcal{P}(\R^{2d})} \mathcal{E}_N[\mu]\d P_N^{\mathrm{DF}}(\mu) \geq \int_{\Xi_{\tau}\setminus\Gamma_{\varepsilon}} \mathcal{E}_N[\mu]\d P_N^{\mathrm{DF}}(\mu) - C N^{d\beta} P^{\mathrm{DF}}_N(\Gamma_{\varepsilon}) - CN^{d\beta}P^{\mathrm{DF}}_N(\Xi_{\tau}^c).
		\end{equation}
		It only remains to bound the probability for an empirical measure to be in $\Gamma_{\varepsilon}$ or in $\Xi_{\tau}$. Because of the definition \eqref{equation_definition_gamma_varepsilon}, we have for any $j \leq 4^d n^{2d}$,
		\begin{equation}\label{equation_borne_proba_Gamma}
			P^{\mathrm{DF}}_N(\Gamma_{\varepsilon}) \leq 4^d n^{2d} P^{\mathrm{DF}}_N(\Gamma_{\varepsilon}^j).
		\end{equation}
		Furthermore, by Markov's inequality and Lemmas \ref{lemma_apriori}, \ref{lemma_semi_classical_approximation_energy} and \ref{lemma_first_marginals}, we have
		\begin{align}
			P^{\mathrm{DF}}_N(\Xi_{\tau}^c) &\leq \tau^{-1} \int_{\mathcal{P}(\R^{2d})}\iint_{\R^d\times\R^d} \big(|p|^2 + V(x)\big)\d \mu(x, p) \d P^{\mathrm{DF}}_N(\mu)\notag\\
			&\leq \tau^{-1} \iint_{\R^d\times\R^d} \big(|p|^2 + V(x)\big) \d m_N^{(1)}(x, p)\notag\\
			&\leq \tau^{-1}N^{-1} \Big\langle \Psi_N\Big| \sum_j \big(-\hbar^2 \Delta_j + V(x_j)\big) \Psi_N\Big\rangle + o(1)\notag\\
			&\leq C \tau^{-1} N^{d\beta / 2} + o(1).\label{equation_borne_proba_Xi}
		\end{align}
	\end{proof}
	We can get rid of $w_N$ and replace it by a singular interaction, that is informally, we want to replace $w_N$ by $I_w\delta_0$. Let us introduce this singular energy.
	\begin{definition}[Singular semi-classical energy]
		Let $\nu \in L^{1}(\R^{2d})$ be a positive density of norm 1. We define its singular semi-classical energy by
		\begin{equation}\label{equation_definition_singular_semi_classical_energy}
			\mathcal{E}_{\infty}[\nu] =  \iint_{\R^d\times\R^d}\big(|p|^2 + V(x)\big)\nu(x, p)\d x\d p -  I_w \int_{\R^d} \rho_{(2\pi)^d\nu}(x)^2 \d x
		\end{equation}
		with
		\begin{equation}
			\rho_{(2\pi)^d\nu}(x) = \int_{\R^d} \nu(x, p)\d p.
		\end{equation}
		Note that 
		\begin{equation}\label{equation_egalite_energie_semi_classique_singuliere_vlasov}
			\mathcal{E}_{\infty}[\nu] = \mathcal{E}^{\mathrm{V}}\big[(2\pi)^d\nu\big].
		\end{equation}
	\end{definition}
	Now, we want to bound from below the energy obtained in Lemma \ref{lemme_restriction_measures_bornees_pauli} by using the averaged measures and replacing $\mathcal{E}_N$ by $\mathcal{E}_{\infty}$. Note that because of its definition \eqref{equation_definition_averaged_measure}, an averaged measure is always in $L^1$ and we can then always consider its singular semi-classical energy.
	\begin{lemma}\label{lemma_singular_or_not_energy}
		Let $w$ satisfy Assumption \ref{assumption_w}. Let $P_N^{\mathrm{DF}}$ be the Diaconis-Freedman measure defined by \eqref{equation_definition_diaconis_freedman} with the convention \eqref{equation_definition_m_N}, $\mathcal{E}_N$ and $\mathcal{E}_{\infty}$ the semi-classical energies defined by \eqref{equation_definition_mathcalE_N} and \eqref{equation_definition_singular_semi_classical_energy}, and $\Gamma_{\varepsilon}$ and $\Xi_{\tau}$ be the sets of measures defined by \eqref{equation_definition_gamma_varepsilon} and \eqref{equation_definition_Xitau}. Then, we have
		\begin{equation}
			\int_{\Xi_{\tau}\setminus\Gamma_{\varepsilon}} \mathcal{E}_N[\mu]\d P_N^{\mathrm{DF}}(\mu) \geq \big(1 - o(1)\big)\int_{\Xi_{\tau}\setminus\Gamma_{\varepsilon}} \mathcal{E}_{\infty}[\bar{\mu}]\d P_N^{\mathrm{DF}}(\mu) - o(1) - \frac{C \tau^2 N^{d\beta}}{\inf(L^4, L^{2s})}.
		\end{equation}
	\end{lemma}
	\begin{proof} 
		Let $\mu \in \Xi_{\tau} \setminus \Gamma_{\varepsilon}$. By Young's inequality, we have
		\begin{multline}
			\iint_{\R^{2d}\times\R^{2d}} w_N(x - y)\d \bar{\mu}^{\otimes 2}(x, y; p, q) = \int_{\R^d} (w_N * \rho_{(2\pi)^{d}\bar{\mu}})\rho_{(2\pi)^{d}\bar{\mu}} \\ \leq \|w_N\|_{L^1} \|\rho_{(2\pi)^{d}\bar{\mu}}\|_{L^2}^2 = I_w \int_{\R^d} \rho_{(2\pi)^{d}\bar{\mu}}(x)^2 \d x.
		\end{multline}
		Therefore, the semi-classical energy is greater than the singular semi-classical energy:
		\begin{equation}\label{equation_borne_energie_semi_classique_singuliere_ou_non}
			\mathcal{E}_N[\bar{\mu}] \geq \mathcal{E}_{\infty}[\bar{\mu}].
		\end{equation}
		Equation \eqref{equation_borne_energie_semi_classique_singuliere_ou_non} together with Lemma \ref{lemma_averaging} concludes the proof.
	\end{proof}
	In order to prove Proposition \ref{proposition_lower_bound}, we have to prove that the singular semi-classical energy is bounded from below by the Thomas-Fermi energy. Before we prove this, we introduce the set of probability measures whose mass outside of $S_L$ -- defined in \eqref{equation_definition_SL} -- is small.
	\begin{definition}[Measures almost compactly supported]
		Let
		\begin{equation}\label{equation_definition_Theta_eta}
			\Theta_{\eta} = \bigg\{\mu \in \mathcal{P}(\R^{2d}),~ \iint_{S_L^c} \d \mu < \eta\bigg\}.
		\end{equation}
	\end{definition}
	With this notation, we can bound the singular semi-classical energy from below.
	\begin{lemma}[Comparison to the Thomas-Fermi energy]\label{lemma_energie_singuliere_bornee_thomas_fermi}
		Let $\Theta_{\eta}$, $\Gamma_{\varepsilon}$, $\mathcal{E}_{\infty}$ and $E^{\mathrm{TF}}$ be defined respectively by \eqref{equation_definition_Theta_eta}, \eqref{equation_definition_gamma_varepsilon}, \eqref{equation_definition_singular_semi_classical_energy} and \eqref{equation_definition_energie_thomas_fermi}, and let $\mu \in \Theta_{\eta} \setminus \Gamma_{\varepsilon}$, and $w$ satisfy Assumption \ref{assumption_w}. There exists a constant $C>0$ that does not depend on $\varepsilon$ or $\eta$ such that
		\begin{equation}
			\mathcal{E}_{\infty}[\bar{\mu}] \geq E^{\mathrm{TF}} - C(\varepsilon + \eta).
		\end{equation}
	\end{lemma}
	\begin{proof}
		Since the averaged measure $\bar{\mu}$ satisfies the approximated local Pauli principle, i.e. $\bar{\mu} \in \Gamma_{\varepsilon}^c$,
		it is clear that 
		\begin{equation}
			\nu := (1 + \varepsilon)^{-1}\bar{\mu}
		\end{equation} 
		satisfies the true Pauli principle, i.e.
		\begin{equation}
			\nu \leq \frac{1}{(2\pi)^d}.
		\end{equation}
		Let 
		\begin{equation}
			\rho := \int_{\R^d} \nu(\cdot, p) \d p = \rho_{(2\pi)^d\nu}.
		\end{equation} 
		By \eqref{equation_egalite_energie_semi_classique_singuliere_vlasov} and Remark \ref{remark_link_vlasov_thomas_fermi}, we have
		\begin{align}
			\mathcal{E}_{\infty}[\bar{\mu}] = \mathcal{E}_{\infty}\big[(1 + \varepsilon)\nu\big] &= (1 + \varepsilon)\iint_{\R^d\times\R^d}\big(|p|^2 + V(x)\big)\nu(x, p)\d x \d p - (1 + \varepsilon)^2 I_w \int_{\R^d} \rho(x)^2 \d x\notag\\
			&\geq (1 + \varepsilon)\mathcal{E}_{\infty}[\nu] \geq \mathcal{E}_{\infty}[\nu] = \mathcal{E}^{\mathrm{V}}\big[(2\pi)^d\nu\big] \geq \mathcal{E}^{\mathrm{TF}}[\rho].
		\end{align}
		To conclude, we have to bound $\mathcal{E}^{\mathrm{TF}}[\rho]$ from below. We have
		\begin{equation}\label{equation_definition_alpha}
			\int_{\R^d} \rho = \frac{1}{1 + \varepsilon}\iint_{S_L}\mu =: 1 - \alpha.
		\end{equation}
		Let us now denote by $\widetilde{\rho}$ the normed density proportional to $\rho$, i.e.
		\begin{equation}
			\widetilde{\rho} := \frac{1}{1 - \alpha}\rho.
		\end{equation}
		Then, since
		\begin{equation}\label{equation_integrale_rho_tilde_egale_1}
			\int_{\R^d} \widetilde{\rho} = 1,
		\end{equation}
		its energy is bounded from below by the Thomas-Fermi energy:
		\begin{equation}
			\mathcal{E}^{\mathrm{TF}}[\widetilde{\rho}] \geq E^{\mathrm{TF}}.
		\end{equation}
		Let us come back to $\rho$. Its energy can be written
		\begin{equation}
			\mathcal{E}^{\mathrm{TF}}[\rho] = (1 - \alpha)^{1 + 2/d}c_{\mathrm{TF}} \int_{\R^d} \widetilde{\rho}^{1 + 2/d} + (1 - \alpha)\int_{\R^d} V \widetilde{\rho} - (1 - \alpha)^2 I_w \int_{\R^d} \widetilde{\rho}^2.
		\end{equation}
		If $d = 2$, we have directly
		\begin{equation}
			\mathcal{E}^{\mathrm{TF}}[\rho] \geq (1-\alpha)^2\mathcal{E}^{\mathrm{TF}}[\widetilde{\rho}]\geq (1 - 2\alpha)\mathcal{E}^{\mathrm{TF}}[\widetilde{\rho}].
		\end{equation}
		If $d = 1$, the computations are a bit more tedious. We have
		\begin{equation}
			\mathcal{E}^{\mathrm{TF}}[\rho] \geq (1 - 3\alpha)c_{\mathrm{TF}}\int_{\R^d}\widetilde{\rho}^3 + (1 - 3\alpha)\int_{\R^d} V \widetilde{\rho} - (1 - \alpha)I_w\int_{\R^d} \widetilde{\rho}^2 = (1 - 3\alpha)\mathcal{E}^{\mathrm{TF}}[\widetilde{\rho}] - 2 \alpha I_w\int_{\R^d} \widetilde{\rho}^2.
		\end{equation}
		By Young's inequality and \eqref{equation_integrale_rho_tilde_egale_1}, 
		\begin{equation}
			2\alpha I_w \int_{\R^d} \widetilde{\rho}^2 \leq \alpha c_{\mathrm{TF}}\int_{\R^d} \widetilde{\rho}^3 + \alpha\frac{I_w^2}{c_{\mathrm{TF}}} 
		\end{equation}
		and therefore
		\begin{equation}
			\mathcal{E}^{\mathrm{TF}}[\rho] \geq (1 - 4\alpha)\mathcal{E}^{\mathrm{TF}}[\widetilde{\rho}] - C\alpha . 
		\end{equation}
		Therefore, for both values of $d$, we have
		\begin{equation}
			\mathcal{E}^{\mathrm{TF}}[\rho] \geq (1 - 4\alpha)E^{\mathrm{TF}} - C.
		\end{equation}
		Because of \eqref{equation_definition_alpha} and $\mu \in \Theta_{\eta}$, we have
		\begin{equation}
			1 - \alpha > \frac{1 - \eta}{1 + \varepsilon} > 1 - \eta - \varepsilon, 
		\end{equation}
		and hence
		\begin{equation}
			\mathcal{E}^{\mathrm{TF}}[\rho] \geq (1 - 4\varepsilon - 4\eta)E^{\mathrm{TF}} - C(\varepsilon + \eta).
		\end{equation}
	\end{proof}
	We can now use the several lemmas we have established to bound the energy from below.
	\begin{proof}[Proof of Proposition \ref{proposition_lower_bound}]
		First, by Lemmas \ref{lemme_restriction_measures_bornees_pauli} and \ref{lemma_singular_or_not_energy}, we have the following lower bound
		\begin{multline}
			\int_{\mathcal{P}(\R^{2d})} \mathcal{E}_N[\mu]\d P_N^{\mathrm{DF}}(\mu) \geq \big(1 - o(1)\big)\int_{\Xi_{\tau}\setminus\Gamma_{\varepsilon}} \mathcal{E}_{\infty}[\bar{\mu}]\d P_N^{\mathrm{DF}}(\mu) - o(1)  \\- C \bigg(\frac{\tau^2 N^{d\beta}}{\inf(L^4, L^{2s})} + N^{d \beta} n^{2d}P^{\mathrm{DF}}_N(\Gamma_{\varepsilon}^j) + N^{3d\beta/2}\tau^{-1}\bigg).
		\end{multline}
		To use Lemma \ref{lemma_energie_singuliere_bornee_thomas_fermi}, we notice that
		\begin{equation}\label{equation_inclusion_xitau_thetaeta}
			\xi_{\tau} \subset \Theta_{\eta}~~~~~\mathrm{with}~~~~~ \eta = \frac{\tau}{\inf(L^2, L^s)}
		\end{equation}
		and therefore
		\begin{multline}
			\int_{\mathcal{P}(\R^{2d})} \mathcal{E}_N[\mu]\d P_N^{\mathrm{DF}}(\mu) \geq \big(1 - o(1)\big) E^{\mathrm{TF}} - o(1)  - C \bigg(\varepsilon + \frac{\tau}{\inf(L^2, L^s)} +\frac{\tau^2 N^{d\beta}}{\inf(L^4, L^{2s})} \\+ N^{d \beta} n^{2d}P^{\mathrm{DF}}_N(\Gamma_{\varepsilon}^j) + N^{3d\beta/2}\tau^{-1}\bigg).
		\end{multline}
		We now fix $\varepsilon>0$ and take the following choice of parameters\footnote{This choice does not optimize the errors.} (recall that $s$ is fixed in Assumption \ref{assumption_V})
		\begin{equation}\label{equation_scaling}
			\begin{cases}
				\tau = N^{3\beta d/2 + \sigma}\\
				\gamma = 1 - d\sigma\\
				\delta = \sigma/4 \\
				L \sim N^{2d\beta + 1}\\
				l_p = N^{- \sigma}\\
				l_x = N^{-\frac{1}{d} + 2\sigma}\\
				0< \sigma < \frac{1}{2d} - \frac{d+1}{2}\beta,
			\end{cases}
		\end{equation}
		which are compatible with \eqref{equation_volume_hyperrectangle}, \eqref{equation_condition_delta} and \eqref{equation_condition_lx_lp}. Then, we have immediately 
		\begin{equation}
			\frac{\tau}{\inf(L^2, L^s)} +\frac{\tau^2 N^{d\beta}}{\inf(L^4, L^{2s})} + N^{3d\beta/2}\tau^{-1} \ll 1,
		\end{equation}
		and by Proposition \ref{proposition_proba_violating_pauli_principle},
		\begin{equation}
			N^{d \beta} n^{2d}P^{\mathrm{DF}}_N(\Gamma_{\varepsilon}^j) \ll 1.
		\end{equation}
		Therefore, 
		\begin{equation}
			\liminf \int_{\mathcal{P}(\R^{2d})} \mathcal{E}_N[\mu]\d P_N^{\mathrm{DF}}(\mu) \geq E^{\mathrm{TF}} - C\varepsilon.
		\end{equation}
		We can now take the limit $\varepsilon \to 0$ to recover
		\begin{equation}
			\liminf \int_{\mathcal{P}(\R^{2d})} \mathcal{E}_N[\mu]\d P_N^{\mathrm{DF}}(\mu) \geq E^{\mathrm{TF}},
		\end{equation}
		and by Lemma \ref{lemma_energy_involving_diaconis_freedman},
		\begin{equation}
			\frac{1}{N}\langle \Psi_N|H_N\Psi_N\rangle \geq E^{\mathrm{TF}} - o(1).
		\end{equation}
	\end{proof}
	\section{Convergence of states}\label{section_convergence_states}
	In this section, we prove Theorem \ref{theorem2}. To do so, we use the upper bound along with Lemmas \ref{lemma_energy_involving_diaconis_freedman}, \ref{lemme_restriction_measures_bornees_pauli} and \ref{lemma_singular_or_not_energy}, and the scaling \eqref{equation_scaling}, to find
	\begin{equation}\label{equation_premiere_derniere_section}
		E^{\mathrm{TF}} \geq \int_{\Xi_{\tau} \setminus \Gamma_{\varepsilon}} \mathcal{E}_{\infty}[\bar{\mu}] \d P_N^{\mathrm{DF}}(\mu) - o(1).
	\end{equation}
	We show that $P_N^{\mathrm{DF}}$ converges to a measure $P^{\mathrm{DF}}$ concentrated on probability measures satisfying a Pauli principle, and that we can replace $\bar{\mu}$ by $\mu$ to find
	\begin{equation}
		E^{\mathrm{TF}} \geq \int_{\mathcal{P}(\R^{2d})} \sigma^{\otimes k} \d P^{\mathrm{DF}}(\sigma).
	\end{equation}
	To prove that $P_N^{\mathrm{DF}}$ converges, we need the following lemma.
	\begin{lemma}[Tightness of the one-body Husimi functions]\label{lemma_m_tight}
		Let $V$ and $w$ satisfy Assumptions \ref{assumption_V} and \ref{assumption_w}. Let $\Psi_N$ such that 
		\begin{equation}
			\langle \Psi_N  |H_N \Psi_N\rangle  = O(N),
		\end{equation}
		then $(m_{\Psi_N}^{(1)})$ is tight.
	\end{lemma}
	From this lemma, we can deduce 
	\begin{corollary}[Convergence of the Diaconis-Freedman measure]\label{corollary_convergence_P}
		Let $\Psi_N$ such that 
		\begin{equation}
			\langle \Psi_N  |H_N \Psi_N\rangle  = O(N),
		\end{equation}
		then there exists a limit measure $P^{\mathrm{DF}}$ concentrated on measures that satisfy the Pauli principle such that $P_N^{\mathrm{DF}}$ converges to $P^{\mathrm{DF}}$ weakly as measures. Moreover, we have for all $k \in \N$,
		\begin{equation}
			(2\pi)^{-dk}m_{\Psi_N}^{(k)} \rightharpoonup \int_{\mathcal{P}(\R^{2d})} \sigma^{\otimes k} \d P^{\mathrm{DF}}(\sigma)
		\end{equation}
		weakly as measures when $N\to+\infty$.
	\end{corollary} 
	\begin{definition}[Wasserstein distance]
		The $1$-Wasserstein distance, that metricizes the weak convergence on $\mathcal{P}(\R^{2d})$, is defined by
		\begin{equation}
			d_1^\mathrm{W}(\mu, \nu) = \sup_{\|\phi\|_{\mathrm{Lip}}\leq 1} \bigg|\int \phi\d \mu - \int \phi \d \nu\bigg|.
		\end{equation}
	\end{definition}
	\begin{proof}[Proof of Corollary \ref{corollary_convergence_P}]
		Because of Theorem \ref{theorem_diaconis_freedman}, Lemma \ref{lemma_m_tight} implies that there exists a limit measure $P^{\mathrm{DF}}$ such that $P_N^{\mathrm{DF}}$ converges weakly to $P^{\mathrm{DF}}$. Then, since the $m_N^{(k)}$ satisfy a Pauli principle, $P^{\mathrm{DF}}$ is concentrated on measures $\mu$ satisfying
		\begin{equation}\label{equation_pauli_principle_P}
			\mu \leq (2\pi)^{-d}.
		\end{equation}
		Moreover, by Theorem \ref{theorem_diaconis_freedman}, we have
		\begin{equation}\label{equation_convergence_m_tilde_corollaire}
			\widetilde{m}_N^{(k)} \rightharpoonup \int_{\mathcal{P}(\R^{2d})} \sigma^{\otimes k} \d P^{\mathrm{DF}}(\sigma).
		\end{equation}
		Besides, \eqref{equation_difference_m_tilde_ou_non_TV} and \eqref{equation_definition_m_N_k} imply that
		\begin{equation}
			\big\|(2\pi)^{-dk} m_{\Psi_N}^{(k)} - \widetilde{m}_N^{(k)}\big\|_{\mathrm{TV}} \leq C \frac{(k - 1)^2}{N}\|m_{\Psi_N}^{(k)}\|_{\mathrm{TV}} + \frac{2k(k-1)}{N}\leq C \frac{k^2}{N}.
		\end{equation}
		In particular, we have
		\begin{equation}
			d_1^{\mathrm{W}}\Big((2\pi)^{-dk} m_{\Psi_N}^{(k)}, \widetilde{m}_N^{(k)}\Big) \to 0,
		\end{equation}
		which together with \eqref{equation_convergence_m_tilde_corollaire} gives
		\begin{equation}
			(2\pi)^{-dk} m_{\Psi_N}^{(k)} \rightharpoonup \int_{\mathcal{P}(\R^{2d})} \sigma^{\otimes k} \d P^{\mathrm{DF}}(\sigma).
		\end{equation}
	\end{proof}
	Now that we have proved Corollary \ref{corollary_convergence_P}, we have to prove Lemma \ref{lemma_m_tight}.
	\begin{proof}[Proof of Lemma \ref{lemma_m_tight}]
		Let us treat the case $d=2$ first. Let $\varepsilon > 0$ and $\widetilde{w} = (1 + \varepsilon)w$. Then, we can rewrite Lemma \ref{lemma_semi_classical_approximation_energy}:
		\begin{multline} 
			(1 + \varepsilon) E^{\mathrm{TF}} \geq (1 + \varepsilon)\frac{\langle \Psi_N|H_N\Psi_N\rangle}{N} = \frac{1 + \varepsilon}{(2\pi)^d}\iint_{\R^d\times\R^d} \big(|p|^2 + V(x)\big)\d m_{\Psi_N}^{(1)}(x, p) \\ - \frac{1}{2(2\pi)^{2d}}\iint_{\R^{2d}\times\R^{2d}} \widetilde{w}_N(x-y)\d m_{\Psi_N}^{(2)}(x, y; p, q) + o(1).
		\end{multline}
		If we take $\varepsilon$ small enough, namely
		\begin{equation}
			\varepsilon < \frac{c_{\mathrm{TF}}}{I_w} - 1,
		\end{equation} 
		we have
		\begin{equation}
			\frac{1}{(2\pi)^d}\iint_{\R^d\times\R^d} \big(|p|^2 + V(x)\big)\d m_{\Psi_N}^{(1)}(x, p) - \frac{1}{2(2\pi)^{2d}}\iint_{\R^{2d}\times\R^{2d}} \widetilde{w}_N(x-y)\d m_{\Psi_N}^{(2)}(x, y; p, q) \geq - o(1).
		\end{equation}
		Therefore, we have
		\begin{equation}
			\iint_{\R^d\times\R^d} \big(|p|^2 + V(x)\big) \d m_{\Psi_N}^{(1)}(x, p) \leq C\varepsilon^{-1},
		\end{equation}
		and by coercivity of $(x, p) \mapsto |p|^2 + V(x)$, $m_{\Psi_N}^{(1)}$ is indeed tight.\\
		
		Now let us consider the case $d=1$, and $\varepsilon = 1$. Here, we can apply Proposition \ref{proposition_lower_bound} to $2w$ and find
		\begin{equation}
			\frac{1}{(2\pi)^d} \iint_{\R^d\times\R^d} \big(|p|^2 + V(x)\big) \d m^{(1)}_{\Psi_N}(x, p) - \frac{1}{(2\pi)^{2d}} \iint_{\R^{2d}\times\R^{2d}} w_N(x-y)\d m^{(2)}_{\Psi_N}(x, y; p, q) \geq -C 
		\end{equation}
		and thus
		\begin{equation}
			\frac{1}{(2\pi)^d} \iint_{\R^d\times\R^d} \big(|p|^2 + V(x)\big) \d m^{(1)}_{\Psi_N}(x, p) \leq 2E^{\mathrm{TF}} + C,
		\end{equation}
		which implies the tightness of $m^{(1)}_{\Psi_N}$.
	\end{proof}
	To prove Theorem \ref{theorem2}, let us first state Lemma \ref{lemma_convergence_averaged_measure}, whose proof can be found in \cite[Lemma 4.10]{girardot_semiclassical_2021}.
	\begin{lemma}[Convergence of the averaged measure]\label{lemma_convergence_averaged_measure}
		Let $\mu \in \mathcal{P}(\R^d)$, we have the weak convergence of measures
		\begin{equation}
			\bar{\mu} = \sum_{j = 1}^{4^d n^{2d}} \mathds{1}_{\Omega_j}\iint_{\Omega_j}\frac{\d \mu}{|\Omega_j|} \rightharpoonup \mu
		\end{equation}
		when $N, L \to + \infty$, uniformly in $\mu$ i.e.
		\begin{equation}
			d_1^{\mathrm{W}}(\bar{\mu}, \mu) = o(1).
		\end{equation}
	\end{lemma}
	\begin{corollary}[Convergence of an averaged sequence]\label{corollary_concergence_averaged_sequence}
		Let $(\mu_N)$ be a sequence of probability measures converging weakly to a probability measure $\mu$. Then, we have
		\begin{equation}
			\overline{\mu_N} \rightharpoonup \mu
		\end{equation}
		weakly as measures.
	\end{corollary}
	\begin{proof}
		Because of Lemma \ref{lemma_convergence_averaged_measure}, we have 
		\begin{equation}
			d_1^{\mathrm{W}}(\overline{\mu_N}, \mu) \leq d_1^{\mathrm{W}}(\overline{\mu_N}, \mu_N) + d_1^{\mathrm{W}}(\mu_N, \mu) = o(1).
		\end{equation}
	\end{proof}
	\begin{definition}[Modified Diaconis-Freedman measure]\label{definition_modified_diaconis_freedman_measure}
		Let 
		\begin{equation}\label{equation_definition_qn}
			\d Q_N = \mathds{1}_{\Xi_{\tau}} \mathds{1}_{\Gamma_{\varepsilon}^c}\d P^{\mathrm{DF}}_N.
		\end{equation}
	\end{definition}
	\begin{lemma}[Direct properties of $Q_N$]\label{lemma_properties_qn}
		We have
		\begin{equation}\label{equation_etf_qn}
			E^{\mathrm{TF}} \geq \int_{\mathcal{P}(\R^{2d})} \mathcal{E}_{\infty}[\bar{\mu}]\d Q_N - o(1)
		\end{equation}
		and the following weak convergence
		\begin{equation}\label{equation_wcv_qn}
			Q_N \rightharpoonup P^{\mathrm{DF}}.
		\end{equation}
	\end{lemma}
	\begin{proof}
		Equation \eqref{equation_etf_qn} is just \eqref{equation_premiere_derniere_section} with our new notation \eqref{equation_definition_qn}, and \eqref{equation_wcv_qn} comes from \eqref{equation_borne_proba_Gamma}, \eqref{equation_borne_proba_Xi}  and Corollary \ref{corollary_convergence_P}.
	\end{proof}
	To go from $Q_N$ to $P^{\mathrm{DF}}$, we will need the following lemma, that is a consequence of Fatou's lemma for weakly converging probabilities \cite[Theorem 1.1]{Lemme_Fatou} and \eqref{equation_wcv_qn}:
	\begin{lemma}[Fatou's lemma for weakly converging probabilities]\label{lemma_fatou_bizarre}
		Let $\mathcal{E}$ be a measurable positive function on $\big(\mathcal{P}(\R^d), \mathcal{A}\big)$, then, for $\Lambda \in \mathcal{A}$, 
		\begin{equation}
			\liminf_{N\to+\infty}\frac{1}{Q_N(\Lambda)} \int_{\Lambda} \mathcal{E}\big[\overline{\mu}\big]\d Q_N(\mu) \geq \frac{1}{P^{\mathrm{DF}}(\Lambda)}\int_{\Lambda}\liminf_{N\to + \infty, \nu\rightharpoonup\mu}\mathcal{E}\big[\overline{\nu}\big]\d P^{\mathrm{DF}}(\mu).
		\end{equation}
		In particular, if $\mathcal{E}$ is lower semi-continuous for the weak convergence of measures, we have
		\begin{equation}
			\liminf_{N\to+\infty}\frac{1}{Q_N(\Lambda)} \int_{\Lambda} \mathcal{E}\big[\overline{\mu}\big]\d Q_N(\mu) \geq \frac{1}{P^{\mathrm{DF}}(\Lambda)} \int_{\Lambda}\mathcal{E}[\mu]\d P^{\mathrm{DF}}(\mu).
		\end{equation}
	\end{lemma}
	\begin{lemma}[The Diaconis-Freedman measure charges minimizers of Vlasov energy, $d=2$]\label{lemma_limit_df_fatou_2}
		Let $d=2$. Let $V$ and $w$ satisfy Assumptions \ref{assumption_V} and \ref{assumption_w}. The limit Diaconis-Freedman measure $P^{\mathrm{DF}}$ only charges $(2\pi)^2\mu$ where $\mu$ is the minimizer of the Vlasov energy $\mathcal{E}^{\mathrm{V}}$.
	\end{lemma}
	\begin{proof}
		Let $\varepsilon > 0$ and $\mu\in \Gamma_{\varepsilon}^c$, we have
		\begin{equation}
			(1+\varepsilon)\iint_{\R^2\times\R^2} |p|^2 \d \overline{\mu} \geq c_{\mathrm{TF}}\int_{\R^2}\rho_{(2\pi)^2\overline{\mu}}^2,
		\end{equation}
		and hence
		\begin{equation}
			\mathcal{E}_{\infty}[\overline{\mu}] \geq \iint_{\R^2\times\R^2} \Bigg(\bigg(1 - \frac{I_w}{(1+\varepsilon)c_{\mathrm{TF}}}\bigg)|p|^2+V(x)\Bigg)\d \overline{\mu} =: \mathcal{E}^{\varepsilon}\big[\overline{\mu}\big].
		\end{equation}
		For $\varepsilon$ small enough, we have
		\begin{equation}
			1 - \frac{I_w}{(1+\varepsilon)c_{\mathrm{TF}}} > 0,
		\end{equation}
		and then $\mathcal{E}$ is a positive and weakly lower semi-continuous. Hence, by Lemma \ref{lemma_properties_qn} and Lemma \ref{lemma_fatou_bizarre}, we have
		\begin{equation}
			E^{\mathrm{TF}} \geq \liminf_{N\to+\infty} \int_{\mathcal{P}(\R^2)} \mathcal{E}_{\infty}\big[\overline{\mu}\big]\d Q_N(\mu) \geq \liminf_{N\to+\infty} \int_{\mathcal{P}(\R^2)} \mathcal{E}^{\varepsilon}\big[\overline{\mu}\big]\d Q_N(\mu) \geq \int_{\mathcal{P}(\R^2)}\mathcal{E}^{\varepsilon}[\mu]\d P^{\mathrm{DF}}(\mu).
		\end{equation}
		Then, we can take the limit $\varepsilon \to 0$ to find
		\begin{equation}
			E^{\mathrm{TF}} \geq \int_{\mathcal{P}(\R^2)}\mathcal{E}^{0}[\mu]\d P^{\mathrm{DF}}(\mu).
		\end{equation}
		Since $P^{\mathrm{DF}}$ only charges probability measures that satisfy the Pauli principle \eqref{equation_pauli_principle_P}, we have
		\begin{equation}\label{equation_d_2_a_la_fin}
			E^{\mathrm{TF}}\geq \int_{\mathcal{P}(\R^2)} \iint_{\R^d\times\R^d} \Bigg(\bigg(1 - \frac{I_w}{c_{\mathrm{TF}}}\bigg)|p|^2+V(x)\Bigg)\d \mu\d P^{\mathrm{DF}}(\mu)\geq \int_{\mathcal{P}(\R^2)}\mathcal{E}^{\mathrm{TF}}\big[(2\pi)^2\rho_{\mu}\big]\geq E^{\mathrm{TF}}.
		\end{equation}
		Thus, we have equality in all the inequalities of \eqref{equation_d_2_a_la_fin}. In particular, $P^{\mathrm{DF}}$ only charges the measre $\mu$ such that $(2\pi)^d\rho_{\mu}$ is the unique minimizer of $\mathcal{E}^{\mathrm{TF}}$, and such that 
		\begin{equation}
			\mu(x, p) = (2\pi)^{-2}\mathds{1}_{|p|\leq \sqrt{4\pi\rho(x)}}
		\end{equation}
		with $\rho$ this minimizer of $\mathcal{E}^{\mathrm{TF}}$. Then $(2\pi)^2\mu$ is the minimizer of the Vlasov energy $\mathcal{E}^{\mathrm{V}}$.
	\end{proof}
	\begin{lemma}[The Diaconis-Freedman measure charges minimizers of Vlasov energy, $d=1$]\label{lemma_limit_df_fatou_1}
		Let $d=1$. Let $V$ and $w$ satisfy Assumptions \ref{assumption_V}, \ref{assumption_V2} and \ref{assumption_w}. The limit Diaconis-Freedman measure $P^{\mathrm{DF}}$ only charges probability measures $\mu$ such that $(2\pi) \mu$ is a minimizer of the Vlasov energy $\mathcal{E}^{\mathrm{V}}$.
	\end{lemma}
	To prove Lemma \ref{lemma_limit_df_fatou_1}, we have to introduce some notation.
	\begin{definition}[Relaxed energy]\label{definition_relaxed_energy_section4}
		For $\eta \in [0, 1]$, let  
		\begin{equation}
			e_{\eta}(x) = (1-\eta) c_{\mathrm{TF}}x^3 - I_w x^2 + \alpha x
		\end{equation}
		with $\alpha > 0$ such that $e_{\eta}$ has exactly two minimizers on $\R_+$, $0$ and $\rho_{\alpha}\in \R_+^*$, with minimum $0$ (see Definition \ref{definition_local_energy_appendix} and Lemma \ref{lemma_min_ealpha} for more detail). We set
		\begin{equation}
			J_{\eta}(x) = \mathds{1}_{x\geq\rho_{\alpha}}e_{\eta}(x) \leq e_{\eta}(x)~~~~~\mathrm{and}~~~~~\mathcal{J}_{\eta}[\rho] = \int_{\R^d} J_{\eta}(\rho).
		\end{equation}
	\end{definition}
	We use relaxed functionals to recover some weak lower semi-continuity. Indeed, the 1D Vlasov functional is not weakly lower semi-continuous itself. See Appendix \ref{section_appendixe_2}.
	\begin{proof}[Proof of Lemma \ref{lemma_limit_df_fatou_1}]
		Let $\varepsilon > 0$ and $\mu \in \Gamma_{\varepsilon}^c$, we have 
		\begin{equation}
			(1+\varepsilon)^2 \iint_{\R\times\R}|p|^2 \d \overline{\mu} \geq c_{\mathrm{TF}}\int_{\R}\rho_{2\pi\overline{\mu}}^3
		\end{equation}
		and hence for $\eta\in(0,1)$ fixed\footnote{One can take $\eta = 1/2$ for instance.}
		\begin{equation}\label{equation_borne_inf_en_separant_moitie_densite_moitie_husimi}
			\mathcal{E}_{\infty}[\overline{\mu}] \geq \mathcal{E}^{\varepsilon, \eta}[\overline{\mu}] + \mathcal{J}_{\eta}[\rho_{(2\pi)^d\overline{\mu}}] - \alpha
		\end{equation}
		with
		\begin{equation}
			\mathcal{E}^{\varepsilon, \eta}[\overline{\mu}] = \iint_{\R^d\times\R^d} \bigg((1-3\varepsilon)\eta |p|^2 + V(x)\bigg)\d \overline{\mu}.
		\end{equation}
		Therefore, by Lemma \ref{lemma_properties_qn}, 
		\begin{equation}
			E^{\mathrm{TF}} \geq \liminf_{N\to+\infty} \int_{\mathcal{P}(\R)} \Big(\mathcal{E}^{\varepsilon, \eta}[\overline{\mu}] + \mathcal{J}_{\eta}[\rho_{2\pi\overline{\mu}}] - \alpha\Big) \d Q_N(\mu).
		\end{equation}
		Let 
		\begin{equation}\label{equation_def_lambdaa}
			\Lambda_A = \big\{\mu\in \mathcal{P}(\R),~\mathcal{E}_{\infty}[\mu] \leq A\big\},
		\end{equation}
		by Markov inequality, we have 
		\begin{equation}
			P_N^{\mathrm{DF}}(\Lambda_A^c) \leq \frac{E^{\mathrm{TF}}}{A} + o(1).
		\end{equation}
		Moreover, as $Q_N \rightharpoonup P^{\mathrm{DF}}$, we have for almost every $A \in \R_+$
		\begin{equation}\label{convergence_QN_LambdaA}
			Q_N(\Lambda_A) \to P^{\mathrm{DF}}(\Lambda_A).
		\end{equation}  
		Hence, using Lemma \ref{lemma_fatou_bizarre} and Lemma \ref{lemme_semi_conitnuite_inferieure_kcal}, we find
		\begin{equation}
			E^{\mathrm{TF}} \geq \liminf_{N\to+\infty} \int_{\Lambda_A} \Big(\mathcal{E}^{\varepsilon, \eta}[\overline{\mu}] + \mathcal{J}_{\eta}[\rho_{2\pi\overline{\mu}}]-\alpha\Big) \d Q_N(\mu) \geq \int_{\Lambda_A} \Big(\mathcal{E}^{\varepsilon, \eta}[\mu] + \mathcal{J}_{\eta}[\rho_{2\pi\mu}]-\alpha\Big) \d P^{\mathrm{DF}}(\mu).
		\end{equation}
		Then, we can take the limit $\varepsilon \to 0$ and $A\to + \infty$ to get
		\begin{equation}\label{equation_quasiment_laè_derniere}
			E^{\mathrm{TF}} \geq \int_{\mathcal{P}(\R)} \Big(\mathcal{E}^{0, \eta}[\mu] + \mathcal{J}_{\eta}[\rho_{2\pi\mu}]-\alpha\Big) \d P^{\mathrm{DF}}(\mu) \geq \int_{\mathcal{P}(\R)} \mathcal{E}^{\mathrm{TF}}_{J_0}\big[\rho_{2\pi\mu}]\d P^{\mathrm{DF}}(\mu) \geq E^{\mathrm{TF}}_{J_0},
		\end{equation}
		where we used the notation introduced in Definition \ref{definition_relaxed_energy_appendice_la}. By Corollary \ref{corollary_same_minimizers}, we know that 
		\begin{equation}
			E^{\mathrm{TF}}_{J_0} = E^{\mathrm{TF}},
		\end{equation}
		and hence  $P^{\mathrm{DF}}$ only charges measures $\mu$ such that $\rho_{2\pi\mu}$ is a minimizer of $\mathcal{E}^{\mathrm{TF}}_{J_0}$, that is a minimizer of $\mathcal{E}^{\mathrm{TF}}$, once again by Corollary \ref{corollary_same_minimizers}. Furthermore, there is equality in the inequalities of \eqref{equation_quasiment_laè_derniere}, and in particular
		\begin{equation}\label{equation_ohlala_lenergie_cest_la_meme_en_vlasov_ou_pas}
			\int_{\mathcal{P}}\iint_{\R\times\R} |p|^2 \d \mu(x, p)\d P^{\mathrm{DF}}(\mu) = \int_{\mathcal{P}}c_{\mathrm{TF}}\int_{\R} \rho_{2\pi\mu}^3\d P^{\mathrm{DF}}(\mu).
		\end{equation}
		Since $P^{\mathrm{DF}}$ only charges measures $\mu$ that satisfy the Pauli principle \eqref{equation_pauli_principle_P}, \eqref{equation_ohlala_lenergie_cest_la_meme_en_vlasov_ou_pas} implies that $P^{\mathrm{DF}}$ only charges probability measures of the form
		\begin{equation}
			\mu(x, p) = \frac{1}{2\pi}\mathds{1}_{|p|\leq \pi \rho(x)}
		\end{equation}
		with $\rho$ a minimizer of $\mathcal{E}^{\mathrm{TF}}$, that is probability measures such that $2\pi\mu$ is a minimizer of $\mathcal{E}^{\mathrm{V}}$.
	\end{proof}
	With those lemmas, Theorem \ref{theorem2} can be deduced directly.
	\begin{proof}[Proof of Theorem \ref{theorem2}]
		Theorem \ref{theorem2} is a consequence of Corollary \ref{corollary_convergence_P}, Lemma \ref{lemma_limit_df_fatou_2} and Lemma \ref{lemma_limit_df_fatou_1} .
	\end{proof}
	
	\section{The non-purely attractive case}\label{section5}
	In this section, we do not intend to give a full proof of Theorem \ref{theorem_repulsif}, with every details, but rather to point out what in the proof differs from that of the purely attractive case.
	
	\subsection{Upper bound on the energy}
	First, let us define the Hartree and $N$-Vlasov energies (these correspond to Definitions \ref{definition_hartree_energy} and \ref{definition_N_Vlasov} in our new context).
	\begin{definition}[Hartree energy]\label{definition_hartree_energy2}
		For $\gamma$ a trace-class operator on $\Hcal$ satisfying 
		\begin{equation}
			0\leq \gamma \leq 1,
		\end{equation}  
		let $\mathcal{E}^{\mathrm{H}}[\gamma]$ be the Hartree energy of $\gamma$ defined by
		\begin{equation}\label{equation_definition_fonctionnelle_energie_hartree2}
			\mathcal{E}^{\mathrm{H}}[\gamma] = \hbar^2 \tr(-\Delta\gamma) + \tr(V\gamma) + \frac{N^{-1}}{2}\int_{\R^d} (w_N * \rho_{\gamma})\rho_{\gamma},
		\end{equation}
		using the notation introduced in Definition \ref{definition_space_momentum_density} for the density of $\gamma$.
	\end{definition}
	
	\begin{definition}[$N$-Vlasov energy]\label{definition_N_Vlasov2}
		We set
		\begin{equation}
			\mathcal{E}_{\infty,N}^{\mathrm{V}}[m] = \frac{1}{(2\pi)^d}\iint_{\R^d \times \R^d} \big(|p|^2 + V(x)\big) m(x, p)\d x\d p + I_{w^+}\int_{\R^d}\rho_m^2  - \frac{1}{2}\int_{\R^{2d}} w_N^-(x - y) \rho_m(x)\rho_m(y)\d x\d y
		\end{equation}
		the $N$-Vlasov energy functional. Note that formally, we have taken $N = \infty$ for the repulsive part (i.e. local interactions), but not the attractive one.
	\end{definition}
	Then, we have the following upper bound, for which we do not need the assumption on the sign of $\widehat{w_+}$.
	\begin{proposition}[Upper bound]\label{proposition_upper_bound2}
		For $d= 1$ or $2$ and for $0 < \beta < 1/d$, let $V$ and $w$ satisfy Assumptions \ref{assumption_V} and \ref{assumption_w}, we have
		\begin{equation}
			E(N) \leq N E^{\mathrm{TF}} + o(N).
		\end{equation}
	\end{proposition}
	
	\begin{proof}
		The proof is essentially the same as that of Proposition \ref{proposition_upper_bound}. First, Proposition \ref{proposition_reduction_hartree_energy} remains true since $w^+ \in L^{\infty}(\R^d)$. This means that for $\gamma_N$ satisfying \eqref{equation_condition_gamma_N},
		\begin{equation}\label{equation_inegalite_energie_fondamentale_energie_hartree}
			E(N) \leq \big(1 + o(1)\big) \mathcal{E}^{\mathrm{H}}[\gamma_N] + o(N).
		\end{equation}
		Moreover, Lemma \ref{lemma_approximation_semi_classique2} remains true if we replace the equality in \eqref{equation_lemme_approximation_semi_classique2} by an inequality. Namely, for $m$ satisfying \eqref{equation_hypotheses_m_lemme_approximation_semiclassique2} and $\gamma^{\hbar}$ defined by \eqref{equation_defintion_gamma_hbar_lemme_approx}, with $\hbar_x$ and $\hbar_p$ defined as in Definition \ref{definition_coherent_state}, we have 
		\begin{equation}\label{equation_inequality_energy_hartree_vlasov_infty_n}
			\mathcal{E}^{\mathrm{H}}[\gamma^{\hbar}] \leq N \mathcal{E}_{\infty, N}^{\mathrm{V}}[m] + o(N).
		\end{equation}
		Indeed, 
		\begin{align}
			\int_{\R^d} (w_N^+ * \rho_{\gamma^{\hbar}})\rho_{\gamma^{\hbar}} &= \int_{\R^d} \Big(w_N^+ * \big(\rho_m * |f^{\hbar}|^2\big)\Big) \big(\rho_m * |f^{\hbar}|^2\big)\notag\\
			&\leq \big\|w_N^+\big\|_{L^1} \big\|\rho_m * |f^{\hbar}|^2\big\|_{L^2}^2 = \int_{\R^d} w^+ \int_{\R^d} \rho_m^2.
		\end{align}
		Now let $m$ be a minimizer of the Vlasov energy $\mathcal{E}^{\mathrm{V}}$ defined by \eqref{equation_definition_vlasov_energy_functional2}, under the constraints stated in \eqref{equation_egalite_energie_thomas_fermi_vlasov}. One can check that $m$ satisfies \eqref{equation_hypotheses_m_lemme_approximation_semiclassique2}. Then,
		\begin{equation}
			\frac{1}{2}\int_{\R^d} (w_N^-*\rho_m)\rho_m \to I_{w^-} \int_{\R^d}\rho_m^2,
		\end{equation}
		so that \eqref{equation_inequality_energy_hartree_vlasov_infty_n} implies
		\begin{equation}
			\mathcal{E}^{\mathrm{H}}[\gamma^{\hbar}] \leq N E^{\mathrm{TF}} + o(N).
		\end{equation}
		Together with \eqref{equation_inegalite_energie_fondamentale_energie_hartree}, this proves that
		\begin{equation}
			E(N) \leq N E^{\mathrm{TF}} + o(N).
		\end{equation}
	\end{proof}
	
	\subsection{Lower bound on the energy}
	We have the following lower bound on the energy.
	\begin{proposition}[Lower bound]\label{proposition_lower_bound2}
		Let $d=1$ or $2$, $\beta < \frac{1}{d(d + 1)}$, $V$ and $w$ satisfy Assumptions \ref{assumption_V} and \ref{assumption_w2}. Let $(\Psi_N)$ be a sequence of normalized approximate minimizers of the energy, i.e.
		\begin{equation}
			\langle \Psi_N|H_N\Psi_N\rangle = E(N) + o(N)~~~~~\mathrm{and}~~~~~\langle \Psi_N|\Psi_N\rangle = 1,
		\end{equation}
		then we have
		\begin{equation}
			\langle\Psi_N |H_N\Psi_N\rangle \geq N E^{\mathrm{TF}} + o(N).
		\end{equation}
		In particular, 
		\begin{equation}
			E(N) \geq N E^{\mathrm{TF}} + o(N).
		\end{equation}
	\end{proposition}
	Most of the intermediary results (including Sections \ref{subsection_a_priori_bounds}--\ref{subsection_averaging}) for the lower bound remain directly true for a non-necessarily negative $w$, either because we only bound the norm of $w$, or because the positive part can be easily thrown away. Then, we just have to keep the $N$-interaction energy for the repulsive part in Lemma \ref{lemma_singular_or_not_energy2} (rather than replacing it with the local one) and prove the convergence of the mixed Thomas-Fermi energy in Lemma \ref{equation_convergence_energie_thomas_fermi} to conclude. Before we can state a proper result, we have to define the semi-classical energy of a one-body measure.
	\begin{definition}[Semi-classical energy of a one-body measure]
		Let $\mu$ be a probability measure on the phase space $\R^{2d}$. We define its semi-classical energy by
		\begin{equation}\label{equation_definition_mathcalE_N2}
			\mathcal{E}_N[\mu] = \iint_{\R^d\times\R^d} \big(|p|^2 + V(x)\big) \d \mu(x, p) + \frac{1}{2} \iint_{\R^{2d}\times\R^{2d}} w_N(x-y) \d \mu^{\otimes 2}(x, p ; y, q).
		\end{equation}
	\end{definition}
	From now on, we fix a sequence of states $\Psi_N \in \Hcal_N$ of norm 1 that approximately minimizes the energy in the following sense
	\begin{equation}
		\langle \Psi_N|H_N\Psi_N\rangle = E(N) + o(N).
	\end{equation} 
	Then, we have the following result.
	\begin{lemma}[Direct adaptation of previous arguments] \label{lemme_restriction_measures_bornees_pauli2}
		Let $P_N^{\mathrm{DF}}$ be the Diaconis-Freedman measure defined by \eqref{equation_definition_diaconis_freedman} with the convention \eqref{equation_definition_m_N}, $\mathcal{E}_N$ the semi-classical energy defined by \eqref{equation_definition_mathcalE_N2}, and $\Gamma_{\varepsilon}$ and $\Xi_{\tau}$ be the sets of measures defined by \eqref{equation_definition_gamma_varepsilon} and \eqref{equation_definition_Xitau}. Then, we have the following lower bound on the energy:
		\begin{align}
			\frac{1}{N}\langle \Psi_N|H_N\Psi_N\rangle &= \int_{\mathcal{P}(\R^{2d})} \mathcal{E}_N[\mu]\d P_N^{\mathrm{DF}}(\mu) + o(1)\notag\\
			&\geq \big(1 - o(1)\big)\int_{\Xi_{\tau}\setminus\Gamma_{\varepsilon}} \mathcal{E}_N[\overline{\mu}]\d P_N^{\mathrm{DF}}(\mu) \notag\\ 
			&~~~~~~~~~~~~~~~~~~~~~~~~~~~- C N^{d \beta} n^{2d}P^{\mathrm{DF}}_N(\Gamma_{\varepsilon}^j) - CN^{3d\beta/2}\tau^{-1} - \frac{C \tau^2 N^{d\beta}}{\inf(L^4, L^{2s})} - o(1).
		\end{align}
	\end{lemma}
	Then let us introduce the partially local semi-classical energy, where the negative part is made local, and the positive part is kept non-local.
	\begin{definition}[Partially local semi-classical energy]\label{definition_partially_singular_semiclassical_energy}
		Let $\nu \in L^{1}(\R^{2d})$ be a positive density of norm 1. We define its (partially) local semi-classical energy by
		\begin{equation}\label{equation_definition_singular_semi_classical_energy2}
			\mathcal{E}_{N, \infty}[\nu] =  \iint_{\R^d\times\R^d}\big(|p|^2 + V(x)\big)\nu(x, p)\d x\d p +\frac{1}{2} \int_{\R^d} \big(w_N^+ * \rho_{(2\pi)^d\nu}\big)\rho_{(2\pi)^d\nu} -  I_{w^-} \int_{\R^d} \rho_{(2\pi)^d\nu}^2 
		\end{equation}
		with
		\begin{equation}
			\rho_{(2\pi)^d\nu}(x) = \int_{\R^d} \nu(x, p)\d p.
		\end{equation}
	\end{definition}
	Now, we want to bound from below the energy obtained in Lemma \ref{lemme_restriction_measures_bornees_pauli2} by using the averaged measures and replacing $\mathcal{E}_N$ by $\mathcal{E}_{N, \infty}$. Note that because of its definition \eqref{equation_definition_averaged_measure}, an averaged measure is always in $L^1$ and we can then always consider its singular semi-classical energy.
	\begin{lemma}[Bounding the non-local energy by the partially local energy]\label{lemma_singular_or_not_energy2}
		Let $w$ satisfy Assumption \ref{assumption_w2}. Let $n$ and $L$ be defined by Notation \ref{notation_hypercube}, $P_N^{\mathrm{DF}}$ be the Diaconis-Freedman measure defined by \eqref{equation_definition_diaconis_freedman} with the convention \eqref{equation_definition_m_N}, $\mathcal{E}_N$ and $\mathcal{E}_{N, \infty}$ the semi-classical energies defined by \eqref{equation_definition_mathcalE_N2} and \eqref{equation_definition_singular_semi_classical_energy2}, and $\Gamma_{\varepsilon}$ and $\Xi_{\tau}$ be the sets of measures defined by \eqref{equation_definition_gamma_varepsilon} and \eqref{equation_definition_Xitau}. Then, we have
		\begin{equation}
			\int_{\Xi_{\tau}\setminus\Gamma_{\varepsilon}} \mathcal{E}_N[\overline{\mu}]\d P_N^{\mathrm{DF}}(\mu) \geq \big(1 - o(1)\big)\int_{\Xi_{\tau}\setminus\Gamma_{\varepsilon}} \mathcal{E}_{N, \infty}[\bar{\mu}]\d P_N^{\mathrm{DF}}(\mu) - o(1).
		\end{equation}
	\end{lemma}
	The proof of this lemma is exactly that of Lemma \ref{lemma_singular_or_not_energy}. We cannot directly adapt Lemma \ref{lemma_energie_singuliere_bornee_thomas_fermi}, so we have to introduce an $N$-dependant Thomas-Fermi energy.
	\begin{definition}[$N$-Thomas-Fermi energy]
		Let 
		\begin{equation}
			\mathcal{E}_N^{\mathrm{TF}}[\rho] = c_{\mathrm{TF}}\int_{\R^d} \rho^{1+2/d} + \int_{\R^d} V\rho + \frac{1}{2}\int_{\R^d}\big(w_N^+ * \rho\big)\rho - I_{w^-}\int_{\R^d} \rho^2
		\end{equation}
	and 
	\begin{equation}\label{equation_definition_energie_thomas_fermi3}
		E_N^{\mathrm{TF}} = \inf \left\{\mathcal{E}_N^{\mathrm{TF}}[\rho],~\rho\geq 0,~\int_{\R^d} \rho =1\right\}.
	\end{equation}
	\end{definition}
	\begin{lemma}[Comparison to the $N$-Thomas-Fermi energy]\label{lemma_energie_singuliere_bornee_thomas_fermi2}
		Let $\Theta_{\eta}$, $\Gamma_{\varepsilon}$, $\mathcal{E}_{N, \infty}$ and $E_N^{\mathrm{TF}}$ be defined respectively by \eqref{equation_definition_Theta_eta}, \eqref{equation_definition_gamma_varepsilon}, \eqref{equation_definition_singular_semi_classical_energy2} and \eqref{equation_definition_energie_thomas_fermi3}, and let $\mu \in \Theta_{\eta} \setminus \Gamma_{\varepsilon}$, and $w$ satisfy Assumption \ref{assumption_w2}. There exists a constant $C>0$ that does not depend on $\varepsilon$ or $\eta$ such that
		\begin{equation}
			\mathcal{E}_{N, \infty}[\bar{\mu}] \geq E_N^{\mathrm{TF}} - C(\varepsilon + \eta).
		\end{equation}
	\end{lemma}
	The proof is the same as that of Lemma \ref{lemma_energie_singuliere_bornee_thomas_fermi}. In order to conclude in the same way as the purely attractive case, we have to check the convergence of $E_N^{\mathrm{TF}}$ to $E^{\mathrm{TF}}$.
	\begin{lemma}[Convergence of the $N$-Thomas-Fermi energy]\label{equation_convergence_energie_thomas_fermi}
		We have 
		\begin{equation}
			\liminf_{N\to + \infty} E_N^{\mathrm{TF}} \geq E^{\mathrm{TF}} + o(1).
		\end{equation}
	\end{lemma}
	\begin{proof}
		Let us treat separately the cases $d=1$ and $d = 2$. \\~\\
		\textbf{Case $d=2$.} Let $\rho_N\in L^1(\R^d)$ be such that\footnote{We could take minimizers of $\mathcal{E}_N^{\mathrm{TF}}$ -- they exist -- but we do not need to, as shown in the proof.}
		\begin{equation}
			\rho_N \geq0,~~~~~~~~\int_{\R^d}\rho_N = 1,~~~~~~~~\mathcal{E}_N^{\mathrm{TF}}[\rho_N] = E_N^{\mathrm{TF}} + o(1).
		\end{equation} 
	Then, the energy of $\rho_N$ is bounded: by Young's inequality for convolutions, we have, for $\rho_{\mathrm{TF}}$ minimizer of 
	\begin{equation}\label{equation_inegalites_alasuite}
		\mathcal{E}_N^{\mathrm{TF}}[\rho_N] \leq \mathcal{E}_N^{\mathrm{TF}}[\rho_{\mathrm{TF}}]+o(1) \leq \mathcal{E}^{\mathrm{TF}}[\rho_{\mathrm{TF}}]+o(1) = E^{\mathrm{TF}} + o(1).
	\end{equation}
	Then, since $c_{\mathrm{TF}}>I_{w^-}$, $(\rho_N)$ is bounded in $L^2(\R^d)$. Moreover,
	\begin{equation}
		\int_{\R^d} V\rho_N \leq C
	\end{equation}
	and hence, up to a subsequence, $\rho_N$ converges weakly in $L^2(\R^d)$ to a positive density $\rho_{\infty}$ of integral 1, and we have
	\begin{equation}\label{equation_liminf_d2_immediat}
		\liminf_{N\to+\infty} \int_{\R^d} \rho_N^2 \geq \int_{\R^d} \rho_{\infty}^2~~~~~\text{and}~~~~~\liminf_{N\to+\infty} \int_{\R^d} V\rho_N \geq \int_{\R^d} V\rho_{\infty}.
	\end{equation}
	Let's assume for the moment that 
	\begin{equation}\label{equation_liminf_d2_pas_immediat}
		\liminf_{N\to + \infty}\frac{1}{2}\int_{\R^d}\big(w_N^+*\rho_N\big)\rho_N \geq I_{w^+}\int_{\R^d}\rho_{\infty}^2.
	\end{equation}
	Then, using \eqref{equation_liminf_d2_immediat} and \eqref{equation_liminf_d2_pas_immediat}, we have
	\begin{equation}
		\liminf \mathcal{E}_N^{\mathrm{TF}}[\rho_N] \geq \mathcal{E}^{\mathrm{TF}}[\rho_{\infty}] \geq E^{\mathrm{TF}}.
	\end{equation}
	In order to conclude the proof, we have to check that \eqref{equation_liminf_d2_pas_immediat} holds true. When $N\to + \infty$, we have the following convergences\footnote{Respectively weak and strong in $L^2(\R^d)$.}
	\begin{equation}
		\rho_N \overset{L^2}{\rightharpoonup}\rho_{\infty}~~~~~~~~\text{and}~~~~~~~~w^+_N * \rho_{\infty}\overset{L^ 2}{\to} \left(\int_{\R^d}w^+\right) \rho_{\infty}.
	\end{equation}
	Thus,
	\begin{equation}\label{equation_interaction_singuliere_limite_interaction_N}
		I_{w^+}\int_{\R^d}\rho_{\infty}^2 = \lim_{N\to + \infty}\frac{1}{2}\int_{\R^d} \big(w^+_N*\rho_{\infty}\big)\rho_N.
	\end{equation}
	For a given $N$, as $\widehat{w^+} \geq 0$, we have
	\begin{equation}\label{equation_pseudo_cauchy_schwarz}
		\int_{\R^d} \big(w^+_N*\rho_{\infty}\big)\rho_N \leq \sqrt{\int_{\R^d} \big(w^+_N*\rho_{\infty}\big)\rho_{\infty}}\sqrt{\int_{\R^d} \big(w^+_N*\rho_N\big)\rho_N}.
	\end{equation}
	Then, \eqref{equation_interaction_singuliere_limite_interaction_N} and \eqref{equation_pseudo_cauchy_schwarz} imply \eqref{equation_liminf_d2_pas_immediat}.\\~\\
	\textbf{Case $d=1$.} In this case, we have to use the relaxed functionals introduced in Appendix \ref{section_appendixe_1}. In particular, we use the notation\footnote{With $w$ replaced by $w^-$.} introduced in Definition \ref{definition_local_energy_appendix}, Lemma \ref{lemma_min_ealpha} and \eqref{equation_def_J}. Let
	\begin{equation}
		\mathcal{E}_{N, J}^{\mathrm{TF}}[\rho] = \mathcal{J}[\rho] + \int_{\R^d} V\rho + \frac{1}{2}\int_{\R^d}\big(w_N^+ * \rho\big)\rho - \alpha\int_{\R^d}\rho := \int_{\R^d} J(\rho) + \int_{\R^d} V\rho + \frac{1}{2}\int_{\R^d}\big(w_N^+ * \rho\big)\rho - \alpha\int_{\R^d}\rho.
	\end{equation} 
	Then, for $\rho \geq 0$ of integral 1, we have
	\begin{equation}\label{equation_rapport_differentes_fonctionnelles}
		\mathcal{E}_{N, J}^{\mathrm{TF}}[\rho]\leq \mathcal{E}_{N}^{\mathrm{TF}}[\rho] \leq \mathcal{E}^{\mathrm{TF}}[\rho].
	\end{equation}
	Now, let $\rho_N \in L^1(\R^d)$ be such that
	\begin{equation}\label{equation}
		\rho_N \geq0,~~~~~~~~\int_{\R^d}\rho_N = 1,~~~~~~~~\mathcal{E}_{N, J}^{\mathrm{TF}}[\rho_N] = E_{N, J}^{\mathrm{TF}} + o(1) := \inf\left\{\mathcal{E}_{N, J}^{\mathrm{TF}}[\rho],~\rho\geq0,~\int_{\R^d}\rho= 1 \right\} + o(1).
	\end{equation} 
	Moreover, let
	\begin{equation}\label{equation_relaxed_eneergy_mixed}
		\mathcal{E}_{J}^{\mathrm{TF}}[\rho] = \mathcal{J}[\rho] + \int_{\R^d} V\rho + I_{w^+}\int_{\R^d}f_{\alpha}(\rho) - \alpha\int_{\R^d}\rho~~~~~\text{and}~~~~~E^{\mathrm{TF}}_J = \inf\bigg\{\mathcal{E}^{\mathrm{TF}}_J[\rho],~\rho\geq 0,~\int_{\R^d}\rho = 1\bigg\}
	\end{equation}
	with
	\begin{equation}
		f_{\alpha}(\rho) = \mathds{1}_{\rho_{\alpha}/2 \leq \rho \leq \rho_{\alpha}}\big(2\rho_{\alpha} \rho - \rho_{\alpha}^2\big) + \mathds{1}_{\rho\geq \rho_{\alpha}^2} \rho^2\leq \rho^2.
	\end{equation}
	We replace $\rho^2$ by $f_{\alpha}(\rho)$, that is still positive and convex, in order to keep the result of Corrolary \ref{corollary_same_minimizers}. Since 
	\begin{equation}
		\mathcal{E}_{N, J}^{\mathrm{TF}}[\rho_N]\leq E_{J}^{\mathrm{TF}},
	\end{equation}
	we have (see the proof of Proposition \ref{proposition_minimizing_problem_relaxed} for details)
	\begin{equation}
		\int_{\R^d}\rho_N^3 + \int_{\R^d} V\rho_N \leq C.
	\end{equation}
	Thus, up to extraction, $\rho_N$ converges weakly in $L^3(\R^d)$ and in $L^2(\R^d)$ to a positive density $\rho_{\infty}$ of integral 1, and
	\begin{equation}\label{equation_liminf_vrho_d1}
		\liminf_{N\to+\infty}\int_{\R^d}V\rho_N \geq \int_{\R^d}V\rho_{\infty}.
	\end{equation}
	Since $\mathcal{J}$ is $L^3$-weakly lower semi continuous, we also have
	\begin{equation}\label{equation_liminf_jcal_d1}
		\liminf_{N\to+\infty}\mathcal{J}[\rho_N] \geq \mathcal{J}[\rho_{\infty}].
	\end{equation}
	Furthermore, when $N\to+\infty$, we have the following convergences 
	\begin{equation}
		\rho_N \overset{L^2}{\rightharpoonup}\rho_{\infty}~~~~~~~~\text{and}~~~~~~~~w^+_N * \rho_{\infty}\overset{L^ 2}{\to} \left(\int_{\R^d}w^+\right) \rho_{\infty},
	\end{equation}
	hence
	\begin{equation}
		I_{w^+}\int_{\R^d}\rho_{\infty}^2 = \lim_{N\to + \infty} \frac{1}{2}\int_{\R^d}\big(w_N^+ * \rho_{\infty}\big)\rho_N \leq \liminf_{N\to+\infty} \sqrt{\frac{1}{2}\int_{\R^d}\big(w_N^+ * \rho_{\infty}\big)\rho_{\infty}}\sqrt{\frac{1}{2}\int_{\R^d}\big(w_N^+ * \rho_{N}\big)\rho_{N}}
	\end{equation}
	and thus
	\begin{equation}\label{equation_liminf_interaction_d1}
		\liminf_{N\to+\infty} \frac{1}{2}\int_{\R^d}\big(w_N^+ * \rho_{N}\big)\rho_{N} \geq I_{w^+}\int_{\R^d}\rho_{\infty}^2\geq I_{w^+}\int_{\R^d}\rho_{\infty}^2\mathds{1}_{\rho\geq \rho_{\alpha}}.
	\end{equation}
	Therefore, using \eqref{equation_rapport_differentes_fonctionnelles}, \eqref{equation_liminf_vrho_d1}, \eqref{equation_liminf_jcal_d1} and \eqref{equation_liminf_interaction_d1}, as well as a direct adaptation of Corollary \ref{corollary_same_minimizers}, we find
	\begin{equation}
		\liminf_{N\to+\infty} E_N^{\mathrm{TF}}\geq\liminf_{N\to+\infty} E_{N, J}^{\mathrm{TF}} = \liminf_{N\to+\infty}\mathcal{E}_{N, J}^{\mathrm{TF}}[\rho_N] \geq \mathcal{E}_J^{\mathrm{TF}}[\rho_{\infty}]\geq E_J^{\mathrm{TF}} = E^{\mathrm{TF}}. 
	\end{equation}
	\begin{remark}[Sign of $\widehat{w^+}$]
		The assumption on the sign of $\widehat{w^+}$ was added in order to prove \eqref{equation_liminf_d2_pas_immediat}, in a fashion inspired by \cite{lieb_thomas-fermi_1977}. 
	\end{remark}
	\end{proof}
	We can finally conclude the proof of the lower bound.
	\begin{proof}[Proof of Proposition \ref{proposition_lower_bound2}]
		Using Lemmas \ref{lemme_restriction_measures_bornees_pauli2}, \ref{lemma_singular_or_not_energy2}, \ref{lemma_energie_singuliere_bornee_thomas_fermi2} and \ref{equation_convergence_energie_thomas_fermi}, as well as the choices of parameters \eqref{equation_inclusion_xitau_thetaeta} and \eqref{equation_scaling}, and taking the limit $\varepsilon\to 0$, we have indeed
		\begin{equation}
			\frac{1}{N}\langle \Psi_N|H_N\Psi_N\rangle \geq E^{\mathrm{TF}} - o(1).
		\end{equation}
	\end{proof}
	
	\subsection{Convergence of states}
	
	Since $w^+ \geq 0$, Lemma \ref{lemma_m_tight} remains true with a repulsive part in the interaction (we can discard it in the proof). Then, Lemma \ref{lemma_convergence_averaged_measure} as well as Corollaries \ref{corollary_convergence_P} and \ref{corollary_concergence_averaged_sequence} do not depend on the interaction $w$. Moreover, Lemma \ref{lemma_properties_qn} is replaced by
	\begin{lemma}[Direct properties of $Q_N$]\label{lemma_properties_qn2}
		Using the notation introduced in Definitions \ref{definition_modified_diaconis_freedman_measure} and \ref{definition_partially_singular_semiclassical_energy}, we have
		\begin{equation}\label{equation_etf_qn2}
			E^{\mathrm{TF}} \geq \int_{\mathcal{P}(\R^{2d})} \mathcal{E}_{N, \infty}[\bar{\mu}]\d Q_N - o(1)
		\end{equation}
		and the following weak convergence
		\begin{equation}\label{equation_wcv_qn2}
			Q_N \rightharpoonup P^{\mathrm{DF}}.
		\end{equation}
	\end{lemma}
	\begin{proof}
		Equation \eqref{equation_etf_qn} comes from Proposition \ref{proposition_upper_bound2}, Lemmas \ref{lemme_restriction_measures_bornees_pauli2} and \ref{lemma_singular_or_not_energy2}, as well as the notation \eqref{equation_definition_qn}. Equation \eqref{equation_wcv_qn} comes from \eqref{equation_borne_proba_Gamma}, \eqref{equation_borne_proba_Xi}  and Corollary \ref{corollary_convergence_P}.
	\end{proof}
	
	\begin{lemma}[The Diaconis-Freedman measure charges minimizers of the Vlasov energy, $d=2$]\label{lemma_limit_df_fatou_2d2}
		Let $d=2$. Let $V$ and $w$ satisfy Assumptions \ref{assumption_V} and \ref{assumption_w2}. The limit Diaconis-Freedman measure $P^{\mathrm{DF}}$ only charges $(2\pi)^2\mu$ where $\mu$ is the minimizer of the Vlasov energy $\mathcal{E}^{\mathrm{V}}$.
	\end{lemma}
	\begin{proof}
		Let $\varepsilon > 0$ and $\mu\in \Gamma_{\varepsilon}^c$, we have
		\begin{equation}\label{equation_lien_deux_energies_cinetiques_epsilon}
			(1+\varepsilon)\iint_{\R^2\times\R^2} |p|^2 \d \overline{\mu} \geq c_{\mathrm{TF}}\int_{\R^2}\rho^2,
		\end{equation}
		with
		\begin{equation}
			\rho := \rho_{(2\pi)^2\overline{\mu}},
		\end{equation}
		and hence
		\begin{equation}
			\mathcal{E}_{N, \infty}[\overline{\mu}] - \frac{1}{2}\int_{\R^2}\big(w^+_N *\rho\big)\rho \geq \iint_{\R^2\times\R^2} \Bigg(\bigg(1 - \frac{I_w}{(1+\varepsilon)c_{\mathrm{TF}}}\bigg)|p|^2+V(x)\Bigg)\d \overline{\mu} =: \mathcal{E}^{\varepsilon}\big[\overline{\mu}\big].
		\end{equation}
		For $\varepsilon$ small enough, we have
		\begin{equation}
			1 - \frac{I_w}{(1+\varepsilon)c_{\mathrm{TF}}} > 0,
		\end{equation}
		and then $\mathcal{E}^{\varepsilon}$ is a positive and weakly lower semi-continuous functional. Hence, by Lemma \ref{lemma_properties_qn} and Lemma \ref{lemma_fatou_bizarre}, we have
		\begin{equation}
			\liminf_{N\to+\infty} \int_{\mathcal{P}(\R^2)} \mathcal{E}^{\varepsilon}\big[\overline{\mu}\big]\d Q_N(\mu) \geq \int_{\mathcal{P}(\R^2)}\mathcal{E}^{\varepsilon}[\mu]\d P^{\mathrm{DF}}(\mu).
		\end{equation}
		Let $(\mu_N)\in (\Lambda_A)^{\N}$ converging weakly as measures toward $\mu\in \Lambda_A$, with the notation \eqref{equation_def_lambdaa}. Moreover, let
		\begin{equation}
			\rho_N = \rho_{(2\pi)^2\overline{\mu_N}}.
		\end{equation} 
		Then, by the definition of $\Lambda_A$ and \eqref{equation_lien_deux_energies_cinetiques_epsilon}, $(\rho_N)$ is bounded in $L^2(\R^2)$, and hence converges weakly up to extraction. Since we also have convergence of $\overline{\mu_N}$ to $\mu$ -- because of Corollary \ref{corollary_concergence_averaged_sequence} -- $(\rho_N)$ converges\footnote{Without extraction.} weakly in $L^2(\R^2)$ to $\rho_{(2\pi)^2\mu}$ (test against $f\in C^{\infty}(\R^2, \R)$). Thus, as for \eqref{equation_liminf_d2_pas_immediat}, we have, using Fatou's lemma for weakly converging probabilities (Lemma \ref{lemma_fatou_bizarre}) and \eqref{convergence_QN_LambdaA},
		\begin{align}\label{equation_liminf_interaction_repulsive}
			 \liminf_{N\to + \infty} \int_{\Lambda_A} \frac{1}{2}\left(\int_{\R^2}\big(w^+_N*\rho_{(2\pi)^2\overline{\mu}}\big)\rho_{(2\pi)^2\overline{\mu}}\right) \d Q_N(\mu) &\geq \int_{\Lambda_A}\underset{\substack{N\to + \infty\\ \mu_N\rightharpoonup\mu}}{\liminf}\frac{1}{2}\left(\int_{\R^2}\big(w^+_N*\rho_{(2\pi)^2\overline{\mu_N}}\big)\rho_{(2\pi)^2\overline{\mu_N}}\right) \d Q_N(\mu)\notag\\
			 &\geq \int_{\Lambda_A}I_{w^+}\left(\int_{\R^2}\big(\rho_{(2\pi)^2\mu}\big)^2\right) \d P^{\mathrm{DF}}(\mu).
		\end{align}
		By Markov's inequality, we know that
		\begin{equation}
			P^{\mathrm{DF}}(\Lambda_A) \to 1~~\mathrm{when}~~ A\to + \infty, 
		\end{equation}
		and hence
		\begin{equation}
			\liminf_{N\to + \infty} \int_{\mathcal{P}(\R^2)} \frac{1}{2}\left(\int_{\R^2}\big(w^+_N*\rho_{(2\pi)^2\overline{\mu}}\big)\rho_{(2\pi)^2\overline{\mu}}\right) \d Q_N(\mu) \geq \int_{\mathcal{P}(\R^2)}I_{w^+}\left(\int_{\R^2}\big(\rho_{(2\pi)^2\mu}\big)^2\right) \d P^{\mathrm{DF}}(\mu).
		\end{equation}	
		Then, we can take the limit $\varepsilon \to 0$ to find
		\begin{equation}
			E^{\mathrm{TF}}\geq \int_{\mathcal{P}(\R^2)} \Bigg\{\iint_{\R^d\times\R^d} \Bigg(\bigg(1 - \frac{I_{w^-}}{c_{\mathrm{TF}}}\bigg)|p|^2+V(x)\Bigg)\d \mu + I_{w^+}\left(\int_{\R^2}\big(\rho_{(2\pi)^2\mu}\big)^2\right) \Bigg\} \d P^{\mathrm{DF}}(\mu).
		\end{equation}
		Since $P^{\mathrm{DF}}$ only charges probability measures that satisfy the Pauli principle \eqref{equation_pauli_principle_P}, we have
		\begin{multline}\label{equation_d_2_a_la_fin2}
			E^{\mathrm{TF}}\geq\int_{\mathcal{P}(\R^2)} \Bigg\{\iint_{\R^d\times\R^d} \Bigg(\bigg(1 - \frac{I_{w^-}}{c_{\mathrm{TF}}}\bigg)|p|^2+V(x)\Bigg)\d \mu + I_{w^+}\left(\int_{\R^2}\big(\rho_{(2\pi)^2\mu}\big)^2\right) \Bigg\} \d P^{\mathrm{DF}}(\mu)\\ \geq \int_{\mathcal{P}(\R^2)}\mathcal{E}^{\mathrm{TF}}\big[(2\pi)^2\rho_{\mu}\big]\geq E^{\mathrm{TF}}.
		\end{multline}
		Thus, we have equality in all the inequalities of \eqref{equation_d_2_a_la_fin2}. In particular, $P^{\mathrm{DF}}$ only charges the measure $\mu$ such that $(2\pi)^d\rho_{\mu}$ is the unique minimizer of $\mathcal{E}^{\mathrm{TF}}$, and such that 
		\begin{equation}
			\mu(x, p) = (2\pi)^{-2}\mathds{1}_{|p|\leq \sqrt{4\pi\rho(x)}}
		\end{equation}
		with $\rho$ this minimizer of $\mathcal{E}^{\mathrm{TF}}$. Then $(2\pi)^2\mu$ is the minimizer of the Vlasov energy $\mathcal{E}^{\mathrm{V}}$.
	\end{proof}
	
	\begin{lemma}[The Diaconis-Freedman measure charges minimizers of the Vlasov energy, $d=1$]\label{lemma_limit_df_fatou_1_2}
		Let $d=1$. Let $V$ and $w$ satisfy Assumptions \ref{assumption_V}, \ref{assumption_V2} and \ref{assumption_w2}. The limit Diaconis-Freedman measure $P^{\mathrm{DF}}$ only charges probability measures $\mu$ such that $(2\pi) \mu$ is a minimizer of the Vlasov energy $\mathcal{E}^{\mathrm{V}}$.
	\end{lemma}
	We use relaxed functionals to recover some weak lower semi-continuity (see Definition \ref{definition_decomposition_fof_the_energy}, with $w$ replaced by $w^-$, for the notation used in the proof). Indeed, the 1D Vlasov functional is not weakly lower semi-continuous itself. See Appendix \ref{section_appendixe_2}.
	\begin{proof}[Proof of Lemma \ref{lemma_limit_df_fatou_1_2}]
		Let $\varepsilon > 0$ and $\mu \in \Gamma_{\varepsilon}^c$, we have 
		\begin{equation}
			(1+\varepsilon)^2 \iint_{\R\times\R}|p|^2 \d \overline{\mu} \geq c_{\mathrm{TF}}\int_{\R}\rho_{2\pi\overline{\mu}}^3
		\end{equation}
		and hence for $\eta\in(0,1)$ fixed, $\eta = 1/2$ for instance,
		\begin{equation}\label{equation_borne_inf_en_separant_moitie_densite_moitie_husimi2}
			\mathcal{E}_{N,\infty}[\overline{\mu}] \geq \mathcal{E}_N^{\varepsilon, \eta}[\overline{\mu}] + \mathcal{J}_{\eta}[\rho_{2\pi\overline{\mu}}] - \alpha
		\end{equation}
		with
		\begin{equation}
			\mathcal{E}_N^{\varepsilon, \eta}[\overline{\mu}] = \iint_{\R\times\R} \bigg((1-3\varepsilon)\eta |p|^2 + V(x)\bigg)\d \overline{\mu} + \frac{1}{2}\int_{\R}\big(w_N^+*\rho_{2\pi\overline{\mu}}\big)\rho_{2\pi\overline{\mu}}.
		\end{equation}
		Therefore, by Lemma \ref{lemma_properties_qn}, 
		\begin{equation}
			E^{\mathrm{TF}} \geq \liminf_{N\to+\infty} \int_{\mathcal{P}(\R)} \Big(\mathcal{E}_N^{\varepsilon, \eta}[\overline{\mu}] + \mathcal{J}_{\eta}[\rho_{2\pi\overline{\mu}}] - \alpha\Big) \d Q_N(\mu).
		\end{equation}
		Let 
		\begin{equation}\label{equation_def_lambdaaa}
			\Lambda_A = \big\{\mu\in \mathcal{P}(\R),~\mathcal{E}_{\infty}[\mu] \leq A\big\},
		\end{equation}
		by Markov inequality, we have 
		\begin{equation}
			P_N^{\mathrm{DF}}(\Lambda_A^c) \leq \frac{E^{\mathrm{TF}}}{A} + o(1).
		\end{equation}
		Moreover, as in the previous proof, we have \eqref{equation_liminf_interaction_repulsive}. Let
		\begin{equation}
			\mathcal{E}^{\varepsilon, \eta}[\overline{\mu}] = \iint_{\R\times\R} \bigg((1-3\varepsilon)\eta |p|^2 + V(x)\bigg)\d \overline{\mu} + I_{w^+}\int_{\R}\big(\rho_{2\pi\overline{\mu}}\big)^2.
		\end{equation}
		Then, using Lemma \ref{lemma_fatou_bizarre} and Lemma \ref{lemme_semi_conitnuite_inferieure_kcal}, we find
		\begin{equation}
			E^{\mathrm{TF}} \geq \liminf_{N\to+\infty} \int_{\Lambda_A} \Big(\mathcal{E}_N^{\varepsilon, \eta}[\overline{\mu}] + \mathcal{J}_{\eta}[\rho_{2\pi\overline{\mu}}]-\alpha\Big) \d Q_N(\mu) \geq \int_{\Lambda_A} \Big(\mathcal{E}^{\varepsilon, \eta}[\mu] + \mathcal{J}_{\eta}[\rho_{2\pi\mu}]-\alpha\Big) \d P^{\mathrm{DF}}(\mu).
		\end{equation}
		Hence, we can take the limit $\varepsilon \to 0$ and $A\to + \infty$ to get
		\begin{equation}\label{equation_quasiment_laè_derniere2}
			E^{\mathrm{TF}} \geq \int_{\mathcal{P}(\R)} \Big(\mathcal{E}^{0, \eta}[\mu] + \mathcal{J}_{\eta}[\rho_{2\pi\mu}]-\alpha\Big) \d P^{\mathrm{DF}}(\mu) \geq \int_{\mathcal{P}(\R)} \mathcal{E}^{\mathrm{TF}}_{J}\big[\rho_{2\pi\mu}]\d P^{\mathrm{DF}}(\mu) \geq E^{\mathrm{TF}}_{J},
		\end{equation}
		where we used the notation introduced in \eqref{equation_relaxed_eneergy_mixed}. By a direct adaptation of Corollary \ref{corollary_same_minimizers}, we know that 
		\begin{equation}
			E^{\mathrm{TF}}_{J} = E^{\mathrm{TF}},
		\end{equation}
		and hence  $P^{\mathrm{DF}}$ only charges measures $\mu$ such that $\rho_{2\pi\mu}$ is a minimizer of $\mathcal{E}^{\mathrm{TF}}_{J}$, that is a minimizer of $\mathcal{E}^{\mathrm{TF}}$, once again by Corollary \ref{corollary_same_minimizers}. Furthermore, there is equality in the inequalities of \eqref{equation_quasiment_laè_derniere2}, and in particular
		\begin{equation}\label{equation_ohlala_lenergie_cest_la_meme_en_vlasov_ou_pas2}
			\int_{\mathcal{P}}\iint_{\R\times\R} |p|^2 \d \mu(x, p)\d P^{\mathrm{DF}}(\mu) = \int_{\mathcal{P}}c_{\mathrm{TF}}\int_{\R} \rho_{2\pi\mu}^3\d P^{\mathrm{DF}}(\mu).
		\end{equation}
		Since $P^{\mathrm{DF}}$ only charges measures $\mu$ that satisfy the Pauli principle \eqref{equation_pauli_principle_P}, \eqref{equation_ohlala_lenergie_cest_la_meme_en_vlasov_ou_pas2} implies that $P^{\mathrm{DF}}$ only charges probability measures of the form
		\begin{equation}
			\mu(x, p) = \frac{1}{2\pi}\mathds{1}_{|p|\leq \pi \rho(x)}
		\end{equation}
		with $\rho$ a minimizer of $\mathcal{E}^{\mathrm{TF}}$, that is probability measures such that $2\pi\mu$ is a minimizer of $\mathcal{E}^{\mathrm{V}}$.
	\end{proof}
	
	\appendix
	\section{Appendix}
	\subsection{Existence of Thomas-Fermi minimizers when $d=1$}\label{section_appendixe_1}
	In this appendix, we prove Theorem \ref{theorem3}. To do so, we prove the existence of minimizers for a relaxed energy, and we then show that these minimizers are actually minimizers of the Thomas-Fermi energy. In this section, we use Assumption \ref{assumption_V2}, i.e. we assume that the trapping potential has no flat spots.
	\begin{definition}[Local energy]\label{definition_local_energy_appendix}
		For $\alpha\in \R_+$, let
		\begin{equation}
			e_{\alpha}:x\in \R_+ \mapsto c_{\mathrm{TF}}x^3 - I_w x^2 + \alpha x
		\end{equation}
		with 
		\begin{equation}
			\alpha = \frac{I_w^2}{4c_{\mathrm{TF}}}.
		\end{equation}
	\end{definition}
	We add $\alpha x$ here to have a positive function. Since we will then use it with $\rho$ of integral $1$, we can add freely the integral of $\rho$ without changing the minimizers.
	\begin{lemma}[Minimum of $e_{\alpha}$]\label{lemma_min_ealpha}
		We have
		\begin{equation}\label{equation_minimiseur_e_alpha}
			\min_{x\in\R_+} e_{\alpha}(x) = e_{\alpha}(0) = e_{\alpha}(\rho_{\alpha})=0
		\end{equation}
		with
		\begin{equation}
			\rho_{\alpha} = \frac{I_w}{2c_{\mathrm{TF}}}.
		\end{equation}
	\end{lemma}
	\begin{proof} 
		With such a choice of $\alpha$, it is clear that 
		\begin{equation}
			e_{\alpha}(\rho_{\alpha}) = c_{\mathrm{TF}}\left(\frac{I_w}{2c_{\mathrm{TF}}}\right)^3 - I_w \left(\frac{I_w}{2c_{\mathrm{TF}}}\right)^2 +  \frac{I_w^2}{4c_{\mathrm{TF}}}\frac{I_w}{2c_{\mathrm{TF}}} = 0 = e_{\alpha}(0).
		\end{equation}
		Moreover, we have
		\begin{equation}
			e_{\alpha}'(\rho_{\alpha}) = 3c_{\mathrm{TF}}\left(\frac{I_w}{2c_{\mathrm{TF}}}\right)^2 - 2I_w \frac{I_w}{2c_{\mathrm{TF}}} + \frac{I_w^2}{4c_{\mathrm{TF}}} = 0.
		\end{equation}
		Then, since $c_{\mathrm{TF}}>0$, we have indeed \eqref{equation_minimiseur_e_alpha}.
	\end{proof}
	\begin{definition}[Relaxed energy]\label{definition_relaxed_energy_appendice_la}
		Let $J$ be the local relaxed energy defined by
		\begin{equation}\label{equation_def_J}
			J:x\in \R_+ \mapsto e_{\alpha}(x)\mathds{1}_{x\geq \rho_{\alpha}}.
		\end{equation}
		We set the relaxed energy
		\begin{equation}
			\mathcal{E}^{\mathrm{TF}}_J[\rho] := \mathcal{J}[\rho] + \int_{\R^d}V\rho - \alpha\int_{\R^d}\rho := \int_{\R^d}J(\rho) + \int_{\R^d}V\rho - \alpha\int_{\R^d}\rho
		\end{equation}
		and its minimum
		\begin{equation}\label{equation_definition_etfj}
			E^{\mathrm{TF}}_J = \inf\bigg\{\mathcal{E}^{\mathrm{TF}}_J[\rho],~\rho\geq 0,~\int_{\R^d}\rho = 1\bigg\}.
		\end{equation}
	\end{definition}
	Let us state first some properties of this relaxed energy.
	\begin{properties}[On the relaxed energy]\label{properties_jcal}
		The relaxed energy is bounded from above by the Thomas-Fermi energy: for any $\rho \geq 0$ of integral $1$, we have
		\begin{equation}
			\mathcal{E}^{\mathrm{TF}}_J[\rho] \leq \mathcal{E}^{\mathrm{TF}}[\rho].
		\end{equation}
		Moreover, the functional $\mathcal{J}$ is weakly lower semi-continuous on the sets of positive $L^3$ functions of integral 1.
	\end{properties}
	\begin{proof}
		First, since $e_{\alpha}$ is positive on $\R_+$, we have for $\rho\geq 0$ of integral 1
		\begin{equation}
			\mathcal{E}^{\mathrm{TF}}_J[\rho] = \int_{\R^d}J(\rho) + \int_{\R^d}V\rho - \alpha\int_{\R^d}\rho \leq 
			\int_{\R^d}e_{\alpha}(\rho) + \int_{\R^d}V\rho - \alpha\int_{\R^d}\rho = \mathcal{E}^{\mathrm{TF}}[\rho].
		\end{equation}
		Moreover, let us write
		\begin{equation}
			S = \left\{\rho\in L^3(\R),~\rho\geq 0,~\int_{\R} = 1\right\}.
		\end{equation}
		The set $S$ is clealy convex, and since $J$ is convex on $\R_+$, $\mathcal{J}$ is convex on $S$. Furthermore, $\mathcal{J}$ is continuous for the $L^3$ norm on $S$ because $J$ is continuous and for $\rho, \sigma\in S$
		\begin{equation}
			\int_{\R} |\rho-\sigma|^2 \leq \sqrt{\int_{\R} |\rho-\sigma|}\sqrt{\int_{\R}|\rho-\sigma|^3} \leq \sqrt{2}\sqrt{\int_{\R}|\rho-\sigma|^3}.
		\end{equation}
		Hence, $\mathcal{J}$ is weakly lower semi-continuous (see \cite[Corollary 3.9]{brezis_functional_2011} for instance). 
	\end{proof}
	\begin{proposition}[Minimization problem for the relaxed energy]\label{proposition_minimizing_problem_relaxed}
		Problem \eqref{equation_definition_etfj} admits at least one minimizer. All the minimizers $\rho$ satisfy
		\begin{equation}\label{equation_espece_euler_lagrange}
			\begin{cases}
				e_{\alpha}'(\rho)\mathds{1}_{\rho\geq\rho_{\alpha}} + V - \alpha = \lambda& \mathrm{ on }~~ \supp (\rho)\\
				V - \alpha \geq \lambda& \mathrm{ on }~~ \supp (\rho)^c.
			\end{cases}
		\end{equation}
	\end{proposition}
	\begin{proof}
		Let $(\rho_n)$ be a minimizing sequence of $\mathcal{E}_J^{\mathrm{TF}}$. Then for $n$ large,
		\begin{equation}
			E_J^{\mathrm{TF}}+1 +\alpha \geq \mathcal{E}_J^{\mathrm{TF}}[\rho_n] + \alpha = \mathcal{J}[\rho_n] + \int V\rho_n.
		\end{equation}
		Thus, $V\rho_n$ is bounded in $L^1$ and hence $(\rho_n)$ is a tight sequence of measures. Moreover, $\mathcal{J}[\rho_n]$ is bounded, which implies
		\begin{equation}
			\int_{\R} \rho^2 \leq \frac{c_{\mathrm{TF}}}{2I_w}\int_{\R} \rho^3 + \frac{2I}{c_{\mathrm{TF}}}\int_{\R}\rho = \frac{c_{\mathrm{TF}}}{2I_w}\int_{\R} \rho^3 + \frac{2I_w}{c_{\mathrm{TF}}},
		\end{equation}
		hence
		\begin{equation}
			E_J^{\mathrm{TF}}+1 +\alpha \geq \mathcal{J}[\rho_n] \geq c_{\mathrm{TF}}\int_{\rho_n\geq\rho_{\alpha}}\rho_n^3 - \frac{c_{\mathrm{TF}}}{2}\int_{\R}\rho_n^3 - 2\frac{I_w^2}{c_{\mathrm{TF}}}.
		\end{equation}
		Therefore
		\begin{equation}
			\int_{\R}\rho_n^3 \leq \int_{\rho_n \geq \rho_{\alpha}} \rho_n^3 + \rho_{\alpha}^2\int_{\rho_n \leq \rho_{\alpha}}  \rho_n \leq \int_{\rho_n \geq \rho_{\alpha}} \rho_n^3 + \rho_{\alpha}^2\leq C +  \frac{1}{2}\int_{\R}\rho_n^3,
		\end{equation}
		and thus $(\rho_n)$ is bounded in $L^3$. Therefore, up to extraction, we have the convergence of $(\rho_n)$ weakly in $L^3$ and weakly-* as measures towards a $\rho_{\infty}\in L^3$. Since $(\rho_n)$ is tight, we have 
		\begin{equation}
			\int_{\R}\rho_{\infty} = \lim_{n\to+\infty}\int_{\R} \rho_n = 1.
		\end{equation}
		Furthermore, by weak lower semi-continuity of $\mathcal{J}$, we have
		\begin{equation}
			E^{\mathrm{TF}}_J = \liminf_{n\to+\infty} \mathcal{E}^{\mathrm{TF}}_J[\rho_n] \geq \mathcal{E}^{\mathrm{TF}}_J[\rho_{\infty}],
		\end{equation}
		and $\rho_{\infty}$ is indeed a minimizer. Then, \eqref{equation_espece_euler_lagrange} corresponds to the Euler-Lagrange equations of the minimization problem. To prove the first equation, one can take $\sigma\in L^1\cap L^3$ such that
		\begin{equation}
			\int_{\R} \sigma = 0~~~~\mathrm{and}~~~~|\sigma|\leq \rho_{\infty}
		\end{equation}
		and consider for $|\varepsilon| < 1$ the energy of $\rho_{\infty} + \varepsilon \sigma$ in the limit $\varepsilon \to 0$. To prove the second equation, one can take $\sigma\in L^1\cap L^3$ such that
		\begin{equation}
			\int_{\R} \sigma = 0,~~~~~~\sigma \geq 0 ~~\mathrm{on}~~\supp(\rho_{\infty})^c ~~~~\mathrm{and}~~~~\sigma \leq -\rho_{\infty} ~~\mathrm{on}~~\supp(\rho_{\infty}),
		\end{equation}
		and consider the energy of $\rho_{\infty} + \varepsilon \sigma$ in the limit $\varepsilon\to 0^+$.
	\end{proof}
	\begin{corollary}[The relaxation does not change minimizers]\label{corollary_same_minimizers}
		We have
		\begin{equation}\label{equation_ohlala_enfait_les_energies_cest_les_memes}
			E^{\mathrm{TF}} = E^{\mathrm{TF}}_J
		\end{equation} 
		and the minimizers of $\mathcal{E}^{\mathrm{TF}}_J$ and $\mathcal{E}^{\mathrm{TF}}$ on $S$ coincide. 
	\end{corollary}
	\begin{proof}
		Let $\rho$ be a minimizer of $\mathcal{E}^{\mathrm{TF}}_J$. Then, by \eqref{equation_espece_euler_lagrange}, we have
		\begin{equation}
			V = \lambda - \alpha ~~\mathrm{on}~~ \{0<\rho<\rho_{\alpha}\},
		\end{equation}
		and since we have assumed that the level sets of $V$ are of measure 0, we have
		\begin{equation}\label{equation_minimiseurs_bornés_inférieurement}
			\rho\geq \rho_{\alpha}
		\end{equation}
		almost surely on $\supp(\rho)$. This implies in particular that
		\begin{equation}
			e_{\alpha}(\rho) = J(\rho)
		\end{equation}
		almost everywhere, and thus
		\begin{equation}
			E^{\mathrm{TF}}\geq E^{\mathrm{TF}}_J = \mathcal{E}^{\mathrm{TF}}_J[\rho] = \mathcal{E}^{\mathrm{TF}}[\rho].
		\end{equation}
		Therefore, $\rho$ is indeed a minimizer of $\mathcal{E}^{\mathrm{TF}}$, and \eqref{equation_ohlala_enfait_les_energies_cest_les_memes} is true. On the contrary, if $\rho$ is a minimizer of $\mathcal{E}^{\mathrm{TF}}$, we have
		\begin{equation}\label{equation_ou_en_fait_les_inegalites_se_trouvent_etre_des_egalites}
			E^{\mathrm{TF}} = \mathcal{E}^{\mathrm{TF}}[\rho] \geq \mathcal{E}^{\mathrm{TF}}_J[\rho] \geq E^{\mathrm{TF}}_J = E^{\mathrm{TF}}.
		\end{equation}
		Thus, we have equality in all the inequalities of \eqref{equation_ou_en_fait_les_inegalites_se_trouvent_etre_des_egalites}, and in particular
		\begin{equation}
			\mathcal{E}^{\mathrm{TF}}_J[\rho] = E^{\mathrm{TF}}_J,
		\end{equation}
		i.e. $\rho$ is a minimizer of $\mathcal{E}^{\mathrm{TF}}_J$.
	\end{proof}
	Theorem \ref{theorem3} is now just a consequence of Proposition \ref{proposition_minimizing_problem_relaxed}, Corollary \ref{corollary_same_minimizers}, \eqref{equation_minimiseurs_bornés_inférieurement} and Remark \ref{remark_link_vlasov_thomas_fermi}.
	
	\subsection{Relaxed Vlasov energy}\label{section_appendixe_2}
	Similarly as in the previous section, we relax the Vlasov functional to recover some weak lower semi-continuity.
	\begin{definition}[Decomposition of the energy]\label{definition_decomposition_fof_the_energy}
		For $\eta \in [0, 1]$, let  
		\begin{equation}
			e_{\eta}(x) = (1-\eta) c_{\mathrm{TF}}x^3 - I_w x^2 + \alpha x
		\end{equation}
		with $\alpha > 0$ such that $e_{\eta}$ has exactly two minimizers on $\R_+$, $0$ and $\rho_{\alpha}\in \R_+^*$, with minimum $0$ (see Definition \ref{definition_local_energy_appendix} and Lemma \ref{lemma_min_ealpha} for more detail). We set
		\begin{equation}
			J_{\eta}(x) = \mathds{1}_{x\geq\rho_{\alpha}}e_{\eta}(x) \leq e_{\eta}(x)~~~~~\mathrm{and}~~~~~\mathcal{J}_{\eta}[\rho] = \int_{\R^d} J_{\eta}(\rho) - \alpha\int_{\R}\rho.
		\end{equation}
		Moreover, let 
		\begin{equation}
			\mathcal{E}^{\varepsilon, \eta}[\overline{\mu}] = \iint_{\R^d\times\R^d} \bigg((1-3\varepsilon)\eta |p|^2 + V(x)\bigg)\d \overline{\mu}
		\end{equation}
		and
		\begin{equation}
			\mathcal{K}^{\varepsilon, \eta}[\mu] = \mathcal{E}^{\varepsilon, \eta}[\mu] + \mathcal{J}_{\eta}[\rho_{2\pi\mu}] - \alpha.
		\end{equation}
	\end{definition}
	We recall that the $\varepsilon$ in the energy functional comes from \eqref{equation_borne_inf_en_separant_moitie_densite_moitie_husimi}.
	\begin{lemma}[Weak lower semi-continuity of the relaxed energy]\label{lemme_semi_conitnuite_inferieure_kcal}
		Let
		\begin{equation} 
			\mu \in \Lambda_A = \big\{\mu\in \mathcal{P}(\R),~\mathcal{E}_{\infty}[\mu] \leq A\big\}
		\end{equation} 
		a probability measure satisfying the Pauli principle, i.e. 
		\begin{equation}
			\mu \leq (2\pi)^{-d}.
		\end{equation} 
		Moreover, let $(\mu_n) \in (\Lambda_A)^{\N}$ a sequence of probability measures that converge weakly to $\mu$, and $V$ and $w$ satisfying Assumptions \ref{assumption_V}, \ref{assumption_V2} and \ref{assumption_w}. Then, 
		\begin{equation}
			\liminf_{n\to+\infty} \mathcal{K}^{\varepsilon, \eta}[\mu_n] \geq \mathcal{K}^{\varepsilon, \eta}[\mu].
		\end{equation}
	\end{lemma}
	\begin{proof}
		First, it is clear that $\mathcal{E}^{\varepsilon, \eta}$ is weakly lower semi-continuous. Let us write
		\begin{equation}
			\rho_n = \rho_{2\pi\mu_n} = \int_{\R^d} \d \mu_n(\cdot, p).
		\end{equation}
		Then, since $\mu_n\in \Lambda_A$, we have
		\begin{equation}
			A + \alpha \geq \mathcal{J}_{\eta}[\rho_n]
		\end{equation}
		and thus $(\rho_n)$ is bounded in $L^3$. Hence, the sequence converges weakly (up to a subsequence) to a $\rho_{\infty}\in L^3$, and since $\mathcal{J}_{\eta}$ is convex and strongly continuous for the $L^3$ norm, we have
		\begin{equation}\label{equation_liminf_jcal}
			\liminf_{n\to+\infty} \mathcal{J}_{\eta}[\rho_n] \geq \mathcal{J}_{\eta}[\rho_{\infty}].
		\end{equation}
		As we also have $\mu_n\rightharpoonup \mu$, we get 
		\begin{equation} \label{equation_ohlala_rho_en_fait_cest_le_bon}
			\rho_{\infty} = \rho_{(2\pi)^d\mu}= \int_{\R^d} \d \mu(\cdot, p).
		\end{equation}
	\end{proof}
	
	\subsection{Additionnal material}\label{subsection_appendice_bonus}
	
	\begin{proof}[Proof of Remark \ref{remark_constraint_dimension}]
		Let us assume $d\geq 3$ and take
		\begin{equation}
			\rho_n = n^d\rho(n\cdot)
		\end{equation}
		with a positive $\rho \in C^{\infty}_c$ of integral $1$. Then, the potential energy is bounded:
		\begin{equation}
			\int_{\R^d} V\rho_n \leq C,
		\end{equation}
		while the kinetic and interaction energies diverge:
		\begin{equation}
			\int_{\R^d} \rho_n^{1+2/d} = n^2\int_{\R^d} \rho^{1+2/d},~~~~~\int_{\R^d} \rho_n^2 = n^d\int_{\R^d} \rho^2.
		\end{equation}
		As $d > 2$, we have 
		\begin{equation}
			\mathcal{E}^{\mathrm{TF}}[\rho_n] \to -\infty.
		\end{equation}
		For $d = 2$, we have
		\begin{equation}
			\mathcal{E}^{\mathrm{TF}}[\rho] = \big(c_{\mathrm{TF}} - I_w\big)\int_{\R^d} \rho^2 + \int_{\R^d} V\rho,
		\end{equation}
		and therefore, as long as \eqref{equation_lien_c_I} is satisfied, the Thomas-Fermi functional is bounded from below. When $d = 1$, we have for all $\alpha > 0$
		\begin{equation}
			\rho^2 \leq  \frac{\alpha}{2}\rho^3 + \frac{\alpha^{-1}}{2}\rho
		\end{equation}
		by Young's inequality. Thus, for a $\rho$ of integral $1$,
		\begin{equation}
			\mathcal{E}^{\mathrm{TF}}[\rho] \geq \frac{c_{\mathrm{TF}}}{2}\int_{\R^d} \rho^3 - \frac{I_w^2}{2c_{\mathrm{TF}}} + \int_{\R^d} V\rho\geq  - \frac{I_w^2}{2c_{\mathrm{TF}}}.
		\end{equation}
	\end{proof}
	
	\section*{Acknowledgements}
	I thank Nicolas Rougerie for the many long and fruitful discussions we had, especially on the existence of Thomas-Fermi minimizers, and for his proofreading of this manuscript.
	
	\bibliographystyle{abbrv}
	\bibliography{biblio}
	
\end{document}